\documentclass[acmsmall, nonacm, screen,svgnames,x11names]{acmart}

\AtBeginDocument{%
	}

\setcopyright{rightsretained}
\acmDOI{10.1145/3798215}
\acmYear{2026}
\acmJournal{PACMPL}
\acmVolume{10}
\acmNumber{OOPSLA1}
\acmArticle{107}
\acmMonth{4}
\received{2025-10-06}
\received[accepted]{2026-02-17}

\usepackage{csquotes}
\usepackage[disable]{todonotes}
\usepackage{xcolor}
\usepackage{tabu}
\usepackage{adjustbox}
\usepackage{mathtools}
\usepackage{xspace}
\usepackage{proof} 
\usepackage{semantic}
\usepackage{listings}
\usepackage{stmaryrd}
\usepackage{wrapfig}
\usepackage{tikz}
\usepackage[nameinlink, noabbrev]{cleveref}
\usepackage{caption}
\usepackage{subcaption}
\usepackage{array}
\usepackage[utf8]{inputenc}

\definecolor{codegreen}{rgb}{0,0.6,0}
\definecolor{codegray}{rgb}{0.5,0.5,0.5}
\definecolor{codepurple}{rgb}{0.58,0,0.82}
\definecolor{backcolour}{rgb}{0.95,0.95,0.92}

\lstdefinestyle{mystyle}{
	commentstyle=\color{codegreen},
	keywordstyle=\color{magenta},
	numberstyle=\tiny\color{codegray},
	stringstyle=\color{codepurple},
	basicstyle=\ttfamily\footnotesize,
	breakatwhitespace=false,         
	breaklines=true,                 
	captionpos=b,                    
	keepspaces=true,                 
	numbers=left,                    
	numbersep=5pt,                  
	showspaces=false,                
	showstringspaces=false,
	showtabs=false,                  
	tabsize=2
}
\DeclareMathAlphabet{\mathpzc}{OT1}{pzc}{m}{it}

\newcommand{\sfsymbol}[1]{\textsf{\upshape {#1}}}

\newcommand{\morespace}[1]{~{}#1{}~}
\newcommand{\qmorespace}[1]{\quad{}#1{}\quad}

\newcommand{\ttriangleq}{~{}\triangleq{}~}

\newcommand{\iin}{~{}\in{}~}

\newcommand{\IFSYMBOL}{\ensuremath{\textnormal{\texttt{if}}}}

\newcommand{\ELSESYMBOL}{\ensuremath{\textnormal{\texttt{else}}}}

\newcommand{\ITE}[3]{\ensuremath{\IFSYMBOL\;{#1}\;\texttt{then}\; #2 \; \ELSESYMBOL\; #3  }}

\newcommand{\observesymbol}{\textnormal{\texttt{observe}}}
\newcommand{\OBSERVE}[1]{\ensuremath{\observesymbol (#1)}}

\newcommand{\statesubst}[2]{\left[ {#1} \mapsto {#2}\right]}

\newcommand{\monus}{
	\mathbin{
		\vphantom{+}
		\text{
			\mathsurround=0pt 
			\ooalign{
				\noalign{\kern-.5ex}
				\hidewidth$\smash{\cdot}$\hidewidth\cr 
				\noalign{\kern.5ex}
				$-$\cr 
			}%
		}%
	}%
}

\newcounter{blacktrianglelefteq}

\newcounter{leftsliceeq}

\definecolor{webgreen}{rgb}{0,.5,0}

\definecolor{mydarkorange}{RGB}{220,80,0}
\colorlet{cemphcolor}{mydarkorange}

\newcounter{computationarrowsone}
\newcounter{computationarrowstwo}

\newcounter{sarrow}

\makeatletter
\newlength{\negph@wd}
\DeclareRobustCommand{\negphantom}[1]{%
  \ifmmode
    \mathpalette\negph@math{#1}%
  \else
    \negph@do{#1}%
  \fi
}
\newcommand{\negph@math}[2]{\negph@do{$\m@th#1#2$}}
\newcommand{\negph@do}[1]{%
  \settowidth{\negph@wd}{#1}%
  \hspace*{-\negph@wd}%
}
\makeatother

\newcommand{\listtraverse}[1]{C}

\newcommand{\BOOL}{\ensuremath{\texttt{Bool}}}
\newcommand{\VAL}{\ensuremath{\mathit{Val}}}
\newcommand{\VARS}{\ensuremath{\mathit{Var}}}
\newcommand{\TRUE}{\ensuremath{\texttt{T}}}
\newcommand{\FALSE}{\ensuremath{\texttt{F}}}

\newcommand{\AEXP}{\ensuremath{\mathit{aexp}}}
\newcommand{\FST}{\ensuremath{\texttt{fst}\,}}
\newcommand{\SND}{\ensuremath{\texttt{snd}\,}}
\newcommand{\LET}[2]{\ensuremath{\texttt{let}\; #1 \;\texttt{in}\; #2}}
\newcommand{\FLIP}[1]{\ensuremath{\texttt{flip} \left(#1\right) }}
\newcommand{\NFLIP}{\ensuremath{\texttt{nflip}\left(\right) }}
\newcommand{\FUNC}[2]{\ensuremath{\texttt{fun} \; #1 \; \left\{\, #2 \,\right\} } }
\newcommand{\FUNCNAME}{\mathit{func}}

\newcommand{\NDICE}{\sfsymbol{\mbox{noDice}}\xspace}
\newcommand{\DICE}{\sfsymbol{\mbox{Dice}}\xspace}

\theoremstyle{acmdefinition}
\newtheorem{remark}{Remark}

\newboolean{showRuleNames}
\setboolean{showRuleNames}{true}

\newboolean{isExtended}
\setboolean{isExtended}{true}

\ifthenelse
{\boolean{isExtended}}
{
	\newcommand{\ext}[2]{#1}
}{
	\newcommand{\ext}[2]{#2}
}

\newcommand{\citeExtended}{\cite{extended}}

\newcommand{\citeBisection}{{\cite{quatmann2025bisection}}}
\newcommand{\citetBisection}{{\citet{quatmann2025bisection}}}

\begin{document}
	
	\title
	[noDice: Inference for Discrete Probabilistic Programs with Nondeterminism and Conditioning]
	{noDice: Inference for Discrete Probabilistic Programs with Nondeterminism and Conditioning}
	
	\author{Tobias Gürtler}
	\orcid{0009-0000-7131-0936}
	\affiliation{%
		\institution{Saarland University}
		\city{}
		\country{Germany}
	}
	\affiliation{%
		\institution{Saarland Informatics Campus}
		\city{}
		\country{Germany}
	}
	\email{guertler@cs.uni-saarland.de}
	
	\author{Benjamin Lucien Kaminski}
	\orcid{0000-0001-5185-2324}
	\affiliation{%
		\institution{Saarland University}
		\city{}
		\country{Germany}
	}
	\affiliation{%
		\institution{University College London}
		\city{}
		\country{United Kingdom}
	}
	\email{kaminski@cs.uni-saarland.de}


\begin{abstract}

Probabilistic programming languages (PPLs) are an expressive and intuitive means of representing complex probability distributions. 
In that realm, languages like \DICE target an important class of probabilistic programs: those whose probability distributions are discrete.
Discrete distributions are common in many fields, including text analysis, network verification, artificial intelligence, and graph analysis.
Another important feature in the world of probabilistic modeling are \emph{nondeterministic choices} as found in Markov Decision Processes (MDPs) which play a major role in reinforcement learning.
Modern PPLs usually lack support for \emph{nondeterminism}.
We address this gap with the introduction of \NDICE, which extends the discrete probabilistic inference engine \DICE.
\NDICE performs inference on loop-free programs by constructing an MDP so that the distributions modeled by the program correspond to  schedulers in the MDP.
Furthermore, decision diagrams are used as an intermediate step to exploit the program structure and drastically reduce the state space of the MDP.
\end{abstract}


	
	\received{20 February 2007}
	\received[revised]{12 March 2009}
	\received[accepted]{5 June 2009}
	
	\maketitle
	

\section{Introduction}

Probabilistic programming languages (PPLs) extend standard programming languages with probabilistic statements to model random behavior and to condition on observed evidence.
These constructs allow probabilistic programs to concisely model complex distributions, and PPLs make probabilistic modeling approachable to developers who might lack the required expertise in probability theory. 
Thus, the main purpose of probabilistic programs is not to execute them, but to perform probabilistic inference, i.e.\ to determine the probability distribution over possible output values, conditioned on the observations in the program.
Over the years, a variety of tools and languages has been proposed to perform probabilistic inference \cite{ gehr2016psi, wang2018pmaf, bingham2019pyro, wood-aistats-2014, randone2024inference}.
However, support for (non-probabilistic) nondeterminism is lacking among modern inference engines, and combining nondeterminism with conditioning in PPLs remains an open problem \cite{gordon2014probabilistic, junges2023scalable}.

Nondeterminism is \textit{"a powerful modeling tool to deal with unknown information, as well as to specify abstractions in situations where details are unimportant"}~\cite{gordon2014probabilistic}.
Alas, combining nondeterminism and conditioning significantly complicates inference:
Probabilistic programs no longer model a single distribution, but a \emph{set} of distributions. 
This renders probabilistic inference no longer compositional, i.e.\ maximum conditional probabilities in a program cannot generally be inferred from the maximum conditional probabilities of its subprograms.
Due to this, work on PPLs extended with nondeterminism and conditioning is sparse, and existing approaches often cannot deal with both simultaneously \cite{olmedo2018conditioning, van2010dtproblog}.
Currently, the only option to combine both features is to apply probabilistic model checkers directly to the operational semantics of programs \cite{jansen2016bounded}, but this approach often leads to very large state spaces.

\subsection{The \NDICE Inference Framework}
We introduce \NDICE, a discrete probabilistic program inference engine which aims to reduce the state space of models with nondeterminism by extending the \DICE inference engine~\cite{holtzen2020scaling}.
\DICE and \NDICE specialize in inference for loop-free programs with \textit{discrete} randomness by forgoing support for unbounded numbers and continuous distributions in favor of more efficient inference.
Hence, \NDICE uses a simple first-order, non-recursive language with discrete probabilistic and nondeterministic choices as well as first-class observations for Bayesian reasoning.
\NDICE provides a novel inference algorithm to determine the maximum probability to return a specific value among all ways to resolve the nondeterminism.
To do so, \NDICE proceeds as follows:

\paragraph{Boolean Compilation}
First, \NDICE uses the compilation engine of \DICE to compile the probabilistic program into a pair of boolean formulas along with an annotated trace of boolean variables to encode its behavior.
Here, one formula reflects the programs behavior while ignoring observations, and the other formula represents solely the observations.

\paragraph{Decision Diagrams}
The previously compiled boolean formulas are represented internally as \textit{Binary Decision Diagrams} (BDDs), a compact graphical encoding of boolean functions.
Now, \NDICE diverges from the \DICE algorithm and combines the two BDDs into a single \textit{Algebraic Decision Diagram} (ADD).
ADDs are a generalization of BDDs which allows more than binary outputs, and this enables the ADD to unify the two boolean formulas into a single decision diagram.

\paragraph{Markov Decision Process}
The final step is to use the previously built ADD as scaffolding for a \textit{Markov Decision Process} (MDP).
MDPs naturally combine nondeterministic and probabilistic behavior, making them a suitable tool for our inference method.
We build upon results from \citet{baier2014computing} and {\citetBisection} on conditional reachability properties in MDPs, which allows us to construct an MDP that aligns with the compiled program.
More precisely, the distributions modeled by the program correspond to the memoryless schedulers in the MDP.
\medskip

We can infer the maximum probability that the original program outputs a particular value as a conditional reachability probability in the constructed MDP. 
Inference performance can heavily depend on the size of the final MDP, hence reducing its state space is often critical for the inference algorithm.
In our work, decision diagrams exploit the program structure to efficiently remove redundant states from the MDPs, which significantly reduces their size.

\subsection{Outline and Contributions}

We provide an overview of the \NDICE language and inference algorithm in \Cref{sec:overview}, followed by the formal semantics of \NDICE in \Cref{sec:semantics}.
Then, \Cref{sec:background,sec:inference} present the theoretical foundation of the inference algorithm, which we implement and evaluate in \Cref{sec:implementation}.
Finally, \Cref{sec:related_work,sec:future} explore related and future work.
Our contributions are summarized as follows:

\begin{itemize}
	\item We develop the probabilistic programming language \NDICE, a non-recursive functional language with probabilistic choices, nondeterministic choices and first-class observations.
	\item We present the theoretical foundation of the \NDICE inference algorithm. We use decision diagrams to compile a compact MDP, and we infer the maximum probability  for a \NDICE program to output a value via conditional reachability probabilities in that MDP.
	\item We implement the inference algorithm and evaluate its efficiency against other approaches on a range of benchmarks. 
	Here, \NDICE significantly reduces the size of the MDPs, and even outperforms state-of-the-art model checkers on some benchmarks.
\end{itemize}


\section{An Overview of the Inference Process}\label{sec:overview}

\subsection{A Tour of the \NDICE Language}
We introduce the \NDICE language with an example adapted from \citet{junges2021runtime}:
A plane is approaching a runway to land, while a moving vehicle is crossing said runway. 
Naturally, the plane should not land while the vehicle is occupying the runway, and so the plane has to track the vehicle with (imprecise) sensors and estimate its position. 
Tracking moving vehicles is a common case study in probabilistic programming \cite{randone2024inference, cheng2025inference, baudart2020reactive}, and \NDICE now allows us to model the behavior of the vehicle as nondeterministic, which is not featured in other PPLs.
In particular, nondeterminism allows us to \textit{under-specify} the vehicle's behavior, as we may not be able to describe the movement of the vehicle exactly, and many (not necessarily probabilistic) factors influence how the vehicle could move. 
Hence, we use nondeterminism as an abstraction to reason about the worst case scenario, where the probability that the vehicle is on the runway when the plane lands is maximal.

We will consider a simple instance of the scenario above where the vehicle may be in three possible positions: Left of the runway (Location $0$), on the runway (Location $1$) or right of the runway (Location $2$).
At each time step, the vehicle may move one location to the right, and we take one measurement of the vehicles location with an imprecise sensor.
After three time steps, we observed the sensor measurements $0$, then $1$, and finally $2$. 
Given these measurements, we wish to determine the maximum probability that the vehicle is still on the runway after three time steps if it started at Location~$0$.
To answer this query, we model the scenario with the \NDICE program in \Cref{fig:plane}, which uses bounded numbers to represent the location of the vehicle:

\begin{wrapfigure}[19]{l}{0.5\textwidth}%
	\vspace*{-1.5\intextsep}
	\begin{center}
		\begin{tabular}{l c}
			&
			\begin{lstlisting}[language=Caml]
fun move(pos: int): int {
	let m = if nflip() then flip(0.75) 
          else flip(0.5) in 
	if m && pos != 2 then pos+1 else pos  
}

fun step(pos: int, obs: int): int {
	let new_pos = move(pos) in
	let mes = if flip(0.9) then new_pos 
            else uniform(0, 3) in
	let o = observe(mes == obs) in
	new_pos   
}

let p1 = step(0, 0) in
let p2 = step(p1, 1) in
let p3 = step(p2, 2) in
p3 == 1
\end{lstlisting}
		\end{tabular}
	\end{center}
	\vspace*{-.5em}%
	\caption{A \NDICE program to track a moving vehicle.}
	\label{fig:plane}
\end{wrapfigure}%
\paragraph{Vehicle Movement}
The behavior of the vehicle has both probabilistic and nondeterministic aspects:
For one, we abstract the movement speed of the vehicle with randomness, where at each time step the vehicle has either a 50\% chance or a 75\% chance to progress to the next location. 
We refer to the former as \textit{slow} and to the latter as \textit{fast}.
Furthermore, the vehicle may \enquote{decide} whether to move slow or fast at each individual time step.
However, the available information may be insufficient to model the decision process of the vehicle as a probabilistic process.
We thus model the decision process to be \textit{nondeterministic}, which allows us to bound the behavior of the vehicle across all possible strategies it could use.
In \Cref{fig:plane}, we model this behavior with the \texttt{move} function in Lines~1-5, which takes the current position of the vehicle as an argument.

In Lines 2 and 3, the variable \texttt{m} models whether the vehicle progresses to the next location. 
Here, the choice between speeds is modeled with the statement $\NFLIP$, which nondeterministically returns either true or false, and the probabilistic behavior is modeled with the statement $\FLIP\theta$, which returns $\TRUE$ with probability $\theta$ and $\FALSE$ otherwise.
Afterwards, we return the updated position of the vehicle in Line~4:
If the vehicle made progress and has not left the runway yet, then the vehicle moves to next position, otherwise it does not move.

This \textit{nondeterministic} model of the vehicle abstracts a wide range of behaviors.
For example, the vehicle could always move fast, or always slow, or it may choose randomly.
It may move fast only in the first time step, or it could pick the speed depending on the current position. 
\NDICE considers all of the aforementioned strategies, along with any combination of them. 
More precisely, our model of nondeterminism considers each context where $\NFLIP$ is called as a separate choice, where each choice can be resolved to an arbitrary probabilistic statement $\FLIP\theta$.
This allows nondeterministic choices to be made depending on previous, probabilistic outcomes and enables us to reason about a wide variety of behaviors.
In the context of \Cref{fig:plane}, this allows the ground vehicle to decide its speed individually for every location and time step.

\paragraph{Sensor Measurements}
We now consider the \texttt{step} function in Lines~7-13, which models one full time step.
To do so, \texttt{step} takes two arguments: The current real position of the vehicle \texttt{pos} and the next expected measurement \texttt{obs}.
The function first calls \texttt{move} in Line~8 to update the vehicle's position, performs the imperfect sensor measurement in Lines~9-11, and finally returns the updated position in Line~12. 
We now consider how the sensor is modeled in more detail.

First of, Lines~9-10 define how sensor measurements are derived from the real position:
With a likelihood of $90\%$ the sensor will return the correct position \texttt{new\_pos}, otherwise it returns a position uniformly at random. 
The core language of \NDICE only supports Boolean variables and choices, hence \texttt{uniform(0, 3)} is syntactic sugar for a series of coin flips, and all numbers in \Cref{fig:plane} are syntactic sugar for a tuple of Boolean values.
Thus, \NDICE is limited to \emph{bounded} numbers.

Next, notice that \texttt{mes} could randomly take any value in $\{0,1,2\}$, even though we are only interested in the behavior of the system for a fixed measurement, namely \texttt{obs}.
Thus, we use an \verb*|observe| statement in Line~11 to enforce that \texttt{mes} and \texttt{obs} are equal.
The statement \verb|observe e| enables a form of Bayesian reasoning by "removing" executions where the condition \verb*|e| does not hold.
Formally, the statement \verb|observe e| always returns $\TRUE$, but as a side effect executions where \verb*|e| is not $\TRUE$ are defined to occur with probability $0$.
After the full program has been considered, the semantics of a \NDICE program normalizes the probabilities.
This allows us to reason about the probability of a given output, conditioned by the fact that all observations succeeded.
In our example, this means we can reason about the behavior of the program, under the condition that all random measurements from Line~9 align with the expected observations.

\paragraph{The Main Expression}
As a reminder, our goal is to reason about the probability that the vehicle is at Location~$1$ after three time steps, under the condition that we observed the sensor measurements~$0$, $1$ and $2$.
This query is modeled by the main expression of the program in Lines~15-18:

Here, we can use the \texttt{step} function to model both the moving vehicle and the sensor. 
For example, the call \texttt{step(0, 0)} in Line~15 models the first time step, where the vehicle starts at Location~$0$ and we expect the sensor measurement to be $0$.
Afterwards, Lines 16 and 17 repeat this process with the next measurement and the current position of the vehicle as returned by the previous call to \texttt{step}.
Finally, the program will return $\TRUE$ in Line~18 if the vehicle is on the runway, i.e.\ at Location~$1$.

\paragraph{The Inference Query} 
Now, the \NDICE inference algorithm can determine the maximum probability that the program returns $\TRUE$, given that all \texttt{observe} statements succeed.
However, do we really need a specialized inference algorithm for nondeterminism? 
Can we not cleverly instantiate the $\NFLIP$ statement so that we can use known algorithms for purely probabilistic programs?

One option to reduce inference with nondeterminism to the purely probabilistic case is to approximate $\NFLIP$ with the a probabilistic statement $\FLIP{\frac{1}{2}}$.
However, this approach offers no bounds on the quality of the approximation, and it could lead to \textit{arbitrarily low} results compared to the intended, nondeterministic query.
Another option is to simply consider all possible ways to resolve the nondeterminism, and perform a purely probabilistic query for each of them.
This is sound, but not feasible in practice.
Even if we restrict ourselves to enumerating deterministic behaviors, meaning nondeterminism is never resolved to a random choice, this still results in an exponential amount of queries.
Namely, for a program with $n$ flip operations (both nondeterministic or probabilistic), this may result in up to $2^n$ purely probabilistic queries.

Hence, it is infeasible to reduce inference with nondeterminism to the purely probabilistic case, and a dedicated inference algorithm like \NDICE is needed.
For the example program in \Cref{fig:plane}, \NDICE is able to automatically infer that the maximum probability that the vehicle is still on the runway after the observed sequence of measurements is about $3.6 \%$.

\subsection{Non-Compositionality of Probabilistic Inference with Nondeterminism}\label{sec:subOpt}
Before we present our inference method, there is one more issue which makes combining nondeterminism and conditioning challenging:
Local optimal choices cannot generally be composed to a global optimal solution.
To illustrate this, consider the following \NDICE program:%
\smallskip%
\begin{center}
	\begin{tabular}{c}
\begin{lstlisting}[language=Caml]
if flip(0.5) then flip(0.75)
	else if nflip()
		then let obs = observe flip(0.5) in flip(0.5)
		else let obs = observe flip(0.05) in flip(0.05)
\end{lstlisting}
\end{tabular}
\end{center}%
\smallskip%
We first consider the if-statement from Line 2 in isolation.
Notably, the observations are independent of the output of the if statement, thus the conditioned probability to return $\TRUE$ is $\frac{1}{2}$ if \texttt{nflip} chooses~$\TRUE$, and $\frac{1}{20}$ otherwise.
Therefore, the optimal \textit{local} choice for the nondeterministic flip is $\TRUE$, but this is not the optimal \textit{global} choice.
For the program as a whole, the conditioned probability to return $\TRUE$  is $\frac{2}{3}$ if \texttt{nflip} chooses $\TRUE$, and $\frac{301}{420}\approx 0.716$ otherwise.

Intuitively, this effect occurs as \verb|flip(0.75)| yields a higher chance to return $\TRUE$ than either choice of the nondeterministic flip, hence it is better to choose $\FALSE$ in order to reduce the chance of a successful observation, which minimizes the impact of the nested if-statement on the overall conditioned distribution.
However, minimizing the impact of less productive branches is not always optimal either:
If we replace \verb|flip(0.75)| with \verb|flip(0.55)|, then the optimal choice is $\TRUE$ instead.
Thus, nondeterministic choices need to be made in the context of the whole program.
In a similar vein, \citet{olmedo2018conditioning} showed that expectation transformers, which are defined inductively on the program structure, cannot be extended to combine probabilistic conditioning and nondeterminism.

\subsection{The \NDICE Inference Algorithm}
This section illustrates the \NDICE inference process using the program in \Cref{fig:example_prog}:
The program first tosses a probabilistic and a nondeterministic coin and observes that at least one flip was true.
Then, two probabilistic coins are flipped and a boolean expression decides the program's output.

\begin{figure}
\begin{tabular}{l c r}

\begin{subfigure}{0.35\textwidth}

\begin{lstlisting}[language=Caml, mathescape]
let a = $\verb*|flip|_1$(0.3) in 
let b = $\texttt{nflip}_2$() in
let t = observe(a || b) in
let c = $\verb*|flip|_3$(0.4) in
let d = $\verb*|flip|_4$(0.2) in
(a || c) && d
\end{lstlisting}

$$ \varphi = (f_1 \lor f_3) \land f_4$$
$$\gamma = f_1 \lor f_2$$
$$\left[ f_1\colon 0.3, \; f_2\colon n, \;f_3\colon 0.4, \;f_4\colon 0.2 \right]$$

\caption{Example \NDICE program and the compiled boolean formulae}
\label{fig:example_prog}
\end{subfigure}

&
\begin{subfigure}{0.225\textwidth}
	\begin{center}
		\begin{tikzpicture}[xscale=1.25, yscale=0.85]
			\node (a) at (0,0)   [draw, circle] {$f_1$};
			\node (b) at (0.75,-1)   [draw, circle] {$f_2$};
			\node (c) at (0,-2)   [draw, circle]{$f_3$};
			\node (d) at (-0.75,-3)   [draw, circle] {$f_4$};
			\node (F) at (0,-4.5)   [draw, rectangle] {$F$};
			\node (T) at (-0.75,-4.5)   [draw, rectangle] {$T$};
			\node (R) at (0.75,-4.5)   [draw, rectangle] {$R$};
			
			\path[-]          (a)  edge   [bend right=0]   node[left] {} (d);
			\path[-]          (a)  edge[dashed]   [bend right=0]   node[right] {} (b);
			
			\path[-]          (b)  edge   [bend right=0]   node[right] {} (c);
			\path[-]          (b)  edge[dashed]   [bend right=0]   node[right] {} (R);
			
			\path[-]          (c)  edge   [bend right=0]   node[right] {} (d);
			\path[-]          (c)  edge[dashed]   [bend right=0]   node[right] {} (F);
			
			\path[-]          (d)  edge   [bend right=0]   node[right] {} (T);
			\path[-]          (d)  edge[dashed]   [bend right=0]   node[right] {} (F);
		\end{tikzpicture}
	\end{center}
	\caption{Compiled ADD}
	\label{fig:example_add}
\end{subfigure}
&
\begin{subfigure}{0.375\textwidth}
	\begin{center}
	\begin{tikzpicture}[xscale=1.85, yscale=0.825]
		\node (a) at (0,0)   [draw, circle, fill=cyan!80] {$ $};
		\node (b) at (0.75,-1)   [draw, circle, fill=orange!80] {$ $};
		\node (c) at (0,-2)   [draw, circle]{$ $};
		\node (d) at (-0.75,-3)   [draw, circle] {$ $};
		\node (F) at (0,-4.5)   [draw, rectangle] {$F$};
		\node (T) at (-0.75,-4.5)   [draw, rectangle] {$T$};
		\node (R) at (0.75,-4.5)   [draw, rectangle] {$R$};

		\node (ao) at (0, -0.75)  [circle,fill,inner sep=1.5pt] {};
		\path[->]          (a)  edge   [bend right=0]   node[left, pos=0.5] {$d$} (ao);
		\path[->]          (ao)  edge   [bend right=20]   node[left, pos=0.15] {$0.3\;$} (d);
		\path[->]          (ao)  edge   [bend right=0]   node[above, pos=0.5] {$0.7$} (b);
		
		\path[->]          (b)  edge   [bend right=0]   node[above, pos=0.75] {$l,1$} (c);
		\path[->]          (b)  edge   [bend right=0]   node[right] {$r,1$} (R);
		
		\node (co) at (0, -2.7)  [circle,fill,inner sep=1.5pt] {};
		\path[->]          (c)  edge   [bend right=0]   node[right, pos=0.5] {$d$} (co);
		\path[->]          (co)  edge   [bend right=0]   node[above, pos=0.5] {$0.4$} (d);
		\path[->]          (co)  edge   [bend right=0]   node[right, pos=0.5] {$0.6$} (F);
		
		\node (do) at (-0.75, -3.6)  [circle,fill,inner sep=1.5pt] {};
		\path[->]          (d)  edge   [bend right=0]   node[left, pos=0.5] {$d$} (do);
		\path[->]          (do)  edge   [bend right=0]   node[left, pos=0.5] {$0.2$} (T);
		\path[->]          (do)  edge   [bend right=0]   node[above, pos=0.5] {$\;0.8$} (F);
		
	\end{tikzpicture}
\end{center}
	\caption{Final MDP}
	\label{fig:example_mdp}
\end{subfigure}

\end{tabular}
\caption{The full inference process for a \NDICE program illustrated.}
\label{fig:example_all}
\end{figure}
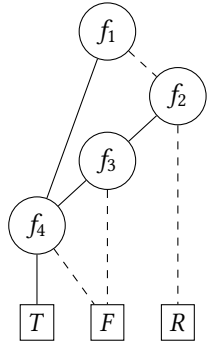
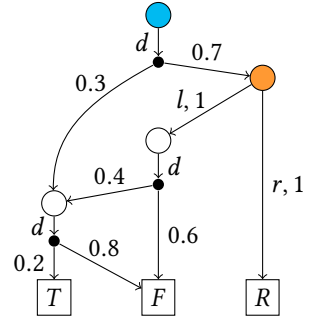

\paragraph{Compilation} 
The first step in the inference process is to compile two boolean formulas $\varphi$ and $\gamma$ as seen in \Cref{fig:example_prog}.
The variables in $\varphi$ and $\gamma$ correspond to flips in the original program, and so we compile a list of variables, where a variable is annotated with an $n$ if it belongs to a nondeterministic flip, and with its probability to return $\TRUE$ if it belongs to a probabilistic flip.
The formula $\varphi$ models the output of the compiled program, hence both the program and $\varphi$ return~$\TRUE$ if the first and last flip evaluate to~\TRUE.
Observations are ignored by $\varphi$, and instead $\gamma$ tracks if all observations in the compiled program would succeed.
Hence, $\gamma$ closely resembles the observation in Line 3.

\paragraph{ADD Construction}
Internally, $\varphi$ and $\gamma$ are maintained as \textit{Binary Decision Diagrams} (BDDs), and the next step combines them into the \textit{Algebraic Decision Diagram} (ADD) in \Cref{fig:example_add}.
This ADD models both the output \textit{and} observations of the compiled program, and we can evaluate it as follows:
We begin at the root, i.e. the topmost node, and assign either $\TRUE$ or $\FALSE$ to the variable $f_1$. 
If we assign $\TRUE$ to $f_1$, then we follow the solid edge to the node $f_4$, otherwise we follow the dashed edge to the node $f_2$.
We then repeat this process with the variable in the next node until we reach a terminal node. 
Notice that for any assignment to flip operations the ADD in \Cref{fig:example_add} either returns the same output as the program, or $R$ (as in \emph{R}eject or \emph{R}esample) if an observation failed.
Further, the ADD automatically removes redundant nodes, e.g. the ADD skips directly to $f_4$ when $f_1$ is $\TRUE$, as the program is unaffected by $\texttt{nflip}_2$ and $\texttt{flip}_3$ when $\texttt{flip}_1$ returns $\TRUE$.

\paragraph{Lifting to MDPs}
Next, we lift the ADD to the \textit{Markov Decision Process} (MDP) in \Cref{fig:example_mdp} by reintroducing the behavior of flips as stored by the list in \Cref{fig:example_prog}.
To do so, we lift each node in the ADD node to an MDP state, and each state's outgoing transitions are lifted from the ADD edges depending on the flip in the ADD node.
Consider for example the initial state of the MDP, which is highlighted in blue and matches to the root node of the ADD.
The root contains the variable $f_1$ which represents a probabilistic flip which is $\TRUE$ with probability $0.3$.
Thus, the MDP state has a single enabled action $d$ and branches randomly with probability $0.3$ between its successors. 
Similarly, the orange MDP state represents the ADD node with the nondeterministic flip $f_2$.
Hence, the orange state nondeterministically chooses between two actions, and each action only leads to one state.

Finally, we wish to reason about the behavior of the original program \textit{under the condition that all observations succeed}, therefore we need to consider the behavior of the MDP \textit{under the condition that R is not reached}.
This amounts to conditional reachability properties in MDPs which have been tackled in \citet{baier2014computing, junges2021runtime} by mimicking rejection sampling with an added \enquote{reset} transitions which returns executions that reach $R$ to the initial state.
However, we will be using a recently developed bisection method implemented in the Storm Model Checker~\citeBisection.
This algorithm is highly suited for our use-case, since it allows one to calculate conditional reachability probabilities very quickly on acyclic MDPs like the one we constructed in \Cref{fig:example_mdp}.

\medskip

Now, our original inference query, i.e. the maximum probability that the program in \Cref{fig:example_prog} outputs \TRUE, aligns with the maximum conditional probability to reach the state $T$ in the MDP.
But why take all these steps to build an MDP, if a valid MDP could be obtained by simply defining an operational semantics?
Simply put, the operational semantics of larger programs are too big to efficiently perform inference.
However, using ADDs as an intermediate step efficiently exploits redundancies to heavily reduce the MDP's size, which makes inference feasible for larger programs.

The use of ADDs might enable efficient inference, but this efficiency comes with the trade-off of lower expressiveness.
Namely, ADDs can only express Boolean functions, hence \NDICE is limited to programs which can be expressed with a finite amount of boolean variables.
Consequently, \NDICE is \textit{not} able to reason about unbounded numbers or continuous distributions, and programs need to be loop-free.
\Cref{sec:related_work} will later discuss tools which may not be subject to these limitations.


\section{Syntax and Semantics of the \protect\NDICE Language}\label{sec:semantics}

This section provides the syntax and denotational semantics for \NDICE, a first-order functional language extended with typical probabilistic statements and a nondeterministic coin flip. 

\subsection{Syntax}

\begin{figure}
	\begin{align*}
		\tau 
		\qmorespace{\Coloneqq} &
		\BOOL 
		~\;{}|{}~\; 
		\tau \times \tau 
		\qquad\qquad
		v 
		\qmorespace{\Coloneqq} 
		\TRUE  
		~\;{}|{}~\; 
		\FALSE  
		~\;{}|{}~\; 
		(v, v)
		\qquad\qquad
		\AEXP 
		\qmorespace{\Coloneqq} 
		x  
		~\;{}|{}~\; 
		v
		\\[.25em]
		e 
		\qmorespace{\Coloneqq}&
		\AEXP  
		~\;{}|{}~\; 
		(\AEXP, \AEXP) 
		~\;{}|{}~\; 
		\FST \,\AEXP  
		~\;{}|{}~\; 
		\SND \,\AEXP  
		\\
		&\quad{}~\;{}|{}~\;
		\ITE{\AEXP}{e}{e} 
		~\;{}|{}~\; 
		\LET{x = e}{e} 
		~\;{}|{}~\; 
		\FUNCNAME(\AEXP)
		\\
		&\quad{}~\;{}|{}~\;
		\FLIP{\theta} 
		~\;{}|{}~\; 
		\NFLIP 
		~\;{}|{}~\; 
		\OBSERVE{\AEXP} 
		\\[.25em]
		\mathit{fun} 
		\qmorespace{\Coloneqq} &
		\FUNC{\FUNCNAME(x\colon \tau)\colon \, \tau}{e}\\[.25em]
		p 
		\qmorespace{\Coloneqq} &
		\bullet e ~\;{}|{}~\; \mathit{fun}\;\, p
	\end{align*}%
	\caption{Syntax for the \NDICE language in A-normal form. Here, $\theta$ ranges over rational numbers in $[0,1]$}%
	\label{fig:syntax}%
\end{figure}%
The syntax for \NDICE is defined in \Cref{fig:syntax}.
\NDICE acts as a conservative extension of \DICE, and we thus also adopt the use of atomic expressions (\AEXP) to enforce an A-normal form \cite{flanagan1993essence}.
This will simplify the semantics and compilation rules without restricting the expressiveness of the language. 
For example, the tuple $(\FLIP{\frac{1}{3}}, \TRUE)$ can be expressed as $\LET{x = \FLIP{\frac{1}{3}}}{(x, \TRUE)}$.
Atomic expressions themselves are either values $v \in \VAL$ or variables $x \in \VARS$.
Much like \DICE, \NDICE expressions may be Booleans, tuples or tuple operations, function calls with non-recursive functions as well as the standard programming language features of sequential execution by let-statements or conditional branching through if-statements.
A \NDICE program $p$ is a series of function definitions 
followed by a single base expression $e$.
Later on, the practical implementation of \NDICE will relax the constraints imposed by the A-normal form, and will also introduce syntactic sugar for Boolean operations and bounded numbers.
We use these relaxations freely in our examples.

Additionally, \NDICE features two kinds of uncertainty:
Firstly, $\FLIP{\theta}$ models a \emph{probabilistic} coin flip which returns $\TRUE$ with probability $\theta \in [0, 1]$ and $\FALSE$ otherwise.
Secondly, $\NFLIP$ models a \emph{nondeterministic} choice, also taking values $\TRUE$ or \FALSE, but with no probability attached.
The use of these expressions is not restricted, and it is possible to interleave probabilistic and nondeterministic behavior arbitrarily.
Finally, \NDICE includes the statement \OBSERVE{e} to enable conditioning via observations. 
For purely probabilistic programs, this amounts to Bayesian reasoning, but \NDICE also features conditioning in the presence of nondeterminism.

\subsection{Semantics}

Probabilistic programs typically model a probability distribution over their output values. 
In the presence of nondeterminism, however, this is not sufficient as each resolution of nondeterminism may lead to a different probability distribution, and thus we generalize from distributions to sets of distributions.
This approach is grounded in category theory, where both probability and nondeterminism can be modeled as a monad. 
Namely, nondeterminism is modeled by the powerset monad, whereas probabilistic choice corresponds to the monad of finite distributions \cite{goy2020combining}.
Combining nondeterminism and probability thus amounts to composing these monads, which requires a distributive law between them \cite{beck1969distributive}.
Sadly, \citet{varacca2006distributing} showed that there is no distributive law between the monad of finite distributions and the powerset monad, making a suitable composition impossible.
This leads to difficulties for the design of programming languages:
A semantics which combines nondeterminism and probability naively will not be compositional, and the semantics of $c_1; (c_2; c_3)$ may differ from $(c_1; c_2); c_3$.

To address this, many approaches to combining the probabilistic and nondeterministic monad have been developed, for example \citet{varacca2006distributing} generalize the probabilistic monad to indexed valuations and \citet{goy2020combining} use weak distributive laws instead.
Our semantics will follow the work of \citet{mislove2000nondeterminism} and \citet{tix2009semantic}, which retains compositionality by requiring the sets of distributions to be convex.
We call a set $S$ of distributions convex, if for all distributions $d_1, d_2 \in S$ and $p \in [0,1]$, the distribution $p\cdot d_1 + (1-p)\cdot d_2$ is also in $S$.
Hence, we will use convex sets to model nondeterministic behavior, which has been applied to programming languages before \cite{zilberstein2025demonic, varacca2006distributing, jifeng1997probabilistic}. 
\\

\begin{figure}
		\begin{center}
		\begin{adjustbox}{max width=\textwidth}
				\renewcommand{\arraystretch}{2.1}
				\begin{tabular}{c@{\qquad}c}
					\multicolumn{2}{c}{ 
					$\llbracket v \rrbracket \ttriangleq \{1|v\rangle \}$ 
					\hfill
					
					$\llbracket \FLIP{\theta} \rrbracket \ttriangleq \{\theta|\TRUE\rangle + (1 - \theta)|\FALSE\rangle\}$
					
					\hfill
					$\llbracket \NFLIP \rrbracket \ttriangleq
					\{\theta|\TRUE\rangle +  (1-\theta)|\FALSE\rangle \,|\, \theta \in [0,1]\}$

					}
					\\

					\multicolumn{2}{c}{ 
						
						$\llbracket \FST\,(v_1, v_2) \rrbracket \ttriangleq \{1|v_1\rangle \}$
						\hfill

						$\llbracket \SND\,(v_1, v_2) \rrbracket \ttriangleq \{1|v_2\rangle \}$
						\hfill
						
						$\llbracket \FUNCNAME(v) \rrbracket^T \ttriangleq T(\FUNCNAME)(v)$
					}
					\\[1em]
					%
					$\llbracket \OBSERVE{v} \rrbracket \ttriangleq \begin{cases}
						\{1|\TRUE\rangle\},	&\text{if } v = \TRUE\\
						\{0|\TRUE\rangle\},	&\text{if } v = \FALSE\\
						\emptyset, 		& \text{otherwise}\\
					\end{cases}$
					&
					
					$\llbracket \ITE{v_g}{e_1}{e_2} \rrbracket \ttriangleq \begin{cases}
						\llbracket e_1\rrbracket, 	& \text{if } v_g = \TRUE\\
						\llbracket e_2\rrbracket, 	& \text{if } v_g = \FALSE\\
						\emptyset,				& \text{otherwise}\\
					\end{cases} $
					\\
					
					\multicolumn{2}{c}{ 
						\renewcommand{\arraystretch}{1.2}
						$\begin{matrix}
							\phantom{space maker}\\
							d \iin \llbracket \LET{x = e_1}{e_2} \rrbracket\\
							\qmorespace{\Longleftrightarrow} 
							\exists\, d_1 \in \llbracket e_1 \rrbracket~{}~{}~
							\forall\, v, d_1(v) > 0~{}~{}~
							\exists\, d_v \in \llbracket e_2[x \mapsto v] \rrbracket~{}~{}~
							\forall\, v'\colon \quad
							d(v') = \sum\limits_{v,\; d_1(v) > 0 } d_1(v) \cdot d_{v}(v')
						\end{matrix}$
					}
						\end{tabular}
			\end{adjustbox}
		\end{center}
	\caption{Semantics of \NDICE expressions. The function table $T$ is left implicit in most rules.}
	\label{fig:semantics} 
\end{figure}

Before we present the semantics of \NDICE, we will briefly introduce some notations. 
Given a finite set $X$, we call a function $d \colon X \rightarrow [0,1]$ a distribution if $\sum_{x \in X} d(x) = 1$, and we call $d$ a \textit{sub}-distribution if $\sum_{x \in X} d(x) \leq 1$. 
We use bra-ket notation $d = p_1{| v_1 \rangle} + {\ldots} +  p_n {|v_n \rangle}$ to denote a sub-distribution with $d(v_i) = p_i$.
We denote the set of all sub-distributions over $X$ as $\mathcal{D}(X)$, and we denote the powerset of $X$ by $\mathcal{P}(X)$.
We then define the semantics of \NDICE expressions as convex sets of (sub-)distributions.
Formally, the semantics function $\llbracket \:\cdot\: \rrbracket^T\colon \NDICE \rightarrow  \mathcal{P}(\mathcal{D}(\mathit{Val}))$ associates to each \NDICE expression a set of sub-distributions over values which may be generated by it.
The semantics function is parametrized with a function table $T$ to store the behavior of functions.
Note that the function table is only needed for function calls, therefore we omit $T$ in most cases.

The semantics for closed \NDICE expressions are defined in \Cref{fig:semantics}, and they conservatively extend the semantics of \DICE~\cite{holtzen2020scaling} for well-formed programs.
For the sake of completeness, we define the semantics of ill-formed programs, e.g.\ programs which use non-Boolean values as guards of if-statements, to be the empty set.
From now on we will assume that programs are well-formed, as our inference algorithm only considers well-formed programs.

We begin by considering the semantics of values and tuple operations.
Notably, these expressions contain neither (a) nondeterminism nor (b) randomness.
Their semantics are hence (a) \emph{singleton} sets with only one (b) \emph{Dirac} distribution assigning probability 1 to a single value.
The probabilistic $\FLIP{\theta}$ is not nondeterministic but is random.
Thus, its semantics is a still a singleton but containing a non-Dirac distribution over \textit{two} possible values. 
The semantics of $\ITE{v_g}{e_1}{e_2}$ is standard. 
Due to the A-normal form, condition $v_g$ is always a value.
If it is $\TRUE$, then the if-statement behaves like $e_1$, otherwise it behaves like $e_2$. 
Next, we examine the remaining rules in greater detail.

\paragraph{Nondeterministic Flips}
The statement $\NFLIP$ is nondeterministic, but not inherently random. 
Still, we want non-determinism to act as an abstraction for all possible decision procedures we could use to return either $\TRUE$ or $\FALSE$, which includes random choices.
Thus, the semantics of $\NFLIP$ do not include just the deterministic distributions $1|\TRUE\rangle$ and $1|\FALSE\rangle$, but also all possible distributions in between.
Additionally, the semantics of $\NFLIP$ then form a convex set.

\paragraph{Let Statement}
The semantics of $\LET{x=e_1}{e_2}$ is a set of distributions that emerges from all resolutions of the nondeterminism in both $e_1$ and $e_2$.
For this, we first pick one distribution $d_1$ from the set $\llbracket e_1 \rrbracket$, thus resolving the nondeterminism in $e_1$.
Now, since there may be many values $v$ with $d_1(v) > 0$, and $e_2$ may depend on $v$, we must next pick a distribution $d_v$ from $\llbracket e_2 [x\mapsto v]\rrbracket$ -- where $e_2 [x\mapsto v]$ is $e_2$ with all free occurrences of $x$ replaced by $v$ -- for all possible intermediate values $v$.
We have now fully resolved the nondeterminism within the let statement, and the probability to output a value $v'$ can be found by considering all possible execution paths, so that $e_1$ first terminates in an intermediate value $v$, and then $e_2[x\mapsto v]$ terminates in $v'$.
Like this, we can determine the distribution for \textit{one} possible resolution to the nondeterminism. 
The semantics of the let-statement as a whole are generated by considering \textit{all} possible resolutions to the nondeterminism, which is done by considering all combinations of distributions which could be picked from the semantics of $e_1$ and $e_2$.
To illustrate this, consider the following example:
\begin{align*}
	\LET{x = \NFLIP}{\LET{y = \FLIP{\tfrac{2}{3}}}{x \leftrightarrow y}} \tag{ExLet1}\label{ExLet1}\\
	\llbracket ExLet1 \rrbracket = \left\{p | \TRUE \rangle + (1-p)| \FALSE \rangle \,|\, p \in \left[\tfrac{1}{3},\tfrac{2}{3}\right]\right\}
\end{align*}
To generate a distribution from $\llbracket \text{\ref{ExLet1}} \rrbracket$, we first pick a distribution from $\llbracket \NFLIP \rrbracket$, for example $d_1 = \frac{1}{4}|\TRUE\rangle + \frac{3}{4}|\FALSE\rangle$.
This gives us a two values with $d_1(v) > 0$, namely $\TRUE$ and $\FALSE$.  
Thus, we have to pick distributions from the two singleton sets $\llbracket \LET{y = \FLIP{\tfrac{2}{3}}}{{\color{orange}\TRUE} \leftrightarrow y} \rrbracket = \left\{\frac{2}{3}|\TRUE\rangle + \frac{1}{3}|\FALSE\rangle\right\}$ and $\llbracket \LET{y = \FLIP{\tfrac{2}{3}}}{{ \color{orange} \FALSE} \leftrightarrow y} \rrbracket = \left\{\frac{1}{3}|\TRUE\rangle + \frac{2}{3}|\FALSE\rangle\right\}$.
Combining the three distributions then generates a possible distribution for the overall let statement, namely:
$$
\left(\tfrac{1}{4}\cdot \tfrac{2}{3} + \tfrac{3}{4}\cdot \tfrac{1}{3}\middle)\middle|\TRUE\right\rangle + 
\left(\tfrac{1}{4}\cdot \tfrac{1}{3} + \tfrac{3}{4}\cdot \tfrac{2}{3}\middle)\middle|\FALSE\right\rangle 
= \left. \tfrac{5}{12}\middle|\TRUE\right\rangle + \left. \tfrac{7}{12}\middle |\FALSE \right\rangle
\in \llbracket \text{\ref{ExLet1}} \rrbracket 
$$
Notably, the behavior of this expression changes if we swap the order of the flips:
\begin{align*}
	\LET{x = \FLIP{\tfrac{2}{3}}}{\LET{y = \NFLIP}{x \leftrightarrow y}}  \tag{ExLet2}\label{ExLet2}\\
	\llbracket ExLet2 \rrbracket = \{p | \TRUE \rangle + (1-p)| \FALSE \rangle \,|\, p \in [0,1]\}
\end{align*}
Intuitively, the additional distributions arise as $\NFLIP$ can act depending on the outcome of the previous flip.
More formally, the probabilistic flip may return two different values, which results in \textit{two} opportunities to pick a distribution from the semantics of the nondeterministic flip.
In particular, we may pick a different distribution depending on the outcome of  the probabilistic flip, and thus we could pick $1|\TRUE\rangle$ or $1|\FALSE\rangle$ to always match the previous flip or not.
This does not occur for \ref{ExLet1} as the outcome of the probabilistic flip is still unclear when the nondeterministic flip happens.
Hence, nondeterministic choices can depend on past outcomes, but they cannot predict the future.%

\paragraph{Observations and Conditioning}
The observe statement allows us to condition the behavior of expressions on the assumption that all observations succeed, i.e.\ the arguments to all observe statements were $\TRUE$.
Our semantics follow prior work and "filter out" failed observations by assigning them a probability of $0$ \cite{claret2013bayesian, fierens2015inference, borgstrom2011measure, huang2016application}.
We will illustrate this with the following example:
\begin{align*}
	\LET{x = \NFLIP}{\LET{ y = \FLIP{\tfrac{2}{3}} }{ \LET{t =\OBSERVE{x\lor y}}{y}} } \tag{ExObs}\label{ExObs}\\
	\llbracket ExObs \rrbracket = \left\{\left. \tfrac{2}{3}\middle|\TRUE \right\rangle + \left(\tfrac{1}{3} \cdot p\middle)\middle|\FALSE\right\rangle \;\middle|\; p\in \left[0, 1\right]\right\}
\end{align*}
The semantics of \ref{ExObs} can be explained as follows: With a probability of $\frac{2}{3}$ the second flip returns $\TRUE$, meaning the observation succeeds and the program returns $\TRUE$.
But, if the probabilistic flip returns $\FALSE$, then the observation only succeeds if the first flip returned $\TRUE$.
Thus, the probability to return $\FALSE$ depends on the resolution of the nondeterminism, which is represented by the parameter $p$.

Notably, $\llbracket \text{\ref{ExObs}} \rrbracket$ contains \textit{sub}-distributions like $\frac{2}{3}|\TRUE \rangle + \frac{1}{6}|\FALSE\rangle$, where the missing probability mass amounts to the likelihood for the observation to fail as both flips returned $\FALSE$.
We thus call the semantics in \Cref{fig:semantics} \textit{unnormalized semantics}. 
Retrieving the normalized semantics of an expression $e$ is straightforward, as we simply normalize the distributions in $\llbracket e \rrbracket$.
To do so, we define the normalizing constant $d_D$ of a distribution $d$ and the normalized semantics $\llbracket e \rrbracket_D$ of an expression $e$:
$$ d_D \coloneqq \sum_{v} d(v) \qquad \qquad \llbracket e \rrbracket_D \coloneqq \left\{d' \,\left|\; \exists d\in \llbracket e \rrbracket \colon \forall v \colon d'(v) = \frac{d(v)}{d_D} \right. \right\}$$
For $d_D = 0$ we define the fraction to be $0$.
The semantics $\llbracket e \rrbracket_D$ describe all output distributions of $e$, conditioned by all observations succeeding.
For example, consider $\llbracket \text{\ref{ExObs}} \rrbracket_D$ compared to $\llbracket \text{\ref{ExObs}} \rrbracket$, where $\TRUE|\frac{2}{3}\rangle + \frac{1}{3} |\FALSE \rangle$ remains unchanged, but $\frac{2}{3}|\TRUE \rangle$ will be normalized to $1|\TRUE\rangle$.
Considering unnormalized semantics first and normalizing at the end is standard, and we share this approach with \DICE and other PPLs \cite{claret2013bayesian, fierens2015inference}.

\paragraph{Function Calls and Programs}
The semantics of a function $\FUNCNAME(x \colon \tau_1)$ takes an input value $v$ of type $\tau_1$ and returns the semantics of $\FUNCNAME$'s function body $e$ (i.e.\ a set of distributions), but in which each occurrence of $x$ is replaced by $v$.
Hence, the semantics of a \NDICE function $\FUNCNAME$ is of type $\llbracket \FUNC{\FUNCNAME(x\colon \tau_1)\colon \, \tau_2}{e} \rrbracket: \mathit{Val} \rightarrow \mathcal{P}(\mathcal{D}(\mathit{Val}))$ and given formally by:
$$\llbracket \FUNC{\FUNCNAME(x\colon \tau_1)\colon \, \tau_2}{e} \rrbracket (v) \triangleq \llbracket e[x \mapsto v]\rrbracket$$
In \Cref{fig:semantics}, we use this semantical object as stored in the function table $T$ to define the semantics for function calls: 
We look up the identifier $\FUNCNAME$ of the called function in the function table $T$ and return the corresponding set of unnormalized distributions for the input value $v$.

We can now define the semantics of a \NDICE program $p$ as $\llbracket p \rrbracket^T \colon  \mathcal{P}(\mathcal{D}(\mathit{Val}))$.
We determine $\llbracket p \rrbracket^T$ by first adding each function to the function table, until eventually the base expression $\bullet \, e$ is reached.
Finally, the semantics of the program as a whole are the semantics of the base expression $e$ in the context of the accumulated function table $T$.
This behavior is formalized by the following rules, where $\bullet$ denotes the empty sequence and $\eta(fun)$ retrieves the function's identifier:
$$ \llbracket \bullet \, e \rrbracket^T \triangleq \llbracket e \rrbracket^T \qquad \llbracket fun \; p \rrbracket^T \triangleq \llbracket p \rrbracket^{T \cup \{\eta(fun) \rightarrow \llbracket fun \rrbracket^T\}} $$

\section{Technical Background of the Inference Process}\label{sec:background}
 This section presents the necessary background information and prior work which is needed for the technical aspects of the \NDICE inference algorithm.
 In particular, we will recap the \DICE-style compilation procedure from \citet{holtzen2020scaling} in \Cref{sec:prelims_Dice}, and we will introduce Decision Diagrams and Markov Decision Processes in \Cref{sec:prelims_ADD} and \Cref{sec:prelims_MDP}.

\subsection{\DICE-style Compilation}\label{sec:prelims_Dice}

We will now summarize \DICE's compilation procedure from \citet{holtzen2020scaling} by recapping how \DICE expressions are compiled to Boolean formulas.
For our purposes, \DICE expressions will be defined as \NDICE expressions \textit{without} the statement $\NFLIP$.

A \emph{compilation judgment} for a \DICE expression $e$ has the form $\Gamma \vdash e \colon \tau \leadsto (\dot{\varphi},\, \gamma,\, w)$. 
We say that, given the type environment $\Gamma$, the expression $e$ of type $\tau$ compiles to the triple $(\dot{\varphi},\, \gamma,\, w)$.
The type environment $\Gamma$ maps variables $x$ to their respective type $\tau_x$, where $\tau_x$ is needed to correctly encode program variables in the compilation result.
We call $\dot{\varphi}$ the \textit{model formula}, as it models the the output value of $e$ given some assignment to flips and free variables, while ignoring all observations.
The added dot in $\dot{\varphi}$ denotes a (possibly nested) \emph{tuple} of Boolean formulas, as opposed to a \emph{single} formula~$\varphi$.
The \textit{accepting formula} $\gamma$ is a single Boolean formula, which is satisfied if and only if all observations in $e$ would succeed for a given assignment to flips and variables.
Each Boolean variable in $\dot{\varphi}$ and $\gamma$ represents a free variable of the compiled expression or a flip operation. 
Hence, we need to remember how the corresponding flip for each variable behaved.
\DICE thus compiles a set of weights $w$, which (for our purposes) is a set of pairs $f_i\colon p_i$, where $f_i$ is the Boolean variable which represents a flip and the label $p_i$ denotes the flip's probability to return $\TRUE$.

Rules for compilation of well-typed \DICE expressions without function calls from \citet{holtzen2020scaling} are shown in  \Cref{fig:dice_Comp}.
We will now briefly go over these rules.
For examples and more explanations we refer to the original work by \citet{holtzen2020scaling}.

\begin{figure}
	\begin{center}
		\begin{adjustbox}{max width=\textwidth}
			\renewcommand{\arraystretch}{3.2}
			\begin{tabular}{c c}
				
				\multicolumn{2}{c}{
				\inference[\ifthenelse{\boolean{showRuleNames}}{\texttt{(C-Val)}}{}]
				{}{  \Gamma \vdash v \colon \tau \leadsto (v, \TRUE, \{\})}
				$\quad$
				
				\inference[\ifthenelse{\boolean{showRuleNames}}{\texttt{(C-Ident)}}{}]
				{\Gamma(x) = \tau}{  \Gamma \vdash x \colon \tau \leadsto (F_\tau(x), \TRUE, \{\})}
				
				$\quad$
				
				\inference[\ifthenelse{\boolean{showRuleNames}}{\texttt{(C-Flip)}}{}]
				{\text{variable $f$ fresh}}{  \Gamma \vdash \FLIP{\theta} \colon \BOOL \leadsto (f, \TRUE, \{f\colon \theta\})}
			}
				
				\\
				
				\inference[\ifthenelse{\boolean{showRuleNames}}{\texttt{(C-Tup)}}{}]
				{\Gamma(x_1) = \tau_1 \qquad \Gamma(x_2) = \tau_2}{  \Gamma \vdash (x_1, x_2) \colon \tau_1 \times \tau_2 \leadsto ((F_{\tau_1}(x_1), F_{\tau_2}(x_2)) , \TRUE, \{\})}
				&
				
				\inference[\ifthenelse{\boolean{showRuleNames}}{\texttt{(C-Obs)}}{}]
				{\Gamma \vdash \AEXP \colon \BOOL \leadsto (\varphi, \TRUE, [])}
				{  \Gamma \vdash \OBSERVE{\AEXP} \colon \BOOL \leadsto (\TRUE, \varphi, \{\})}
				\\
				
				\inference[\ifthenelse{\boolean{showRuleNames}}{\texttt{(C-Fst)}}{}]
				{\Gamma(x) = \tau_1 \times \tau_2}
				{ \Gamma \vdash \FST x \colon \tau_1 \leadsto (F_{\tau_1}(x_l), \TRUE, \{\})}
				&
				
				\inference[\ifthenelse{\boolean{showRuleNames}}{\texttt{(C-Snd)}}{}]
				{\Gamma(x) = \tau_1 \times \tau_2}
				{ \Gamma \vdash \SND x \colon \tau_1 \leadsto (F_{\tau_2}(x_r), \TRUE, \{\})}
				\\
				
				\multicolumn{2}{c}{
					\inference[\ifthenelse{\boolean{showRuleNames}}{\texttt{(C-ITE)}}{}]
					{\Gamma \vdash \AEXP \colon \BOOL \leadsto (\varphi_g, \TRUE, \{\}) 
						\qquad \Gamma \vdash e_T \colon \tau \leadsto (\dot{\varphi}_T, \gamma_T, w_T) 
						\qquad \Gamma \vdash e_F \colon \tau \leadsto (\dot{\varphi}_F, \gamma_F, w_F)}
					{  \Gamma \vdash \ITE{\AEXP}{e_T}{e_F} \colon \tau \leadsto \left( 
						(\varphi_g \overset{\tau}{\land} \dot{\varphi}_T) \overset{\tau}{\lor} (\overline{\varphi}_g \overset{\tau}{\land} \dot{\varphi}_F), \;
						(\varphi_g \land \gamma_T) \lor (\overline{\varphi}_g \land \gamma_F) ,\;
						w_T \,\cup\, w_F \right)}
				}\\
				
				\multicolumn{2}{c}{
					\inference[\ifthenelse{\boolean{showRuleNames}}{\texttt{(C-Let)}}{}]
					{\Gamma \vdash e_1 \colon \tau_1 \leadsto (\dot{\varphi}_1, \gamma_1, w_1) \qquad 
						\Gamma \cup \{x \mapsto \tau_1 \} \vdash e_2 \colon \tau_2 \leadsto (\dot{\varphi}_2, \gamma_2, w_2)}
					{  \Gamma \vdash\LET{x = e_1}{e_2} \colon \tau_2 \leadsto \left(
						\dot{\varphi}_2[x\overset{\tau_1}{\mapsto} \dot{\varphi}_1], \;
						\gamma_1 \land \gamma_2[x\overset{\tau_1}{\mapsto} \dot{\varphi}_1],\;
						w_1 \,\cup\, w_2 \right)}
				}
				
			\end{tabular}
		\end{adjustbox}
	\end{center}
	\caption{
		\DICE-Compilation rules for expressions. These assume, without loss of generality but for simplicity, that fst, snd, and tuple construction are only ever performed with identifiers as arguments.}
	\label{fig:dice_Comp} 
\end{figure}%

\paragraph{Values}
The rule \texttt{(C-Val)} compiles constant values.
Since the output of $v$ is the value~$v$ itself, $v$ is the model formula for $v$.
Furthermore, a value contains neither observations nor flips, thus the accepting formula is simply $\TRUE$ and the set of weights is the empty set.

\paragraph{Free Variables and Tuples}
For variables, the structure of the model formula has to match the structure of the variable's type; thus we need to unfold the variable into its components.
Hence, the rule \texttt{C-Ident} uses the \textit{form function} $F_\tau$, which given the variable $x$ of type $\tau_x$ builds a \enquote{template} for~$x$ by introducing a free variable for each boolean component of $x$.
$F_\tau$ is defined inductively as $F_\BOOL (x) \triangleq x$ and $ F_{\tau_1 \times \tau_2}(x) \triangleq (F_{\tau_1}(x_l), F_{\tau_2}(x_r))$, where we use the subscripts $x_l$ and $x_r$ to distinguish the left and right elements of $x$.
The rules for tuples also make use of $F_\tau$, and \texttt{C-Fst} and \texttt{C-Snd} build representations of $x_l$ and $x_r$ to refer to the correct component of $x$.

\paragraph{Flips}
The rule \texttt{C-Flip} introduces a fresh variable $f$ to represent the flip and uses $f$ as the model formula.
Afterwards, \texttt{C-Flip} pairs $f$ with the flip's likelihood to return $\TRUE$, and defines the set of weights as a singleton set with this pair as its only element.

\paragraph{Observations}
The input to an observe statement is a Boolean atomic expression, and thus either a variable identifier or a value.
\texttt{C-Obs} hence uses the model formula of the input as the accepting formula of the observe statement, as the observation succeeds if the input evaluates to $\TRUE$. 
The model formula is always $\TRUE$ in accordance with the semantics from \Cref{fig:semantics}.

\paragraph{If Statements}
The if statement $\ITE{g}{e_T}{e_F}$ behaves like $e_T$ whenever the guard $g$ is $\TRUE$, and otherwise like $e_F$.
The compiled formulas mirror this idea of switching between $e_T$ and $e_F$ based on the behavior of the guard $g$, as exemplified by the accepting formula $(\varphi_g \land \gamma_T) \lor (\overline\varphi_g \land \gamma_F)$:
Whenever the guard returns $\TRUE$, i.e.\ $\varphi_g$, holds, then we use the accepting behavior of $e_T$ as encoded in $\gamma_T$, otherwise we consider the accepting behavior of $e_F$ as encoded in $\gamma_F$.
The model formula of the if-statement is constructed from the model formulas of $e_T$ and $e_F$ in a similar way, where we once again use $\varphi_g$ to choose between $e_T$ and $e_F$. 
However, $\dot{\varphi}_T$ and $\dot{\varphi}_F$ could be tuples of formulas, and thus we define how the Boolean function $\varphi_g$ interacts with tuples of formulas via broadcast conjunction $\varphi_a \overset{\tau}{\land} \dot{\varphi}_b$ and point-wise disjunction $\dot{\varphi}_a \overset{\tau}{\lor} \dot{\varphi}_b$:
\begin{align*}
	\varphi_a \overset{\BOOL}{\land} \varphi_b &\coloneq \varphi_a \land \varphi_b
	&
	\varphi \overset{\tau_1 \times \tau_2}{\land} (\dot{\varphi}_l,\, \dot{\varphi_r}) &\coloneq( \varphi \overset{\tau_1}{\land} \dot{\varphi}_l, \varphi \overset{\tau_2}{\land} \dot{\varphi}_r) \\
	\varphi_a \overset{\BOOL}{\lor} \varphi_b &\coloneq \varphi_a \lor \varphi_b
	&
	(\dot{\varphi}_{a,l},\, \dot{\varphi}_{a,r}) \overset{\tau_1 \times \tau_2}{\lor} (\dot{\varphi}_{b,l},\, \dot{\varphi}_{b,r}) &\coloneq(  \dot{\varphi}_{a,l} \overset{\tau_1}{\lor} \dot{\varphi}_{b,l},\; \dot{\varphi}_{a,r} \overset{\tau_2}{\lor} \dot{\varphi}_{b,r})
\end{align*}
The sets of weights are combined by simply taking the (disjoint) union.

\paragraph{Let Statement}
The let-statements sequentially executes $e_1$ and $e_2$, where $e_1$ decides the value of~$x$ which is used in $e_2$ to determine the output of the let-statement as a whole.
\texttt{C-Let} mimics this in the model formula with substitutions $\varphi_2[x\mapsto \varphi_1]$ to describe that $x$ behaves according to the model formula of $e_1$ within the model formula of $e_2$.
These substitutions replace all occurrences of~$x$ in $\varphi_2$ with the term $\varphi_1$, and we define typed substitutions with or into tuples as follows:
$$
\varphi_2 [x \overset{\BOOL}{\mapsto} \varphi_1] \morespace{\coloneqq} \varphi_2 [x \mapsto \varphi_1]
\qquad \qquad
\varphi_2 [x \overset{\tau_1 \times \tau_2}{\mapsto} (\dot{\varphi}_{l},\, \dot{\varphi}_{r})] \morespace{\coloneqq} \varphi_2 [x_l \overset{\tau_1}{\mapsto} \dot{\varphi}_{l}][x_r \overset{\tau_2}{\mapsto} \dot{\varphi}_{r}]$$ 
$$(\dot{\varphi}_a,\, \dot{\varphi}_b)[x\overset{\tau}{\mapsto}\dot{\varphi}] \morespace{\coloneqq} (\dot{\varphi}_a[x\overset{\tau}{\mapsto}\dot{\varphi}],\, \dot{\varphi}_b[x\overset{\tau}{\mapsto}\dot{\varphi}])$$
The accepting formula follows the same principle, but we add $\gamma_1$ to ensure that all observations in $e_1$ succeed.
Once again, the sets of weights are combined by taking the (disjoint) union.

\paragraph{Functions and Programs}
We have only discussed how \DICE \textit{expressions} are compiled, but the \DICE language also includes function calls and function definitions which can be used to form programs.
In \citet{holtzen2020scaling}, such \DICE programs are addressed with additional compilation rules, but the result of the compilation is still a triple of the form $(\dot{\varphi},\, \gamma,\, w)$.
Hence, all subsequent inference steps are identical between expressions and programs, and thus we will not cover the compilation of \DICE programs here.
We refer to \citet{holtzen2020scaling} for more details.

\subsection{Algebraic Decision Diagrams}\label{sec:prelims_ADD}
Decision Diagrams are a key component in the inference process, and we will briefly introduce the necessary background on both Algebraic Decision Diagrams (ADDs) and Binary Decision Diagrams (BDDs).
We refer to \citet{husung2024oxidd,andersen1996introduction} for a more extensive introduction.
An ADD (or Multi-Terminal BDD \cite{fujita1997multi}) encodes a function $f \colon \BOOL^k \rightarrow S$ with boolean input variables $x_1,\dots, x_k$ and a set of output values $S$, and is defined as follows:%
\begin{definition}[Algebraic Decision Diagrams]
	An ADD $\mathcal{A} = (V,\, T,\, r,\, E)$ is a directed acyclic graph~(DAG) with root node $r$, a set of inner nodes $V$, a set of terminal nodes $T$, and a set of labeled edges $E$.
	All inner nodes $u \in V$ contain a variable $\mathit{var}(u)$, and $u$ has exactly two outgoing edges: 
	One labeled $\mathit{then}$ and one labeled $\mathit{else}$.  
	All terminal nodes $u \in T$ contain a value $\mathit{val}(u) \in S$.
	\hfill$\triangleleft$
\end{definition}%
\noindent%
For each node $u$, we refer to the node reached by taking the edge labeled $\mathit{then}$ as $\mathit{then}(u)$, and analogously $\mathit{else}(u)$.
Binary Decision Diagrams \cite{bahar1993ADDs} are ADDs with $S = \BOOL$.
We only consider \textit{ordered} ADDs, i.e.\ there exists some order on the input variables $x_i$ so that the variables along any path from the root to some terminal node are in that order.
The semantics of an ADD node are defined recursively as follows:
A terminal node $u$ represents the constant function $f_u = \mathit{val}(u) \in S$.
An inner node $v$ encodes the following function:
$$f_v = \begin{cases}
	f_{\mathit{then}(v)}, & \text{if } \mathit{var}(v) = \TRUE \\
	f_{\mathit{else}(v)}, & \text{if } \mathit{var}(v) = \FALSE
\end{cases}$$%
However, we are less interested in the semantics of ADDs, than in their ability to compactly encode any Boolean function as a DAG.
In fact, ADDs can represent \emph{any} function of type $f \colon \BOOL^k \rightarrow S$, and generic algorithms can be used to construct the ADD for any given function \cite{fujita1997multi}.
The compactness of the resulting ADD stems from the following restrictions:
\begin{itemize}
	\item \textit{No redundant tests:} For any inner node $u$ we have $\mathit{then}(u) \neq \mathit{else}(u)$.
	\item \textit{No duplicate tests:} For any two distinct inner nodes $u$ and $v$ we have either $\mathit{var}(u) \neq \mathit{var}(v)$, $\mathit{then}(u) \neq \mathit{then}(v)$, or $\mathit{else}(u) \neq \mathit{else}(v)$. \label{DD:no_dups}
	\item \textit{No duplicate outputs:} For any two distinct terminal nodes $u$ and $v$ we have $\mathit{val}(u) \neq \mathit{val}(v)$.
\end{itemize}
An ADD adhering to these restrictions is called \textit{reduced}.
For a fixed variable ordering, the reduced ADD of a function is unique and often quite small, making ADDs a suitable representation for many functions.
However, the chosen variable ordering can heavily impact how compactly a given function can be represented \cite{bryant1986bdd}.
In extreme cases, the reduced ADD of some functions may be linear in size to the number of variables for some variable orderings, but exponentially large with other orderings.
Finding the optimal ordering is NP-hard \cite{bollig1996improving}, and there are also functions where the ADD always is exponentially large \cite{bryant1986bdd}.

\subsection{Markov Decision Processes and Conditional Reachability}\label{sec:prelims_MDP}
We will briefly introduce Markov Decision Processes (MDPs) and conditional reachability properties, we refer to \citet{baier2014computing,puterman2014markov} for more details.
We first define MDPs:%
\begin{definition}[Markov Decision Process]
	A Markov Decision Process (MDP) is a six-tuple $(S,\, s_{init},\, \mathit{Act},\, P,\, AP,\, L)$, where $S$ is a set of states, $s_{init} \in S$ is the initial state, $AP$ is a set of atomic propositions, $L\colon S \rightarrow\mathcal{P}(AP)$ is a labeling function, $\mathit{Act}$ is a set of actions and $P\colon S\times \mathit{Act} \times S \rightarrow [0,1]$ is a transition function, so that for all $s\in S$ and $\alpha \in \mathit{Act}$ we have:%
	\begin{align*}
		\sum_{s' \in S} P(s, \alpha, s') \in \{0,1\}
		\tag*{$\triangleleft$}
	\end{align*}
\end{definition}%
\noindent%
An action $\alpha$ is called enabled in $s$ if there is some $s'$ with $P(s, \alpha, s') > 0$.
We write $\smash{s \xrightarrow{\alpha, p} s'}$ for $P(s, \alpha, s') = p$. 
For any state $s$, $L(s)$ is the set of atomic propositions which hold in that state, we also use labels $a \in AP$ to refer to the set $\{s \in S \,|\, a \in L(s)\}$.
We define a path $\pi \in \mathit{Paths}(s_0)$ with starting state $s_0 \in S$ in the MDP as a (finite) sequence of states, written as $\pi = s_0\, s_1 \dots s_n$.
Reasoning about the probability of a path in the MDP requires us to resolve the nondeterministic choice between the different enabled actions in a state.
This is done by a scheduler $\Theta\colon S \rightarrow \mathcal{D}(\mathit{Act})$, sometimes also called a policy or adversary, which associates each state in the MDP with a distribution over (enabled) actions.
We only consider memoryless schedulers as they are sufficient for unbounded reachability properties \cite{forejt2011automated}. 
Given a scheduler $\Theta$, we determine the probability of a path $\pi = s_0\, s_1 \dots s_n$ as the product of all transition probabilities along the path:
$$ Pr^\Theta(\pi) \coloneqq \prod_{0\leq i < n} \; \sum\limits_{\alpha \in \mathit{Act}} \Theta(s_i)(\alpha) \cdot P(s_i,\, \alpha,\, s_{i + 1})$$
We also define the probability to reach a set of states $F \subseteq S$ starting from the initial state $s_{init}$ in $\mathcal{M}$ under a scheduler $\Theta$.
To do so, we consider the set $\lozenge F \subseteq \mathit{Paths}(s_{init})$, which contains all (finite) paths which start in $s_{init}$ and traverse states which are not in $F$ until a state in $F$ is reached. 
Formally, $\lozenge F$ is defined as $\lozenge F \coloneqq \{ s_0\, s_1 \dots s_n \,|\, \forall i < n\colon s_i \notin F \land s_n \in F \}$.
We then define the probability to eventually reach $F$ in the MDP $\mathcal{M}$ under the scheduler $\Theta$, as well as the maximal probability to reach $F$ in $\mathcal{M}$ as follows:%
\begin{align*}
	Pr^\Theta_\mathcal{M}(\lozenge F) \coloneqq \sum_{\pi \in \lozenge F} Pr^\Theta(\pi) 
	\qquad \qquad 
	Pr^{\mathit{max}}_\mathcal{M}(\lozenge F) \coloneqq \max\limits_{\Theta} Pr^\Theta_\mathcal{M}(\lozenge F) 
\end{align*}
Furthermore, we define the conditional reachability probability to reach a set of states $F$ in an MDP $\mathcal{M}$, under the condition of reaching a set of states $G$ for a given scheduler $\Theta$, as follows:
$$ Pr^{\Theta}_\mathcal{M} (\lozenge F \;|\; \lozenge G) \coloneqq \frac{Pr^{\Theta}_\mathcal{M} (\lozenge F \land \lozenge G)}{Pr^{\Theta}_\mathcal{M} (\lozenge G)} 
$$
We define the fraction to be $0$ in the case of $Pr^{\Theta}_\mathcal{M} (\lozenge G) = 0$.
There are two generic algorithms to compute conditional reachability probabilities on MDPs.
For one, \citet{baier2014computing} showed that any MDP $\mathcal{M}$ can be transformed into an MDP $\mathcal{M}'$ so that the conditional reachability probabilities on $\mathcal{M}$ align with the unconditional reachability probabilities on $\mathcal{M}'$.
Such ordinary reachability properties have been studied extensively, and a variety of algorithms to compute them have been developed \cite{hartmanns2023practitioner}.

Alternatively, the Storm Model Checker recently implemented  a bisection based algorithm to calculate conditional reachability probabilities \citeBisection. 
This algorithm is centered around a binary search for the correct conditional probability, which continues until sufficiently small error bounds on the results can be guaranteed.
To facilitate the binary search, an (unconditional) expected reward has to be calculated to decide whether the true conditional probability is in the upper or lower half of the current search space.
This expected reward can be calculated very efficiently on acyclic MDPs, which makes this approach highly suited for our eventual use-case.


\section{Compilation and Inference Algorithm}\label{sec:inference}
This section describes the inference algorithm for closed and well-formed \NDICE programs in several steps.
First, we describe how \NDICE expressions and programs are compiled into Boolean formulas. 
Afterwards, we demonstrate how these formulas can be used to represent the behavior of the compiled expression as an \textit{Algebraic Decision Diagram} (ADD).
We then lift this ADD to a \textit{Markov Decision Process} (MDP) and relate the semantics of the original \NDICE expression to conditional reachability properties on this MDP.

\subsection{Compilation of \NDICE Expressions to Boolean Formulas}\label{sec:nDice_comp}

	\noindent
	We will now present how the \DICE compilation procedure from \Cref{sec:prelims_Dice} can be extended to compile \NDICE expressions.
	In this extended compilation process, a \emph{compilation judgment} for a \NDICE expression $e$ has the form $\Gamma \vdash e \colon \tau \leadsto (\dot{\varphi},\, \gamma,\, t)$.  
	Here, the model formula $\dot{\varphi}$ and accepting formula $\gamma$ are unchanged compared to the $\DICE$ compilation.
	However, since nondeterministic behavior depends on the order of flip operations, an \textit{unordered} set of weights is no longer sufficient.
	We thus compile an annotated trace of flips $t$, which both tracks the order of the flips in $e$ and stores their behavior.
	To do so, $t$ is a list of pairs $f_i\colon a_i$, where $f_i$ is the Boolean variable which represents a flip and the label $a_i$ denotes the flip's behavior:
	We use $a_i = n$ if $f_i$ belongs to a nondeterministic flip, and $a_i = p$ if $f_i$ represents a probabilistic flip $\FLIP{p}$.
	
	The rules for \NDICE compilation of expressions can be found in \Cref{nDice Compilation}, and they are largely identical to the rules for \DICE~\cite{holtzen2020scaling} as seen in \Cref{fig:dice_Comp}.
	However, there are two notable changes:
	For one, there is an additional rule \texttt{nC-NFlip} to compile $\NFLIP$ statements. 
	This rule behaves fully analogous to the rule for probabilistic flips, but the boolean variable $f$ is now annotated with an $n$ to mark that it corresponds to a nondeterministic flip.
	Secondly, all rules now consider a trace of literals instead of a set of weights, which has little effect on most rules.
	However, when compiling a let-statement with the rule \texttt{nC-Let}, the trace of flips has to concatenate the trace for $e_1$ \textit{before} the trace for $e_2$ to capture the order of flip-operations.
	As for if-statements, the traces for the branches are combined by simply concatenating them, but any interleaving of the traces would be sound as long as the order of the individual sub-traces is preserved.

%
%
%
%
\begin{figure}
	\begin{center}
		\begin{adjustbox}{max width=\textwidth}
			\renewcommand{\arraystretch}{3.2}
			\begin{tabular}{c c}
				
				\inference[\ifthenelse{\boolean{showRuleNames}}{\texttt{(nC-Val)}}{}]
				{}{  \Gamma \vdash v \colon \tau \leadsto (v, \TRUE, [])}
				&
				
				\inference[\ifthenelse{\boolean{showRuleNames}}{\texttt{(nC-Ident)}}{}]
				{\Gamma(x) = \tau}{  \Gamma \vdash x \colon \tau \leadsto (F_\tau(x), \TRUE, [])}
				\\
				
				\inference[\ifthenelse{\boolean{showRuleNames}}{\texttt{(nC-Flip)}}{}]
				{\text{variable $f$ fresh}}{  \Gamma \vdash \FLIP{\theta} \colon \BOOL \leadsto (f, \TRUE, [f\colon \theta])}
				&
				
				\inference[\ifthenelse{\boolean{showRuleNames}}{\texttt{(nC-NFlip)}}{}]
				{\text{variable $f$ fresh}}{  \Gamma \vdash \NFLIP \colon \BOOL \leadsto (f, \TRUE, [f\colon n])}
				\\
				
				\inference[\ifthenelse{\boolean{showRuleNames}}{\texttt{(nC-Tup)}}{}]
				{\Gamma(x_1) = \tau_1 \qquad \Gamma(x_2) = \tau_2}{  \Gamma \vdash (x_1, x_2) \colon \tau_1 \times \tau_2 \leadsto ((F_{\tau_1}(x_1), F_{\tau_2}(x_2)) , \TRUE, [])}
				&
				
				\inference[\ifthenelse{\boolean{showRuleNames}}{\texttt{(nC-Obs)}}{}]
				{\Gamma \vdash \AEXP \colon \BOOL \leadsto (\varphi, \TRUE, [])}
				{  \Gamma \vdash \OBSERVE{\AEXP} \colon \BOOL \leadsto (\TRUE, \varphi, [])}
				\\
				
				\inference[\ifthenelse{\boolean{showRuleNames}}{\texttt{(nC-Fst)}}{}]
				{\Gamma(x) = \tau_1 \times \tau_2}
				{ \Gamma \vdash \FST x \colon \tau_1 \leadsto (F_{\tau_1}(x_l), \TRUE, [])}
				&
				
				\inference[\ifthenelse{\boolean{showRuleNames}}{\texttt{(nC-Snd)}}{}]
				{\Gamma(x) = \tau_1 \times \tau_2}
				{ \Gamma \vdash \SND x \colon \tau_1 \leadsto (F_{\tau_2}(x_r), \TRUE, [])}
				\\
				
				\multicolumn{2}{c}{
					\inference[\ifthenelse{\boolean{showRuleNames}}{\texttt{(nC-ITE)}}{}]
					{\Gamma \vdash \AEXP \colon \BOOL \leadsto (\varphi_g, \TRUE, []) 
						\qquad \Gamma \vdash e_T \colon \tau \leadsto (\dot{\varphi}_T, \gamma_T, t_T) 
						\qquad \Gamma \vdash e_F \colon \tau \leadsto (\dot{\varphi}_F, \gamma_F, t_F)}
					{  \Gamma \vdash \ITE{\AEXP}{e_T}{e_F} \colon \tau \leadsto \left( 
						(\varphi_g \overset{\tau}{\land} \dot{\varphi}_T) \overset{\tau}{\lor} (\overline{\varphi}_g \overset{\tau}{\land} \dot{\varphi}_F), \;
						(\varphi_g \land \gamma_T) \lor (\overline{\varphi}_g \land \gamma_F) ,\;
						t_T \,@\, t_F \right)}
				}\\
				
				\multicolumn{2}{c}{
					\inference[\ifthenelse{\boolean{showRuleNames}}{\texttt{(nC-Let)}}{}]
					{\Gamma \vdash e_1 \colon \tau_1 \leadsto (\dot{\varphi}_1, \gamma_1, t_1) \qquad 
						\Gamma \cup \{x \mapsto \tau_1 \} \vdash e_2 \colon \tau_2 \leadsto (\dot{\varphi}_2, \gamma_2, t_2)}
					{  \Gamma \vdash\LET{x = e_1}{e_2} \colon \tau_2 \leadsto \left(
						\dot{\varphi}_2[x\overset{\tau_1}{\mapsto} \dot{\varphi}_1], \;
						\gamma_1 \land \gamma_2[x\overset{\tau_1}{\mapsto} \dot{\varphi}_1],\;
						t_1 \,@\, t_2 \right)}
				}
				
			\end{tabular}
		\end{adjustbox}
	\end{center}
	\caption{\NDICE Compilation rules for expressions. These assume, without loss of generality but for simplicity, that fst, snd, and tuple construction are only ever performed with identifiers as arguments.}
	\label{nDice Compilation} 
\end{figure}%

\begin{example}
	As an example, consider the following expression and its compilation judgment:
	\begin{align*}
		\LET{x = \FLIP{\tfrac{2}{3}}}{\LET{y = \NFLIP}{\LET{z = \OBSERVE{x \lor y}}{(x \land y,\, y)}}}  \tag{ExComp}\label{ExComp} \\
		\{\}\vdash \text{\ref{ExComp}} \colon \BOOL \times \BOOL \leadsto \left((f_1 \land f_2,\, f_2),\; f_1 \lor f_2,\; \left[f_1\colon {\tfrac{2}{3}},\, f_2\colon n \right]\right)
	\end{align*}
	In the compilation result, $f_1$ represents the first flip, which is probabilistic, and $f_2$ represents the second flip, which is nondeterministic.
	This behavior (along with the order of the flips) is stored in the trace $\left[ f_1 \colon {\tfrac{2}{3}}, \; f_2\colon n\right]$.
	The model formula $(f_1 \land f_2,\, f_2)$ represents the output of the compiled expression, and each component of the model formula describes the behavior of the corresponding component in the output of \ref{ExComp}.
	Finally, the accepting formula $f_1 \lor f_2$ encodes the only observe statement in \ref{ExComp}, hence if $f_1 \lor f_2$ is satisfied, then all observations in \ref{ExComp} succeed.
	\hfill$\triangleleft$
\end{example}

\paragraph{Compilation of Functions and Programs}
We are not discussing the compilation rules for \NDICE programs at this point, since they are fully identical to the rules from \DICE~\cite{holtzen2020scaling}.
Furthermore, programs are still compiled to triples of the form $(\dot{\varphi}, \gamma, t)$, thus all subsequent inference steps are identical between the inference process for expressions and programs.
For completeness, rules for functions and programs are given in \ext{\Cref{app:function_comp}}{the extended version \citeExtended}.

\subsection{Algebraic Decision Diagrams}\label{sec:ndice_add}
This section presents how the previously compiled Boolean formulas can be used to build a decision diagram which compactly represents the compiled program.
Thereafter, we will enrich this decision diagram to an MDP on which we then do inference.

Consider a compilation judgment $\{\} \vdash e \colon \tau \leadsto (\dot{\varphi}, \gamma, t)$ for a closed \NDICE expression $e$. 
Our aim is now to use the compilation result $(\dot{\varphi}, \gamma, t)$ to define a function $(\dot{\varphi} \,|\, \gamma)_t\colon \BOOL^{len(t)}\rightarrow \llbracket\tau \rrbracket \, \cup \{R\}$, which mimics the behavior of $e$ for any given outcome of the flip operations in $e$. 
Here, the notation $\llbracket\tau \rrbracket$ refers to the set of all values of type $\tau$.

As a first observation, notice that if we assign either $\TRUE$ or $\FALSE$ to each flip operation in $e$,  then $e$ evaluates to a single value $v\in \llbracket\tau \rrbracket$ (assuming all observations in $e$ would succeed).
Similarly, we can assign $\TRUE$ or $\FALSE$ to each variable in $\dot{\varphi}$ to evaluate $\dot{\varphi}$ to a value $v\in \llbracket\tau \rrbracket$.  
Notably, for closed expressions each variable in $\dot{\varphi}$ and $\gamma$ corresponds to a flip statement in $e$, and hence each assignment to flips in $e$ defines an assignment to the variables in  $\dot{\varphi}$ and $\gamma$ and vice versa.
Now, given any assignment to flip operations in $e$, we thus want the function $(\dot{\varphi} \,|\, \gamma)_t$ to return the same value $v\in \llbracket \tau \rrbracket$ as $e$ if all observations in $e$ would succeed, and otherwise $(\dot{\varphi} \,|\, \gamma)_t$ should return $R$ to denote a "rejected" run.

Second of all, the function $(\dot{\varphi} \,|\, \gamma)_t$ takes a vector of Boolean values $(b_1,\dots,b_k) \colon \BOOL^k$ with $k=len(t)$ as input.
Thus, we need to define how $\gamma$ and $\dot{\varphi}$ are evaluated with a vector $(b_1,\dots,b_k)$ before we define the function $(\dot{\varphi} \,|\, \gamma)_t$.
In particular, we need to decide which variable in $\gamma$ and $\dot{\varphi}$ is substituted with which Boolean value $b_i$.
To this end, we use the trace of flips $t$ and substitute the $i$-th flip $f_i$ from $t$ with $b_i$, written as $\gamma \statesubst{f_i}{b_i}$.
The Boolean formulas resulting from these substitutions are ground, and can be evaluated to a Boolean value.
Thus, the Boolean formula $\gamma$ represents a function of type $\BOOL^k\rightarrow \BOOL$, and the tuple of Boolean formulas $\dot{\varphi}$ can be evaluated recursively as a function of type $\BOOL^k\rightarrow \tau$:
\begin{align*}
\gamma (b_1,\dots,b_k) &\coloneqq \gamma\statesubst{f_1}{b_1}\dots\statesubst{f_k}{b_k}\\
 \left(\dot{\varphi}_1,\; \dot{\varphi}_2\right)(b_1,\dots,b_k) &\coloneqq \left(\dot{\varphi}_1 (b_1,\dots,b_k), \;\dot{\varphi}_2(b_1,\dots,b_k)\right)
\end{align*}
We can now combine the two formulas into a function $(\dot{\varphi} \,|\, \gamma)_t \colon \BOOL^k\rightarrow \tau\, \cup \{R\}$ by mimicking the behavior of \NDICE expressions:
If $\gamma$ returns \TRUE, then all observations succeeded and we return the output of $\dot{\varphi}$.
If $\gamma$ returns \FALSE, then an observation failed and the run will be rejected, so we return the special value $R$.
Formally, we define this as:
$$ (\dot{\varphi} \,|\, \gamma)_t (b_1,\dots,b_k)\coloneqq \begin{cases}
	\dot{\varphi} (b_1,\dots,b_k), & \text{if } \gamma(b_1,\dots,b_k) = \TRUE\\
	R, & \text{if } \gamma(b_1,\dots,b_k) = \FALSE
\end{cases}$$%
Once this \enquote{guarding} operator $(\cdot \,|\, \cdot)$ is defined, the ADD for $(\dot{\varphi} \,|\, \gamma)_t$ can be constructed from decision diagrams for $\dot{\varphi}$ and $\gamma$ using generic algorithms \cite{fujita1997multi}.
Thus, we can represent $(\dot{\varphi} \,|\, \gamma)_t$ as an ADD, and if an ADD 
uses the variable order defined by $t$ and represents $(\dot{\varphi} \,|\, \gamma)_t$, then we call it an ADD for the compiled formulas $(\dot{\varphi}, \gamma, t)$.
If an ADD for $(\dot{\varphi}, \gamma, t)$ is a reduced ADD, then we call it the reduced ADD for $(\dot{\varphi}, \gamma, t)$.

\begin{remark}[Permissible Variable Orderings]
	As mentioned in \Cref{sec:prelims_ADD}, the variable ordering which is used during construction has a stark effect on the size of the ADD.
	This appears to call for experiments with different variable orderings to speed up inference, as was done for \DICE \cite{cheng2021flip}.
	However, changing the variable order amounts to changing the order of flips, which can alter the semantics of a \NDICE expression.
	In fact, variable orders have to adhere to the natural program order to be sound for inference, and the trace $t$ provides one such order. 
	\hfill $\triangleleft$
\end{remark}

\subsection{Lifting ADDs to Markov Decision Processes}\label{sec:MDP_inf}

Given a \NDICE expression $e$ which compiles to a triple $(\dot{\varphi}, \gamma, t)$, the last part of the inference process is to lift the previously constructed ADD for $(\dot{\varphi}, \gamma, t)$ to an MDP by re-introducing the behavior of the individual flip operations into the ADD.
To do so, we consider the variable stored in an inner ADD node, and determine the behavior of the corresponding flip operation by retrieving the annotated variable from the trace $t$.
The ADD node will then be modeled by an MDP state, which behaves either probabilistic or nondeterministic depending on the nature of the flip operation.
Afterwards, the reachability probabilities in this lifted MDP will align with the probabilities that the compiled expression $e$ returns a given value.

We define the \textit{unnormalized} MDP $\mathcal{M}^u = (S, s_{init}, \mathit{Act}, P, AP, L)$ lifted from an ADD $\mathcal{A} = (V, T, r, E)$ for $(\dot{\varphi}, \gamma, t)$ as follows: 
We define $S = V \cup T$ as the set of nodes of the ADD, $s_{init}$ as the root $r$ of the ADD, $\mathit{Act} = \{l, r, d\}$ and we define the transition relation as follows:
\begin{itemize}
	\item If $s \in T$ was lifted from a terminal node, then we only have the transition $s\xrightarrow{d, 1} s$
	\item If $s\in V$ was lifted from an inner node with probabilistic variable $f$ with $(f\colon \theta) \in t$, then we have the transitions $s\xrightarrow{d, \theta} \mathit{then}(s)$ and $s\xrightarrow{d, 1 - \theta} \mathit{else}(s)$
	\item If $s\in V$ was lifted from an inner node with nondeterministic variable $f$ with $(f\colon n) \in t$, then we have the transitions $s\xrightarrow{l, 1} \mathit{then}(s)$ and $s\xrightarrow{r, 1} \mathit{else}(s)$
\end{itemize}
Above, we re-use the notations $\mathit{then}(s)$ and $\mathit{else}(s)$ to refer to the MDP state lifted from the respective ADD successor node.
We define $AP = \llbracket\tau \rrbracket \cup \{A, R\}$, and we define the labeling function $L$ as follows:
$$L(s) = \begin{cases}
	\{v, A\} &, \text{$s$ lifted from a terminal node with value $\mathit{val}(s) = v \neq R$}\\
	\{R\} &, \text{$s$ lifted from a terminal node with value $\mathit{val}(s) = R$}\\
	\emptyset & \text{otherwise}
\end{cases}$$
Here, states which were lifted from terminal nodes for values $v\in \llbracket\tau \rrbracket$ are labeled with their corresponding output in the ADD and the proposition $A$ to model an accepted run.
States lifted from terminal nodes containing the value $R$ are labeled with the proposition $R$.

Now consider some \NDICE expression $e$ with $\{\} \vdash e \colon \tau \leadsto (\dot{\varphi}, \gamma, t)$, as well as the unnormalized MDP $\mathcal{M}^u_e$ lifted from an ADD for $(\dot{\varphi}, \gamma, t)$.
One can prove that the \textit{unnormalized} semantics of $e$ are closely related to reachability probabilities in $\mathcal{M}^u_e$.
However, the central inference question for PPLs is to infer the behavior of the program under the condition that all observations succeeded, which corresponds to the \textit{normalized} semantics of $e$.
Hence, we are not interested in the likelihood to just reach a particular value $v$, but we are instead interested in the likelihood to reach a particular $v$ under the condition that all observations pass, as modeled by reaching any state labeled with $A$.

We are thus interested in the conditional reachability probability $Pr^{max}_{\mathcal{M}^u_e} (\lozenge v \;|\; \lozenge A)$.
As discussed at the end of \Cref{sec:prelims_MDP}, there are two approaches to determine these conditional probabilities:
The reduction to ordinary reachability properties by \citet{baier2014computing}, and the recent, bisection-based approach implemented by the Storm Model Checker \citeBisection.
We will be using the bisection-based algorithm, as it is highly suited for our specific use-case.
Namely, the algorithm is very efficient on acyclic MDPs, as our MDPs naturally are.
Furthermore, the runtime of the algorithm is agnostic to the probability of an accepting run, hence the runtime of the inference process will be unaffected by conditioning with rare events.
This is a desirable trait, since inference engines commonly struggle to infer probabilities which are conditioned by low-probability observations.

 \subsection{Soundness}
 It is now time to formally consider the connection between the semantics of a $\NDICE$ program $p$ and the unnormalized MDP lifted from its compiled Boolean formulas.
 In practice it is infeasible to infer all distributions in $\llbracket p \rrbracket_D$, and thus we provide an upper bound for their behavior instead.
 Formally, we show that the maximal probability to terminate with a value $v$ among all distributions in $\llbracket p \rrbracket_D$ is equal to the maximal conditional probability to reach a state labeled with $v$ in the unnormalized MDP lifted from the compiled formulas of $e$:

\begin{theorem}[Inference Soundness]\label{ProgramSoundness}
	Let $p$ be a \NDICE program  
	and let $\{\}, \{\} \vdash p \colon \tau \leadsto (\dot{\varphi}, \gamma, t)$. 
	Then the reduced ADD for $(\dot{\varphi}, \gamma, t)$ can be lifted to an unnormalized MDP $\mathcal{M}_p$ so that:
	$$ Pr^{max}_{M_p} (\lozenge v \,|\, \lozenge A) = \max \left\{d(v) \,|\, d \in \llbracket p\rrbracket^{\{\}}_D\right\}$$
\end{theorem}
\begin{proof}
	 \ext
	 {See \Cref{ProgramAppendix}, where the theorem is restated as \Cref{app:ProgramSoundness}.}
	 {We refer to the extended version \citeExtended.}
\end{proof}
\noindent%
The theorem above refers to the \textit{reduced} ADD for the formulas, meaning we may use the smallest ADD.
Still, the construction of this ADD is PSPACE-hard \cite{holtzen2020scaling}, and the final ADD itself may be exponentially large in the number of flip operations in the expression.
Additionally, we then still need to lift the ADD to an MDP and determine the necessary reachability probabilities.
Lifting the ADD to an MDP can be done in a single pass, while determining conditional reachability probabilities takes polynomial time relative to the MDP's size \cite{puterman2014markov, baier2014computing}.

On paper, this makes \NDICE inference very expensive.
However, \citet{holtzen2020scaling} show how the \DICE compilation naturally exploits structural properties like conditional independence, determinism or context-specific independence to build compact BDDs quickly.
\NDICE reuses the compilation of \DICE, and thus \NDICE exploits the same program structures for compact ADDs.
While we refer to \citet{holtzen2020scaling} for details on structures exploited by \DICE, we still want to briefly review one central source of efficiency for \DICE and \NDICE: Conditional independence.

Intuitively, variables $x$ and $y$ are conditionally independent given $z$, if the behavior of $x$ does not depend on $y$ if a value for $z$ has been fixed.
Programs contain conditional independence naturally since it arises from constructs like \texttt{let} statements or function calls, as seen in the introductory program from \Cref{fig:plane}:
Here, all calls to \texttt{step} are conditionally independent of each other given their return values.
After all, subsequent function calls and the behavior of the program only depend on the functions output.
Hence, all intermediate values, like e.g. the flip operations in the \texttt{move} function, are conditionally independent of each other given the output of \texttt{step}.
This conditional independence prevents an exponential blow-up during the ADD construction, as each function call only adds another conditionally independent \enquote{layer} of nodes.


\newboolean{useComma}
\setboolean{useComma}{true}

\section{Implementation and Empirical Evaluation}\label{sec:implementation}
In this section, we present a prototypical implementation of our inference engine and evaluate its efficiency.
The implementation is written in Ocaml and largely based on the existing \DICE library~\cite{holtzen2020scaling}, thus inheriting many of its features:
For instance, \NDICE inherits \DICE's relaxation of the A-normal form as well as its support for standard Boolean connectives and bounded numbers.
As the semantics of nondeterministic expressions depend on the order of evaluation, we define tuple elements and arguments of function calls to be evaluated from left to right.
By now, \DICE also supports bounded lists as well as an improved representation of numbers \cite{cao2023scaling}. 
We further provide a \texttt{choose}-statement to model the nondeterministic choice of a number from an interval.
Internally, all statements which choose or sample numbers are syntactic sugar for a sequence of flip operations and nested if-statements.
\DICE also switched the back-end for BDD operations from CUDD~\cite{somenzi2009cudd} to RSDD~\cite{RSDD}.
Since RSDD does not support ADDs, we implemented the ADD construction without a specialized library.
We further directly implemented Storm's bisection based algorithm \citeBisection, and use it to determine conditional reachability probabilities on MDPs to an accuracy of $10^{-6}$. 
We also use two optimizations\ext{}{, whose soundness is shown in the extended version \citeExtended}:

 \paragraph{Reduction to Boolean Programs}
 Consider some \NDICE expression $e$ of type $\BOOL \times \BOOL$, where we would like to infer the maximum probability to return $(\TRUE, \TRUE)$.
 Here, the standard construction creates an ADD with five output nodes: One for each value and one for the special value $R$.
 However, this construction creates unnecessary states in the MDP, as even though we care about the probability to reach $(\TRUE, \TRUE)$, additional states are included to distinguish whether the result will be $(\TRUE, \FALSE)$ or $(\FALSE, \TRUE)$.
 We can remove these states by instead determining the maximum probability that the expression $\LET{x = e}{x \leftrightarrow (\TRUE, \TRUE)}$ returns $\TRUE$.
 \ext{The soundness of this transformation is proven in \Cref{app:boolean_reduction}.}{}

 \paragraph{Probabilistic State Compression}
\begin{wrapfigure}[8]{r}{0.45\textwidth}
	\vspace*{-1\intextsep}%
	\begin{center}
		\begin{tikzpicture}[yscale=1, xscale=0.9]
			\node (i) at (0,0)   [draw, circle] {$r$};
			\node (s) at (-0.9,-1)   [draw, circle] {$s$};
			\node (dots) at (1,-1)   {$\dots\;\;$};
			\node (t) at (-1.8,-2.25)   [draw, circle] {$t$};
			\node (u) at (-0,-2.25)   [draw, circle] {$u$};
			
			\node (ro) at (0, -0.65)  [circle,fill,inner sep=1.5pt] {};
			\path[->]          (i)  edge   [bend right=0]   node[right, pos=0.5] {$\alpha$} (ro);
			\path[->]          (ro)  edge   [bend right=15]   node[above, pos=0.75] {$q$} (s);
			\path[->]          (ro)  edge   [bend right=-15]   node[above, pos=0.5] {} (dots);

			\node (so) at (-0.9, -1.7)  [circle,fill,inner sep=1.5pt] {};
			\path[->]          (s)  edge   [bend right=0]   node[left, pos=0.5] {$d$} (so);
			\path[->]          (so)  edge   [bend right=20]   node[left, pos=0.4] {$p\;$} (t);
			\path[->]          (so)  edge   [bend right=-20]   node[right, pos=0.4] {$\,1 - p$} (u);
			
			

			\def\offset{3};
			
			\node (i) at (0 + \offset,0)   [draw, circle] {$r$};
			\node (dots) at (0.75 + \offset,-1)   {$\dots$};
			\node (t) at (-0.95 + \offset,-2.25)   [draw, circle] {$t$};
			\node (u) at (0.0 + \offset,-2.25)   [draw, circle] {$u$};
			
			\node (io) at (0 + \offset, -0.75)  [circle,fill,inner sep=1.5pt] {};
			\path[->]          (i)  edge   [bend right=0]   node[left, pos=0.5] {$\alpha$} (io);
			\path[->]          (io)  edge   [bend right=30]   node[left, pos=0.5] {$q \cdot p$} (t);
			\path[->]          (io)  edge   [bend right=0]   node[right, pos=0.65] {$q \cdot (1-p)$} (u);
			\path[->]          (io)  edge   [bend right=-10]   node[above, pos=0.5] {} (dots);
			
		\end{tikzpicture}
	\end{center}
\end{wrapfigure}
Intuitively, this removes states with a single enabled action by forwarding all incoming edges:
In the example on the right, state $s$ has a single enabled action $d$. 
Whenever we transit from $r$ via $\alpha$ to $s$ with probability $q$, we must subsequently transit via $d$ to either $t$ or $u$ with probabilities $p$  and $1-p$. 
Thus, we can remove $s$ from the MDP entirely and let $r$ transition via $\alpha$ directly to $t$ or $u$ with probabilities $q \cdot p$ and $q \cdot (1-q)$.
State compression can drastically reduce the size of MDPs in our use-case, as every state which represents a probabilistic flip can be removed.
\ext{We describe state compression formally in \Cref{app:compression}.}{}

\subsection{Empirical Performance Evaluation}\label{sec:ndet_empirical}
We will now perform an in-depth empirical evaluation aimed at investigating the overall efficiency and scalability of our inference method. 
To do so, we will consider a wide range of benchmarks of various sizes and with varying degrees of nondeterminism.
Regarding the evaluation itself, we will focus on answering the following questions:
\begin{itemize}
	\item[\textbf{Q1}:] \textbf{Scaling Behavior of the Compilation.}
	How quickly can the ADDs for \NDICE be compiled, and how does the needed time compare to the original compilation process of \DICE?
	\item[\textbf{Q2}:] \textbf{Reduction in the MDP State Space.}
	How much smaller are the MDPs constructed by \NDICE compared to a direct encoding of the same problem as an MDP? 
	\item[\textbf{Q3}:] \textbf{Inference Efficiency.} 
	How efficient is \NDICE inference compared to directly applying a probabilistic model checker?
	\item[\textbf{Q4}:] \textbf{Influence of Nondeterminism.}
	How does the amount of nondeterminism in a probabilistic program affect the runtime of \NDICE?
\end{itemize}
\noindent
For \textbf{Q1}, we will compare \NDICE with \DICE on purely probabilistic benchmarks.
To answer \textbf{Q2} and \textbf{Q3}, we encode each probabilistic program directly as an MDP, and we use Storm~\cite{hensel2022probabilistic} to infer the probabilities of the possible program outputs.
Storm offers multiple \enquote{engines} to internally represent MDPs, as well as different algorithms to perform model checking for conditional reachability probabilities.
The efficiency of model checking can heavily depend on the used configuration, and we made a best-effort attempt to identify and use the best combination.
As for benchmarks, PPLs typically do not consider nondeterminism -- existing benchmarks hence do not feature nondeterminism.
Thus, we introduce the following benchmark classes:

\begin{table}
	\caption{Performance of \NDICE on the benchmarks.
		The \DICE, \NDICE and Storm columns report the total runtime of the respective algorithm, \NDICE is run with state compression and split into \textit{compilation} and \textit{model checking}.
		$\NDICE_u$ reports the total inference time without state compression.
		All times are in ms, with a timeout (TO) after 20 minutes.
		\DICE runtimes are only reported for purely probabilistic benchmarks.
		The size of the MDP used by Storm is reported as $\text{MDP}_S$, whereas ADD Size and $\text{MDP}_n$ Size report the total size of the constructed ADDs and compressed MDPs for \NDICE.}
		
	\label{tab:Ndet}
	\begin{adjustbox}{max width=\textwidth}
	\begin{tabu}{c | c | ccc | c | c | c c c}
		\toprule
		Benchmark & \DICE & $\textit{comp.}$ & $\textit{checking}$ & \NDICE & $\NDICE_u$ & Storm & $\text{MDP}_S$ Size & ADD Size & $\text{MDP}_n$ Size \\
		\midrule
		Runway-7 & - & $8$ & $\approx 0$ & {$8$} & {$8$} & $18$ & $1.2 \times 10^3$ & $284$ & $40$ \\
		Runway-15 & - & $93$ & {$1 $} & {$94 $} & {$94 $} & {$131$} & $4.3 \times 10^4$ & $1.4\times 10^3 $ & $274$ \\
		Runway-30 & - & $531$ & {$7 $} & {$538 $} & {$538$ } & {$1\,620$} & $3.2 \times 10^5$ & $1.3\times 10^4 $ & $1.0\times 10^3$ \\
		Runway-40 & - & $999$ & {$11 $} & {$1\,010$} & {$1\,011 $} & {$5\,212$} & $7.5 \times 10^5$ & $2.4\times 10^4 $ & $1.8\times 10^3$ \\
		Runway-45 & - & $1\,514$ & {$15$} & {$1\,529$} & {$1\,530 $} & {$8\,673$} & $1.0 \times 10^6$ & $3.0 \times 10^4 $ & $2.2\times 10^3$ \\
				
		\midrule
		CouponProb-10 & 
				$721$ &
		$1\,379$ & {$\approx 0$} & {$1\,379$} & 
		{$1\,379$} & $5\,894$ &
		 $1.2 \times 10^6$ & $3.4\times10^5$ & $8$ \\
		CouponProb-11 & 
				$1\,958$ & 
		$3\,904$ & {$\approx 0$} & {$3\,904$} & 
		{$3\,904$} & $18\,839$ &
		 $3.4 \times 10^6$ & $8.3\times10^5$ & $8$ \\
		CouponProb-12 & 
				$5\,798$ & 
		$10\,484$ & {$\approx 0$} & {$10\,484$} & 
		{$10\,484$} & $56\,434$ &
		 $8.7 \times 10^6$ & $2.0\times10^6$ & $8$ \\
		CouponProb-13 & 
				$20\,691$ & 
		$32\,153$ & {$\approx 0$} & {$32\,153$} & 
		{$32\,153$} & $171\,358$ &
		 $2.2 \times 10^7$ & $4.8\times10^6$ & $8$ \\
		CouponProb-14 & 
				$63\,418$ & 
		$96\,068$ & {$\approx 0$} & {$96\,068$} & 
		{$96\,068$} & $641\,614$ &
		 $5.5 \times 10^7$ & $1.1\times10^7$ & $8$ \\
		
		\midrule
		CouponNdet-6 & - & $105$ &  {$14$} & {$119$} & {$120$} & $162$ & $6.9 \times 10^4$ & $1.9\times10^4 $ & $818$ \\
		CouponNdet-7 & - & $276$ & {$41$} & {$317$} & {$317$} & $607$ & $2.1 \times 10^5$ & $5.6\times10^4 $ & $1.9\times10^3$ \\
		CouponNdet-8 & - & $731$ & {$96$} & {$827$} & {$824$} & $2\,208$ & $6.3 \times 10^5$ & $1.5\times10^5 $ & $5.2\times10^3$ \\
		CouponNdet-9 & - & $2\,264$ & {$389$} & {$2\,653$} & {$2\,662$} & $8\,110$ & $1.7 \times 10^6$ & $4.3\times10^5 $ & $1.4\times10^4$ \\
		CouponNdet-10 & - & $6\,105$ & {$1\,257$} & {$7\,362$} & {$7\,293$} & $25\,862$ & $4.8 \times 10^6$ & $1.0\times10^6 $ & $3.9\times10^4$ \\
		
		\midrule
		Network-2000 & - & $2\,218$ & {$24$} & {$2\,242$} & {$2\,242$} & $25$ & $1.2 \times 10^4$ & $2.4\times10^4$ & $8.0\times10^3$ \\
		Network-4000 & - & $9\,530$ & {$55$} & {$9\,585$} & {$9\,585$} & $39$ & $2.4 \times 10^4$ & $4.8\times10^4$ & $1.6\times10^4$ \\
		Network-6000 & - & $22\,270$ & {$84$} & {$22\,354$} & {$22\,358$} & $54$ & $3.6 \times 10^4$ & $7.2\times10^4$ & $2.4\times10^4$ \\
		Network-8000 & - & $40\,763$ & {$122$} & {$40\,885$} & {$40\,885$} & $68$ & $4.8 \times 10^4$ & $9.6\times10^4$ & $3.2\times10^4$ \\
		Network-10000 & - & $67\,081$ & {$153$} & {$67\,234$} & {$67\,243$} & $84$ & $6.0 \times 10^4$ & $1.2\times10^5$ & $4.0\times10^4$ \\
		
		\midrule
		Rabin-50 & - & $324$ & {$\approx 0$} & {$ 324$} & {$ 324 $} & $17$ & $5.2 \times 10^3 $ & $524 $ & $206 $ \\
		Rabin-100 & - & $1\,962 $ & {$1$} & {$ 1\,963 $} & {$ 1\,963 $} & $23$ & $2.0 \times 10^4$  & $1.0 \times 10^3 $ & $406 $ \\
		{Rabin-250} & - & $23\,355 $ & {$2$} & {$ 23\,357 $} & {$ 23\,357 $} & $82$ & $1.2 \times 10^5$ & $2.8 \times 10^3 $ & $1.0 \times 10^3$ \\
		{Rabin-500} & - & $233\,494 $ & {$6$} & {$233\,500$} & {$233\,501$} & $327$ & $5.0 \times 10^5$ & $5.8 \times 10^3 $ & $2.0 \times 10^3 $ \\
		
		\midrule 
		Survey & $1$ & $1$ & $\approx 0$ & $1$ & $1$ & $14$ & $238$ & $71$ & $14$ \\
		Hepar2 & $99$ & $309$ & $\approx 0$ & $309$ & $310$ & $23\,094$ & $1.8 \times 10^{19} $ & $1.8 \times 10^4$ & $8$ \\
		Insurance & $580$ & $2\,227$ & $\approx 0$ & $2\,227$ & $2\,232$ & $9\,460$ & $4.2 \times 10^{11}$ & $1.5 \times 10^5$ & $16$ \\
		Water & $9\,241$ & $9\,288$ & $\approx 0$ & $9\,288$  & $9\,347$   & $12\,545$ & $2.4 \times 10^{9} $ & $5.2 \times 10^4$ & $16$ \\
		Alarm & $160\,104$ & $161\,385$ & $\approx 0$ & $161\,385$ & $161\,385$ & $5\,335$ & $1.8 \times 10^{16}$ & $1.2 \times 10^5$ & $14$ \\
		
		\midrule 
		Survey-Ndet & - & $1$ & $\approx 0$ & $1$ & $1$ & $14$ & $238$ & $71$ & $26$ \\
		Hepar2-Ndet & - & $280$ & $14$ & $294$ & $295$ & $23\,909$ & $1.8 \times 10^{19} $  & $1.8 \times 10^4$ & $2.0 \times 10^3$ \\
		Insurance-Ndet & - & $2\,217$ & $\approx 0$ & $2\,217$ & $2\,220$ & $10\,345$ & $4.2 \times 10^{11}$ & $1.5 \times 10^5$ & $56$ \\
		Water-Ndet & - & $9\,296$ & $\approx 0$ & $9\,296$ & $9\,298$ & $13\,153$ & $2.4 \times 10^{9} $ & $5.2 \times 10^4$ & $56$ \\
		Alarm-Ndet & - & $161\,786$ & $47$ & $161\,633$ & $161\,636$ & $5\,219$ & $1.8 \times 10^{16}$ & $1.2 \times 10^5$ & $6.9 \times 10^3$ \\

		\midrule
		3sat-0 & $19\,511$ & $19\,535$  & {$\approx 0$} & {$19\,535$} & {$19\,535$} & TO & $1.2 \times10^{13}$ & $4.7\times10^4$ & $8$ \\
		3sat-25 & - & $19\,528$ & {$43$} & {$19\,571$} & {$19\,571$} & TO & $1.2 \times10^{13}$ & $4.7\times10^4$ & $1.2\times10^4$ \\
		3sat-50 & - & $19\,544$ & {$69$} & {$19\,613$} & {$19\,612$} & TO & {$1.2 \times10^{13}$} & $4.7\times10^4$ & $2.3\times10^4$ \\
		3sat-75 & - & $19\,514$ & {$72$} & {$19\,586$} & {$19\,584$} & TO & {$1.2 \times10^{13}$} & $4.7\times10^4$ & $3.6\times10^4$ \\
		3sat-100 & - & $19\,567$ & {$107$} & {$19\,674$} & {$19\,676$} & TO & {$1.2 \times10^{13}$} & $4.7\times10^4$ & $4.7\times10^4$ \\
		\bottomrule
	\end{tabu}
	\end{adjustbox}
\end{table}

\paragraph{Runway}
We consider the adapted vehicle tracking problem from \citet{junges2021runtime} as presented in \Cref{sec:overview}, but with more locations for the ground vehicle, more measurements, and a vehicle which may move up to two locations per time step.
In particular, the benchmark \textit{runway-n} refers to the scenario with $n$ possible locations for the ground vehicle and $n$ measurements of the vehicle position.
In this class of benchmarks, roughly $10\%$ of all flip operations are nondeterministic.

\paragraph{Coupon Collector}
We consider a variation of the Coupon Collector Problem inspired by \citet{jansen2016bounded}:
A collector aims to collect $N$ distinct coupons within $N$ rounds.
In every round, the collector may draw \emph{two} coupons from a bowl.
An observe statement enforces that the two coupons are distinct.
We consider two instances of this problem: One where the coupons are always drawn from a uniform distribution (\textit{CouponProb}), and one where we nondeterministically choose between three bowls with different distributions each round (\textit{CouponNdet}).
\textit{CouponProb} is purely probabilistic, whereas for \textit{CouponNdet} roughly $14\%$ of flip operations are nondeterministic. 

\paragraph{Network Reachability} 
We consider a network reachability problem adapted by \citet{holtzen2020scaling} from \citet{gehr2018bayonet}, and then further extended with nondeterminism by us.
Here, server~$S_1$ forwards an incoming packet to either $S_2$ or $S_3$, which in turn forward the packets to $S_4$, but $S_3$ drops packets with probability $0.1\%$.
For load balancing, $S_1$ chooses randomly between $S_2$ or $S_3$, and nondeterministically between three protocols.
This benchmark is scaled up by repeating this forwarding process, and roughly $25\%$ of all flips are nondeterministic.

\paragraph{Fair Exchange}
We consider a case study of the Prism Model Checker \cite{KNP11}, namely the probabilistic fair exchange protocol from \citet{rabin1983transaction}.
The protocol allows two parties to exchange commitments to a pre-defined contract with the aid of a trusted third party.
Nondeterminism is used to model a bad actor, and we infer the maximum probability that one party commits to the contract while the bad actor doesn't.
Here, over $95\%$ of flip operations are nondeterministic.

\paragraph{Bayesian Networks}
Similar to \DICE~\cite{holtzen2020scaling}, we evaluate \NDICE on some well-known Bayesian networks from the the Bayesian Network Repository\footnote{\url{https://www.bnlearn.com/bnrepository/}}, namely \textit{Survey}, \textit{Hepar2}, \textit{Insurance}, \textit{Water} and \textit{Alarm}.
We consider two instances for each network: A purely probabilistic version which directly encodes the Bayesian network as a \NDICE program, and a version were some parent-less nodes behave nondeterministically.
Here, nondeterminism can model behaviors which may be influenced by non-probabilistic factors.
For example, a  patient may lie about their alcoholism, and a network which models their diagnosis can reflect this uncertainty with nondeterminism.
The amount of nondeterminism varies heavily between networks, ranging from $5\%$ of flips being nondeterministic for \textit{Insurance-Ndet} to up to $30\%$ of flips being nondeterministic for \textit{Alarm-Ndet}.

\paragraph{3SAT} 
We consider a 3SAT problem based on a benchmark retrieved from \citet{hoos2000satlib}
with $40$ variables and $114$ clauses encoded as a \NDICE program.
This benchmark is not scaled by increasing the program size, but instead we increase the amount of nondeterministic variables. 
For example, \textit{3sat-25} is an instance where $25\%$ of the $40$ variables are assigned a value nondeterministically and the remaining variables randomly.
\medskip

The performance of \DICE, \NDICE and Storm on these benchmarks is presented in \Cref{tab:Ndet}.
The considered inference task was to determine the maximum probability for all outputs of the input program.
\NDICE is run with state compression and split into $\mathit{comp.}$ for the construction of the MDPs and \textit{checking} to determine conditional reachability probabilities.
We report the total runtime without state compression as $\NDICE_u$.
We also list the number of states in the MDP used by Storm as \enquote{$\text{MDP}_S$ Size}, as well as the total size of the ADDs and the total number of MDP states \emph{after} state compression as used by \NDICE.
All experiments were run on macOS Sonoma 14.5 with a 3.49 GHz Apple M2 Pro chip and 32~GB RAM.
Benchmarks were run 5 times and \Cref{tab:Ndet} reports the average runtime; standard deviations are reported in \ext{\Cref{app:extended_results}}{the extended version \citeExtended}.

For Storm, we used the following engines: The benchmarks classes \textit{Runway}, \textit{CouponProb} and \textit{CouponNdet} required conditional reachability properties, which are only supported by the \enquote{sparse} engine.	
Here, the \textit{Coupon} benchmarks performed best with the \enquote{restart} algorithm by \citet{baier2014computing}, whereas the \textit{Runway} benchmarks used Storm's custom \enquote{bisection} algorithm to determine conditional probabilities. 
Beyond that, the \textit{Network} and \textit{Rabin} benchmarks as well as the \textit{Survey} Bayesian network performed best with the \enquote{sparse} engine, whereas the remaining benchmarks preferred the symbolic \enquote{dd} engine.
We now turn our attention to the previously posed questions:

\subsubsection*{Q1: Scaling Behavior of the Compilation.}
To answer this question, we will primarily consider the performance of \DICE and \NDICE on the probabilistic Coupon Collector and Bayesian network problems. 
We observe that the compilation time of \NDICE is slower than \DICE, likely because \NDICE does not use a dedicated decision diagram library to construct ADDs.
Still, the compilation time of noDice scales similarly to Dice on the Coupon Collector benchmarks, where \NDICE takes about twice as long as the more specialized \DICE, and their compilation times are comparable on most Bayesian networks.
Since the model checking times on these benchmarks are negligible for noDice, we can conclude that  \NDICE is slower than \DICE on purely probabilistic programs, but both approaches scale similarly for larger programs.

\subsubsection*{Q2: Reduction in the MDP State Space.}
Overall, \NDICE is consistently able to construct significantly smaller MDPs than the direct encoding for the same benchmark problems.
More precisely, the \DICE-style compilation is able to avoid an exponential blow-up of the ADDs by exploiting local structure in the programs, leading to the ADDs themselves usually already being much smaller than the MDPs build by Storm.
The \textit{Network} benchmarks form an exception to this, since \NDICE requires more MDP states than Storm to infer all values.
After all, these benchmarks contain no redundancy which is not already exploited by Storm.
Thus, both \NDICE and Storm build MDPs of the same size, but Storm constructs one MDP which is used to infer the probabilities for both output values, whereas \NDICE builds separate MDPs for the two values, leading to a total state space which is twice as large.
Beyond that, state compression is able to reduce the state space of the MDPs even further depending on the amount of nondeterminism in the benchmark.
In particular, state compression yields larger reductions on problems with less nondeterminism.

\subsubsection*{Q3: Inference Efficiency.}
	Looking at \NDICE, we observe that compilation takes the vast majority of the runtime, whereas model checking takes little to no time in comparison.
	This behavior is consistent with \DICE, which also spends most of the runtime on the compilation, but performs the weighted model count, meaning \DICE's equivalent to model checking, within milliseconds.
	For both \NDICE and \DICE, quick model checking is the result of the large reduction in the state space due to the use of decision diagrams in the compilation process.
	However, the further reduction of the MDP size with state compression did not speed up \NDICE inference in a meaningful way, since the runtimes of \NDICE and $\NDICE_u$ are almost identical.
	
	We will now compare the performance of \NDICE and Storm.
	 We first note that \NDICE outperforms Storm on several problems, namely the \textit{Runway}, \textit{Coupon Collector} and \textit{3sat} benchmarks along with most Bayesian networks.
	On these problems, the small MDPs of \NDICE allow for highly efficient model checking, whereas Storm constructs much larger MDPs which are expensive to construct and model check.
	Still, \NDICE's compilation is not always worth it, as seen with the \textit{Network} and \textit{Rabin} benchmarks:
	Here, the compilation of a small MDP by \NDICE is far too costly compared to Storm, which solves these problems in milliseconds.
	
	This raises the question as to when \NDICE is faster than simply using Storm.
	Judging from the results in \Cref{tab:Ndet}, we believe that a key distinguishing characteristic between the performance of Storm and \NDICE is the dimensionality of the state space.
	Put differently, Storm performs better on problems where only a few variables are needed to describe the state space, whereas \NDICE favors benchmarks which include many variables.
	As an example, the \textit{Network} benchmark can be encoded with just 2 variables: The current status of the package and the remaining number of rounds. 
	In contrast, the \textit{Coupon} benchmarks require many variables more, as every coupon requires a variable to track whether it has already been obtained.
	Storm struggles with the resulting high dimensionality, since it enumerates the reachable states, while \NDICE combines states which behave similarly.
	We do however note that the \textit{Alarm} Bayesian network forms an exception to this intuition, since \NDICE performs surprisingly bad whereas Storm performs surprisingly well.

\subsubsection*{Q4: Influence of Nondeterminism.}
Overall, the share of nondeterministic flip operations does not significantly influence the runtime of \NDICE.
This is best demonstrated with the \textit{3SAT} benchmarks:
Here, both the compilation time as well as the ADD size are unaffected by the number of nondeterministic variables, whereas the compressed MDP size and the time spent on model checking increase.
The same effect can also be observed when comparing the probabilistic and nondeterministic instances of the Bayesian network benchmarks, as compilation times are largely unaffected and model checking is a little slower.
This is to be expected, as the ADD construction is agnostic to the nature of flip operations, whereas model checking treats the flips differently depending on their nature.
However, the overall time spent on model checking remains negligible compared to the compilation, and thus the amount of nondeterminism in a benchmark does not significantly affect the overall runtime.

\medskip

In summary, \NDICE is able to construct compact MDPs for rather large programs, speeding up the inference procedure as a whole.
Thanks to this, the inference process is able to outperform state-of-the-art model checkers on a variety of benchmarks, but the construction of the MDP can be very expensive for other benchmarks classes.
Due to this, we believe that \NDICE is particularly suited for problems which contain a large number of variables, leading to a high dimensionality of the state space.
Beyond that, state compression can further reduce the size of the MDP for \NDICE, but this reduction does not translate to faster inference.
Overall, we believe that our prototype implementation serves as a proof of concept that a \NDICE-style algorithm can efficiently perform inference for probabilistic programs with nondeterminism.


\section{Related Work}\label{sec:related_work}

\paragraph{Expectation Transformers}
Nondeterminism in probabilistic programs with loops and unbounded numbers (but without conditioning) has been studied by \citet{mciver2005abstraction}, through \textit{expectation transformers}.
This work was later extended by \citet{olmedo2018conditioning} to discrete probabilistic programs with conditioning and loops but without nondeterminism.
Additionally,  \citet{olmedo2018conditioning}  showed that nondeterminism and conditioning cannot both be covered by \emph{inductive} expectation transformers.
They (and previously \citet{shoup2009computational}) also showed that loops can model observations, essentially by emulating rejection sampling.
However, this requires reasoning about loops, which is particularly difficult in the presence of randomness~\cite{DBLP:journals/acta/KaminskiKM19}.

\paragraph{Probabilistic Programming Languages \textnormal{(PPLs)}}
A variety of PPLs for Bayesian reasoning has been developed.
In contrast to our work, these often focus on continuous distributions.
In fact, some PPLs like Aqua~\cite{huang2021aqua}, Pyro's Variational Inference~\cite{bingham2019pyro}, or Stan~\cite{carpenter2017stan} \emph{only} support continuous distributions.
Nonetheless, there are many languages which support discrete random variables alongside continuous distributions, for example Anglican~\cite{wood-aistats-2014}, Hakaru~\cite{narayanan2016probabilistic}, SOGA~\cite{randone2024inference} or WebPPL~\cite{dippl}.
However, \emph{none of the above PPLs support nondeterminism}.

A notable PPL which allows nondeterminism via decision making is ProbLog~\cite{fierens2015inference}. 
ProbLog features both conditioning via evidence and an extension to determine optimal decisions~\cite{van2010dtproblog}, but explicitly forbids combining them.
Another example is IBAL~\cite{pfeffer2007design}, which includes a choose-statement to pick the \textit{locally} optimal value to maximize a utility function.
However, \Cref{sec:subOpt} shows that local reasoning can be insufficient. 

Languages like Psi~\cite{gehr2020lambdapsi} allow unspecified inputs, resulting in parameterized inference.
Parameters are fundamentally different from nondeterminism.
In particular, parameters are fixed at program start and hence the semantics of \ref{ExLet1} and \ref{ExLet2} cannot easily be distinguished.

PMAF~\cite{wang2018pmaf} is a framework for the analysis of probabilistic programs with unstructured control flow, recursion, nondeterminism, and observations.
PMAF can be instantiated with a variety of analyses, and it may be possible to perform \NDICE inference within the framework. 
However, the semantics of observations differ between PMAF and \NDICE \cite{wang2018denotational}:
Both assign probability $0$ to failed observations, but PMAF does not normalize the resulting distributions.
Hence, the PMAF equivalent to $\OBSERVE{\FLIP{\frac{2}{3}}}$ only returns $\TRUE$ with probability $\frac{2}{3}$.

\paragraph{Probabilistic Model Checking and Bisimilarity}
Since \NDICE reduces inference to a conditional reachability property on an MDP, every problem which can be modeled by a \NDICE program could also be modeled directly as an MDP.
Then, a \textit{probabilistic model checker} can be used to perform model checking as was seen in the empirical evaluation.
This approach offers many advantages compared to \NDICE:
For one, a variety of mature model checkers has been developed over the years \cite{hahn20192019}.
Secondly, MDPs themselves are more expressive than \NDICE programs as they can model cycles, concurrency and symbolic unbounded numbers.
Leveraging these advantages, \citet{jansen2016bounded} performed inference on probabilistic programs with nondeterminism by simply unrolling the operational semantics of programs as an MDP, and letting a probabilistic model checker determine the conditional reachability probabilities.
\Cref{sec:implementation} shows that this can be an effective approach in practice, but the constructed MDPs tend to be very large.

Hence, it is reasonable to consider approaches like \NDICE, which aim to reduce the size of the final MDP in order to speed up model checking.
Here, model checkers also use ADDs \cite{baier1997symbolic}, but contrary to our work ADDs are used to compactly represent the full state space, rather than to remove states.
Still, probabilistic model checkers have developed a variety of other techniques to speed up model checking \cite{hartmanns2025revised}.
Here, \emph{probabilistic bisimulation} \cite{larsen1989bisimulation, segala1995probabilistic} is particularly relevant to our work, since it presents a generic way to minimize MDPs.
Most notably, the redundant states removed by \NDICE and bisimulations differ, and both approaches remove redundancies which the other misses:
Bisimulations remove symmetries such as e.g. distinct flip operations which use the same probabilities, whereas the decision diagrams of \NDICE remove \enquote{deterministic} states with only a single successor.

Regarding the empirical evaluation in \Cref{sec:implementation}, we found that using bisimulations to reduce the directly encoded MDPs used by Storm did not yield any speedups compared to applying Storm directly.
In fact, calculating the minimized MDP was usually very expensive, and frequently resulted in time-outs.
Consequently, we used Storm without minimizing the MDP first.

\paragraph{MDP Modeling Languages}
	In some ways, \NDICE can be seen as a modeling language for (acyclic) MDPs, and so we will compare it to existing MDP modeling languages like Jani~\cite{budde2017jani} or Prism~\cite{KNP11}.
	Jani and Prism are used to specify the input for model checking tools like Storm \cite{hensel2022probabilistic} or Modest \cite{hartmanns2014modest}, and they allow one to define general MDPs which may contain cycles or be constructed from multiple modules running in parallel. 
	Regarding modeling itself, \NDICE specifies models with the syntax and structure of programs, Prism uses a state-based language whereas Jani encodes models in the JSON data format. 
	However, both Prism and Jani use the same modeling approach:	
	
	First, a set of variables is specified, and the state space of the MDP is defined as the set of possible variable assignments. 
	Transitions then manipulate the variables, and their effects are defined through symbolic expressions which specify the updated values of the variables.
	Whether or not a transition is enabled in a state is defined with a guard, i.e. a boolean condition on the variables. 
	A transition can only be taken when its guard condition holds for the state, but beyond that transitions are always enabled.
	This is useful when modeling e.g. a server which should always be waiting for different requests, but is inconvenient for scenarios where events have to occur in a strict order, such as \Cref{fig:plane}.
	 Hence, MDPs often need an additional state variable to explicitly capture which transitions should be enabled next, see for example \citet{d1997bounded}.
	 In contrast, \NDICE automatically defines the order of operations through the natural program order.

\paragraph{Marginal MAP}
The goal of a marginal maximum a posteriori (MMAP) query is closely related to our inference task: 
Find an optimal assignment to some variables -- which can be interpreted as a nondeterministic choice -- in order to maximize the marginal probability over the remaining variables \cite{liu2013variational, choi2022solving}.
Contrary to this very similar problem statement, we believe that MMAP problems and \NDICE inference require distinct solution approaches.

For one, \NDICE inference cannot easily be performed by MMAP solvers, as the goal of MMAP queries is to find an optimal \textit{assignment} to variables.
An assignment fixes a single value for every nondeterministic choice, which in the context of \NDICE programs means statically replacing every $\NFLIP$ statement with $\TRUE$ or $\FALSE$.
Our model of nondeterminism is more general, hence formulating \NDICE inference as a MMAP query may require enumerating all contexts for $\NFLIP$ as a separate choice, leading to exponentially larger input programs.
Similarly, \NDICE cannot easily solve MMAP queries, as \NDICE may consider too many behaviors.
Still, MMAP queries could be compiled into \NDICE programs by performing all nondeterministic choices at the start of the program to avoid acting based on probabilistic outcomes. 
However, this removes some of the program structure which is exploited by ADDs, which would negatively affect the inference efficiency.

\section{Conclusion and Future Work}\label{sec:future}
We have introduced syntax and semantics of \NDICE, a probabilistic programming language which conservatively extends \DICE with nondeterminism.
We also presented a novel inference algorithm for probabilistic programs with nondeterminism.
The algorithm compiles programs to Boolean formulas and reduces inference queries to conditional reachability in an MDP.
We evaluated a prototype implementation and showed that our algorithm can feasibly do inference for large programs.
We believe our work is a promising step towards efficient inference for probabilistic programs with nondeterminism and conditioning.
As future work, we deem the following worthwhile.

\paragraph{Unbounded Loops}
PPLs like \DICE, \NDICE or Psi \cite{gehr2016psi} can only approximate the behavior of unbounded loops with bounded loops.
However, \citet{torres2024iteration} extended \DICE with a statement for \textit{unbounded} iteration, as well as an exact inference method for programs with this new statement.
We are eager to see if their work can be extended with nondeterminism.

\paragraph{Determining Optimal Programs}
Our work only determines the maximum probability to return a value, but it may also be of interest to retrieve a purely probabilistic program which generates those probabilities.
Synthesizing optimal deterministic programs has been considered without conditioning in \citet{batz2024programmatic}, but it is unclear if this approach can be generalized.

	\newpage

	\section*{Data-Availability Statement}
	The prototype implementation from \Cref{sec:implementation} is available along with the benchmark programs used in the empirical evaluation \cite{nodice_artefact}.

	\bibliographystyle{ACM-Reference-Format}
	\bibliography{bibliography}

	\ifthenelse{\boolean{isExtended}}{
		
	\newpage
	\appendix

\section{\NDICE Compilation Rules for Function Calls}\label{app:function_comp}

In this section we will briefly present the \NDICE compilation rules for function definitions, function calls and \NDICE programs.
Notably, the rules are very similar to the compilation rules for \DICE by \citet{holtzen2020scaling}.

\begin{figure}
	\begin{center}
		\begin{adjustbox}{max width=\textwidth}
			\renewcommand{\arraystretch}{4.0}
			\begin{tabular}{c c}
				\inference[\ifthenelse{\boolean{showRuleNames}}{\texttt{(C-Fun)}}{}]
				{\Gamma\cup\{x_1 \colon \tau_1\}, \Phi \vdash e \colon \tau_2 \leadsto (\dot{\varphi}, \gamma, t)}
				{ \Gamma, \Phi \vdash \FUNC{\FUNCNAME (x_1 \colon \tau_1) \colon \tau_2}{e} \leadsto (\dot{\varphi}, \gamma, t)} 
				
				& 
				\inference[\ifthenelse{\boolean{showRuleNames}}{\texttt{(C-Prog1)}}{}]
				{\Gamma, \Phi \vdash e \colon \tau \leadsto (\dot{\varphi}, \gamma, t)}
				{ \Gamma, \Phi \vdash \bullet\, e \colon \tau\leadsto (\dot{\varphi}, \gamma, t)}
				
				\\
				\multicolumn{2}{c}{
					\inference[\ifthenelse{\boolean{showRuleNames}}{\texttt{(C-Prog2)}}{}]
					{\Gamma, \Phi \vdash \FUNC{\FUNCNAME (x_1 \colon \tau_1) \colon \tau_2}{e} \leadsto (\dot{\varphi}_\FUNCNAME, \gamma_\FUNCNAME, t_\FUNCNAME)\\
						\Gamma \cup \{\FUNCNAME \mapsto \tau_1 \rightarrow \tau_2\}, \Phi \cup \{\FUNCNAME \mapsto (x_1, \dot{\varphi}_\FUNCNAME, \gamma_\FUNCNAME, t_\FUNCNAME)\} \vdash p \colon \tau\leadsto (\dot{\varphi}, \gamma, t)
					}
					{\Gamma, \Phi \vdash  \FUNC{\FUNCNAME (x_1 \colon \tau_1) \colon \tau_2}{e}\; p \colon \tau\leadsto (\dot{\varphi}, \gamma, t)}
				}\\

				\multicolumn{2}{c}{
					\inference[\ifthenelse{\boolean{showRuleNames}}{\texttt{(C-Call)}}{}]
					{\Gamma(x_1) = \tau_1 \qquad \Gamma(\FUNCNAME) = \tau_1\rightarrow \tau_2\\
						\Phi(\FUNCNAME) = (x_{arg}, \dot{\varphi}, \gamma, t) \qquad (\dot{\varphi}', \gamma', t') = \texttt{RefreshFlips}(x_{arg}, \dot{\varphi}, \gamma, t)
					}
					{\Gamma, \Phi \vdash  \FUNCNAME(x_1) \colon \tau_2 \leadsto (\dot{\varphi}'[x_{arg}\mapsto x_1], \gamma'[x_{arg}\mapsto x_1], t')}
				}
			\end{tabular}
		\end{adjustbox}
	\end{center}
	\caption{
		\NDICE Compilation rules for functions and programs. We assume without loss of generality but for simplicity that function calls are only ever given identifiers as arguments.}
	\label{nDice Function Compilation} 
\end{figure}
Compiling programs with function definitions and calls requires extending our compilation judgment with an additional context piece $\Phi$, called the \textit{compiled function table}.
$\Phi$ maps function identifiers to tuples $(x_{arg}, \dot{\varphi}, \gamma, t)$, where $x_{arg}$ is the formal name of the function's argument, and the other elements are the function body's model formula, accepting formula and annotated trace of flips.
Thus, compilation judgments for expressions now have the form $\Gamma, \Phi \vdash e \leadsto (\dot{\varphi}, \gamma, t)$.
\Cref{nDice Function Compilation} presents the new compilation rules for programs and function calls, all other rules are exactly as in \Cref{nDice Compilation} with the added context $\Phi$.
Once again, the rules are very similar to the compilation rules for \DICE~\cite{holtzen2020scaling}.

The judgment $\Gamma, \Phi \vdash p  \colon \tau \leadsto (\dot{\varphi}, \gamma, t)$ compiles programs inductively:
First, \texttt{C-Prog2} uses \texttt{C-Fun} to compile function definitions and accumulates the compiled function table.
Here, \texttt{C-Fun} compiles function definitions into the usual triples $(\dot{\varphi}, \gamma, t)$ by compiling the function's body in the appropriate type environment, and later \texttt{C-Prog2} adds the function to the compiled function table.
After all function definitions have been processed we consider the main expression $\bullet e$, which is compiled as an expression with the accumulated type environment and compiled function table using the rule \texttt{C-Prog1}.
For function calls, \texttt{C-Call} retrieves the compilation result for the called function from $\Phi$ and instantiates the function's formal argument $x_{arg}$ with the real argument $x_1$.
Additionally, $\texttt{RefreshFlips}$ renames all variables aside from $x_{arg}$ in the formulas and trace of flips to ensure that flips remain independent between different calls of the same function.

The added compilation rules are the only difference between inference for expressions and full programs, and both determining the reduced ADD and lifting it to an MDP are done as described in the Cref{sec:inference}.
We now present the Soundness Theorem for \NDICE programs:
\begin{theorem}[Inference Soundness]\label{ProgramSoundnessApp}
	Let $p$ be a \NDICE program  
	and let $\{\}, \{\} \vdash p \colon \tau \leadsto (\dot{\varphi}, \gamma, t)$. 
	Then the reduced ADD for $(\dot{\varphi}, \gamma, t)$ can be lifted to an unnormalized MDP $\mathcal{M}_p$ so that:
	$$ Pr^{max}_{M_p} (\lozenge v \,|\, \lozenge A) = \max \left\{d(v) \,|\, d \in \llbracket p\rrbracket^{\{\}}_D\right\}$$
\end{theorem}
\begin{proof}
	See \Cref{ProgramAppendix}, where the theorem is restated as \Cref{app:ProgramSoundness}.
\end{proof}

\section{Soundness Proof for the Inference Algorithm}

\subsection{Correspondence between unnormalized MDP and unnormalized Semantics}

\begin{lemma}[Value Correctness]\label{ValueCorrectnessLemma}
	Let $v \colon \tau$ be a value. Then there exists an ADD for $(v, \TRUE, [])$ that can be lifted to an unnormalized MDP $\mathcal{M}^u_v$ so that for all schedulers $\Theta$ we have 
	$$Pr^\Theta_{\mathcal{M}^u_v}(\lozenge v') = \begin{cases}
		1 &, v' = v\\
		0 &, v' \neq v
	\end{cases}$$
\end{lemma}
\begin{proof}
	Regardless of the type of $v$, the reduced ADD for $v$ will always consist of only a single leaf node, namely the one with the value $v$. 
	Thus the unnormalized lifted MDP $\mathcal{M}^u_v$ will consist of a single state labeled $v$ with a single transition, and there is only a single scheduler.
	The desired property then holds trivially for $\mathcal{M}^u_v$ as we will always reach a state labeled $v$, but never any states with a different label.
\end{proof}

\begin{lemma}[Substitution Correctness]\label{SubstitutionCorrectnessLemma}
	For all $x\colon \tau$ and $v \colon \tau$ we have $F_\tau(x) [x \overset{\tau}{\mapsto} v] = v$
\end{lemma}
\begin{proof}
	Proof by induction on $\tau$:
	\begin{itemize}
		\item $\tau = \BOOL$: $F_\BOOL(x) [x \overset{\BOOL}{\mapsto} v] = x [x \overset{\BOOL}{\mapsto} v] = v$
		\item $\tau = \tau_1 \times \tau_2$: Let $v = (v_l, v_r)$, we then get:
		\begin{align*}
			F_{\tau_1 \times \tau_2}(x) [x \overset{\tau_1 \times \tau_2}{\mapsto} v] &= (F_{\tau_1 }(x_l),\, F_{ \tau_2}(x_r)) [x \overset{\tau_1 \times \tau_2}{\mapsto} v]\\
			&=_1 (F_{\tau_1 }(x_l)[x_l \overset{\tau_1}{\mapsto} v_l][x_r \overset{\tau_1}{\mapsto} v_r],\, F_{ \tau_2}(x_r)[x_l \overset{\tau_1}{\mapsto} v_l][x_r \overset{\tau_1}{\mapsto} v_r])\\
			&=_2 (F_{\tau_1 }(x_l)[x_l \overset{\tau_1}{\mapsto} v_l],\, F_{ \tau_2}(x_r)[x_r \overset{\tau_2}{\mapsto} v_r])\\
			&=_{IH} (v_l, v_r)\\
			&= v
		\end{align*}
		For $=_1$, one can prove by induction on $\dot{\varphi}$ in a straightforward manner that 
		$$\dot{\varphi}[x \overset{\tau_1 \times \tau_2}{\mapsto} (v_l,v_r)] = \dot{\varphi}[x_l \overset{\tau_1}{\mapsto} v_l][x_r \overset{\tau_2}{\mapsto} v_r]$$ 
		holds for general tuples of boolean formulas.
		The base case follows from the definition of substitution, and the inductive case follows from the inductive hypothesis and the definition of typed substitutions:
		\begin{align*}
			(\dot{\varphi}_1,\, \dot{\varphi}_2)[x \overset{\tau_1 \times \tau_2}{\mapsto} (v_l,v_r)]
			&= (\dot{\varphi}_1[x \overset{\tau_1 \times \tau_2}{\mapsto} (v_l,v_r)],\, \dot{\varphi}_2[x \overset{\tau_1 \times \tau_2}{\mapsto} (v_l,v_r)])\\
			&=_{IH} (\dot{\varphi}_1[x_l \overset{\tau_1}{\mapsto} v_l][x_r \overset{\tau_2}{\mapsto} v_r],\, \dot{\varphi}_2[x_l \overset{\tau_1}{\mapsto} v_l][x_r \overset{\tau_2}{\mapsto} v_r])\\
			&= (\dot{\varphi}_1[x_l \overset{\tau_1}{\mapsto} v_l],\, \dot{\varphi}_2[x_l \overset{\tau_1}{\mapsto} v_l])[x_r \overset{\tau_2}{\mapsto} v_r]\\
			&= (\dot{\varphi}_1,\, \dot{\varphi}_2)[x_l \overset{\tau_1}{\mapsto} v_l][x_r \overset{\tau_2}{\mapsto} v_r]
		\end{align*}
		Afterwards $=_2$ follows as $x_r$ will be unaffected by the substitution into $x_l$ and vice versa. 
	\end{itemize}
\end{proof}

\begin{theorem}[Unnormalized Correspondence]\label{UnnormalizedCorrespondence}
	Let $e$ be a \NDICE expression without function calls 
	and let $\{x_i\colon \tau_i\} \vdash e \colon \tau \leadsto (\dot{\varphi}, \gamma, t)$. Then given any values $\{v_i\colon \tau_i\}$ there exists an ADD for $(\dot{\varphi}[x_i\mapsto v_i], \gamma[x_i\mapsto v_i], t)$ that can be lifted to an unnormalized MDP $\mathcal{M}^u_e$ so that:
	\begin{itemize}
		\item for all distributions $d \in \llbracket e[x_i\mapsto v_i] \rrbracket$ there exists a memoryless scheduler $\Theta$ for $\mathcal{M}^u_e$ so that for all values $v$ we have: $d(v) = Pr^\Theta_{\mathcal{M}^u_e}(\lozenge v)$
		\item for all memoryless schedulers $\Theta$ for $\mathcal{M}^u_e$ there exists a distribution $d \in \llbracket e[x_i\mapsto v_i] \rrbracket$ so that for all values $v$ we have: $d(v) = Pr^\Theta_{\mathcal{M}^u_e}(\lozenge v)$
	\end{itemize} 
\end{theorem}
\begin{proof}
	Proof by structural induction on $e$:
	\begin{itemize}
		\item $e = T$ and $e = F$: Follows directly from \Cref{ValueCorrectnessLemma}\\
		
		\item $e = x$: Assume $\Gamma(x) = \tau$.
		We thus have $\Gamma \vdash x \colon \tau \leadsto (F_\tau(x), \TRUE, [])$.
		Now let $v_x\colon \tau$ be the value substituted for $x$. 
		We thus get $\llbracket x[x \mapsto v_x] \rrbracket = \llbracket v_x \rrbracket = \{1|v_x\rangle\}$. 
		We then apply the substitution to our compilation result to get $(F_\tau(x) [x \overset{\tau}{\mapsto} v_x], \TRUE, []) = (v_x, \TRUE, [])$ via \Cref{SubstitutionCorrectnessLemma}.
		The rest now follows from \Cref{ValueCorrectnessLemma}.\\
		
		\item $e = \FST a$: Here, $a$ is an atomic expression from \AEXP, and thus either a value $v$ or a variable identifier $x$. 
		We assume that $a$ is a variable $x$, as the case where $a$ is a value $v$ is subsumed by this case. 
		
		Assume $\Gamma(x) = \tau_1 \times \tau_2$. 
		We then get $\Gamma \vdash \FST x \colon \tau_1 \leadsto (F_{\tau_1}(x_l), \TRUE, [])$.
		Now let $v_x = (v^l_x, v^r_x)$ be the value substituted for $x$, then we get $\llbracket \FST x[x \overset{\tau}{\mapsto} v_x] \rrbracket = \llbracket v^l_r\rrbracket = \{1|v^l_x\rangle\}$.
		We then apply the substitution to our compilation result to get $(F_{\tau_1}(x_l)[x \overset{\tau}{\mapsto} v_x], \TRUE, []) = (v^l_x, \TRUE, [])$, which can be justified as follows: 
		\begin{align*}
			F_{\tau_1}(x_l)[x \overset{\tau}{\mapsto} v_x] 
			&=_1 F_{\tau_1}(x_l)[x_l \overset{\tau}{\mapsto} v^l_x][x_r \overset{\tau}{\mapsto} v^r_x] \\
			&=_2 F_{\tau_1}(x_l)[x_l \overset{\tau}{\mapsto} v^l_x] \\
			&=_3 v^l_x
		\end{align*}
		Where $=_1$ has already been shown in the proof of \Cref{SubstitutionCorrectnessLemma}, $=_2$ follows as $x_l$ is unaffected by substitutions into $x_r$, and $=_3$ follows from  \Cref{SubstitutionCorrectnessLemma}.
		The rest now follows from \Cref{ValueCorrectnessLemma}.\\

		\item $e = \SND a$: Fully analogous to the case for $\FST a$\\

		\item $e =(a_1, a_2)$ Similarly to the previous cases, we only show the case where the atomic expressions $a_1$ and $a_2$ are identifiers $x_1$ and $x_2$. 
		Assume $\Gamma(x_1) = \tau_1$ and $\Gamma(x_2) = \tau_2$. 
		We then get $\Gamma \vdash (x_1, x_2) \colon \tau_1 \times \tau_2 \leadsto ((F_{\tau_1}(x_1), F_{\tau_2}(x_2)) , \TRUE, [])$.
		Now let $v_1$ be the value substituted into $x_1$ and $v_2$ be the value substituted into $x_2$, then we get $\llbracket (x_1, x_2)[x_1 \overset{\tau_1}{\mapsto} v_1][x_2 \overset{\tau_2}{\mapsto} v_2] \rrbracket = \llbracket( v_1, v_2)\rrbracket = \{1|( v_1, v_2)\rangle\}$.
		We then apply the substitution to our compilation result to get $((F_{\tau_1}(x_1), F_{\tau_2}(x_2))[x_1 \overset{\tau_1}{\mapsto} v_1][x_2 \overset{\tau_2}{\mapsto} v_2], \TRUE, []) = ((v_1, v_2), \TRUE, [])$  via \Cref{SubstitutionCorrectnessLemma}. 
		The rest now follows from \Cref{ValueCorrectnessLemma}.\\

		\item $e =\FLIP{\theta}$: We have $\llbracket \FLIP{\theta} \rrbracket = \{\theta|\TRUE\rangle + (1-\theta)|\FALSE\rangle\}$ and $\Gamma \vdash \FLIP{\theta} \colon \BOOL \leadsto (f, \TRUE, [f\colon \theta])$. 
		The reduced and ordered ADD for $(f, \TRUE, [f\colon \theta])$ only has three nodes, an inner node with the variable $f$ and an output node for $\TRUE$ and $\FALSE$ respectively. This ADD will be lifted to the following MDP $\mathcal{M}^u_f$:
		$$\mathcal{M}^u_f \coloneqq (S, s_{init}, \mathit{Act}, P, AP, L)$$
		\begin{align*}
			S &=  \{f, \TRUE, \FALSE\} & s_{init} &= f \\
			\mathit{Act} &= \{l,r,d\} & P &= \{f\xrightarrow{d,\theta}\TRUE, f\xrightarrow{d,1 - \theta}\FALSE, \TRUE\xrightarrow{d,1}\TRUE, \FALSE \xrightarrow{d,1}\FALSE\} \\
			AP &= \{A,R, \TRUE, \FALSE\} & L &= \{f \mapsto \emptyset,\, \TRUE \mapsto \{\TRUE, A\},\, \FALSE \mapsto \{\FALSE, A\}\}
		\end{align*}
		This MDP only has a single memoryless scheduler $\Theta$, with $Pr^\Theta_{\mathcal{M}_f}(\lozenge \TRUE) = \theta$ and $Pr^\Theta_{\mathcal{M}_f}(\lozenge \FALSE) = 1 - \theta$, which lines up with the only distribution in $\llbracket \FLIP{\theta} \rrbracket$.\\

		\item $e =\NFLIP$: We have $\llbracket \NFLIP \rrbracket = \{p|\TRUE\rangle + (1-p)|\FALSE\rangle \;|\; p\in [0,1]\}$ and $\Gamma \vdash \NFLIP \colon \BOOL \leadsto (f, \TRUE, [f\colon n])$. 
		The reduced and ordered ADD for $(f, \TRUE, [f\colon n])$ only has three nodes, an inner node with the variable $f$ and an output node for $\TRUE$ and $\FALSE$ respectively. This ADD will be lifted to the following MDP $\mathcal{M}^u_f$:
		$$\mathcal{M}^u_f \coloneqq (S, s_{init}, \mathit{Act}, P, AP, L)$$
		\begin{align*}
			S &=  \{f, \TRUE, \FALSE\} & s_{init} &= f \\
			\mathit{Act} &= \{l,r,d\} & P &= \{f\xrightarrow{l,1}\TRUE, f\xrightarrow{r,1}\FALSE, \TRUE\xrightarrow{d,1}\TRUE, \FALSE \xrightarrow{d,1}\FALSE\} \\
			AP &= \{A,R, \TRUE, \FALSE\} & L &= \{f \mapsto \emptyset,\, \TRUE \mapsto \{\TRUE, A\},\, \FALSE \mapsto \{\FALSE, A\}\}
		\end{align*}
		Any memoryless schedulers $\Theta_p$ of this MDP always has the following structure with $p\in [0,1]$:
		$$ \Theta_p(\FALSE) = \Theta_p(\FALSE) = 1|d\rangle \qquad \Theta_p(f) = p|l\rangle \, + \, (1-p)|r\rangle $$
		Thus, we have $Pr^{\Theta_p}_{\mathcal{M}_f}(\lozenge \TRUE) = p$ and $Pr^{\Theta_p}_{\mathcal{M}_f}(\lozenge \FALSE) = 1-p$.
		It is now clear that for each scheduler $\Theta_p$ we can find a matching distributions in $\llbracket \NFLIP \rrbracket$ and vice versa.\\

		\item $e = \OBSERVE{a}$: Here, $a$ is an atomic expression from \AEXP, and thus either a value $v$ or a variable identifier $x$. 
		We assume that $a$ is a variable $x$, as the case where $a$ is a value $v$ is subsumed by this case. 
		
		Now assume $\Gamma \vdash x \colon \BOOL  \leadsto (F_\BOOL(x), \TRUE, []) = (x, \TRUE, [])$, which allows us to compile the observe statement as $\Gamma \vdash \OBSERVE{x} \colon \BOOL  \leadsto (\TRUE, x, [])$.
		Now let $v_x\colon \BOOL$ be the value substituted for $x$. 
		We make a case distinction on $v_x$:
		\begin{itemize}
			\item $v_x = \TRUE$: We have $\llbracket \OBSERVE{x}[x\mapsto \TRUE]\rrbracket = \llbracket \OBSERVE{\TRUE}\rrbracket = \{1|\TRUE\rangle\}$.
			The lifted MDP for the reduced ADD of $(\TRUE, x[x\mapsto \TRUE]\, []) = (\TRUE, \TRUE, [])$ only consists of a single state for $\TRUE$, and thus we trivially have $Pr^{\Theta}_{\mathcal{M}^u_o}(\lozenge \TRUE) = 1$ for all schedulers via \Cref{ValueCorrectnessLemma}.
			\item $v_x = \FALSE$: We have $\llbracket \OBSERVE{x}[x\mapsto \FALSE]\rrbracket = \llbracket \OBSERVE{\FALSE}\rrbracket = \{0|\TRUE\rangle\}$.
			The lifted MDP for the reduced ADD of $(\TRUE, [x\mapsto \FALSE]\, []) = (\TRUE, \FALSE, [])$ only consists of a single state for $R$, and thus we trivially have $Pr^{\Theta}_{\mathcal{M}^u_o}(\lozenge \TRUE) = 0$ and $Pr^{\Theta}_{\mathcal{M}^u_o}(\lozenge \FALSE) = 0$ for all schedulers. \\
		\end{itemize}

		\item $e = \ITE{a}{e_T}{e_F}$: Here, $a$ is an atomic expression from \AEXP, and thus either a value $v$ or a variable identifier $x$. 
		We assume that $a$ is a variable $x$, as the case where $a$ is a value $v$ is subsumed by this case. 

		We assume $\Gamma \vdash x \colon \BOOL \leadsto (x, \TRUE, [])$, as well as the compilations $\Gamma \vdash e_T \colon \tau \leadsto (\dot{\varphi}_T,\, \gamma_T,\, t_T)$ and $\Gamma \vdash e_F \colon \tau \leadsto (\dot{\varphi}_F,\, \gamma_F,\, t_F)$.
		We thus can conclude 
		$$\Gamma \vdash \ITE{x}{e_T}{e_F} \colon \tau \leadsto ( 
		(x \overset{\tau}{\land} \dot{\varphi}_T) \overset{\tau}{\lor} (\overline{x} \overset{\tau}{\land} \dot{\varphi}_F), \,
		(x \land \gamma_T) \lor (\overline{x}\land \gamma_F) ,\,
		t_T \,@\, t_F)$$
		Now let $v_x$ be the value substituted for $x$ and assume for the sake of readability that the substitutions on all other free variables have already been preformed.
		We now make a case distinction on the value of $v_x$.
		We assume $v_x = \TRUE$, the case for $v_x = \FALSE$ is fully analogous.
		We thus get $\llbracket ( \ITE{x}{e_T}{e_F})[x \mapsto \TRUE]\rrbracket = \llbracket  \ITE{\TRUE}{e_T}{e_F}\rrbracket = \llbracket e_T\rrbracket$. We further get:
		\begin{align*}
			&( 
			(x \overset{\tau}{\land} \dot{\varphi}_T) \overset{\tau}{\lor} (\overline{x} \overset{\tau}{\land} \dot{\varphi}_F), \,
			(x \land \gamma_T) \lor (\overline{x}\land \gamma_F) ,\,
			t_T \,@\, t_F) [x \mapsto \TRUE] \\
			&= ( 
			(\TRUE \overset{\tau}{\land} \dot{\varphi}_T) \overset{\tau}{\lor} (\overline{\TRUE} \overset{\tau}{\land} \dot{\varphi}_F), \,
			(\TRUE \land \gamma_T) \lor (\overline{\TRUE}\land \gamma_F) ,\,
			t_T \,@\, t_F)\\
			&= ( 
			\TRUE \overset{\tau}{\land} \dot{\varphi}_T, \,
			\TRUE \land \gamma_T ,\,
			t_T \,@\, t_F)\\
			&= ( 
			\dot{\varphi}_T, \,
			\gamma_T ,\,
			t_T \,@\, t_F)
		\end{align*}
		Note that the final expression above aligns with the compilation result for $e_T$, except for the longer trace of flips. 
		However, by design of the compilation process none of the variables in $t_F$ occur in the formulas $\dot{\varphi_T}$ and $\gamma_T$, meaning no variable in $t_F$ affects the value of $(\dot{\varphi_T} \,|\, \gamma_T)_{t_T \, @ \, t_F}$.
		Therefore any ADD for $(\dot{\varphi_T} \,|\, \gamma_T)_{t_T}$ is also an ADD for $(\dot{\varphi_T} \,|\, \gamma_T)_{t_T \, @ \, t_F}$, since any inner node $u$ with $\mathit{var}(u) \in t_F$ will always be a redundant  test.
		The behavior of the if-statement can thus be fully reduced to the behavior of $e_T$, both in the semantics and the compiled expression.
		The rest thus follows from the inductive hypothesis for $e_T$.\\

		\item $e = \LET{x = e_1}{e_2}$:
		Assume the two compilation judgments $\Gamma \vdash e_1 \colon \tau_1 \leadsto (\dot{\varphi}_1, \gamma_1, t_1)$ and $\Gamma \cup \{x \mapsto \tau_1 \} \vdash e_2 \colon \tau_2 \leadsto (\dot{\varphi}_2, \gamma_2, t_2)$. 
		We thus get: 
		$$\Gamma \vdash\LET{x = e_1}{e_2} \colon \tau_2 \leadsto (
		\varphi_2[x\overset{\tau_1}{\mapsto} \varphi_1], 
		\gamma_1 \land \gamma_2[x\overset{\tau_1}{\mapsto} \varphi_1],\, 
		t_1 \,@\, t_2)$$
		For notational simplicity, we now assume that substitutions to all free variables except for $x$ in $e_2$ have been made in both the expressions and the compiled formulas.
		
		By the inductive hypothesis for $e_1$, there exists an ADD $M^u_1$ for $(\dot{\varphi}_1, \gamma_1, t_1)$, so that the distributions generated by the schedulers for the lifted MDP for $M^u_1$ correspond to the distributions in $\llbracket e_1 \rrbracket$.
		Now consider the compiled formulas for $e_2$, where $x$ has been replaced by some value $v$.
		Those formulas will always have the form $(\dot{\varphi}_2[x \mapsto v], \gamma_2[x \mapsto v], t_2)$.
		By the inductive hypothesis for $e_2$, there exists a corresponding ADD $M^u_v$ for each $(\dot{\varphi}_2[x \mapsto v], \gamma_2[x \mapsto v], t_2)$.
		
		We now construct the ADD $M^ u_e$ for $e$ by starting with $M^u_1$, and redirecting all edges into terminal nodes with value $v \neq R$ to the root of the corresponding ADD $M^u_v$.
		We now lift $M^u_e$ to an MDP, and prove the desired property:
		\begin{itemize}
			\item Let $d \in \llbracket e \rrbracket$ be a distribution. 
			Then, by definition of $\llbracket e \rrbracket$ there exist distributions $d_1 \in \llbracket e_1 \rrbracket$ and $\forall v', d_1(v') > 0: \,\exists d_{v'} \in \llbracket e_2[x \mapsto v'] \rrbracket$, so that:
			$ \forall v : d(v) = \sum\limits_{v',\;d_1(v') > 0} d_1(v') \cdot d_{v'}(v) $.
			By the inductive hypothesis, there exist (memoryless) schedulers for $M^u_1$ and all the $M^u_v$ so that:
			\begin{align*}
				\sum\limits_{v',\;d_1(v') > 0} d_1(v') \cdot d_{v'}(v) &=_{IH} \sum\limits_{v',\;d_1(v') > 0} Pr^{\Theta_1}_{\mathcal{M}^u_1}(\lozenge v') \cdot Pr^{\Theta_{v'}}_{\mathcal{M}^u_{v'}}(\lozenge v)\\
				&= Pr^{\Theta_e}_{\mathcal{M}^u_e}(\lozenge v)
			\end{align*}
			The second equation above can be justified by considering how we constructed $M^u_e$ and it's structure as a DAG:
			In $M^u_e$, there is (by design) no overlap between any inner nodes of $M^u_1$ and all the $M^u_v$, thus we can build the scheduler $\Theta_e$ by "composing" the (memoryless) schedulers for the individual MDPs.
			This "composition" of schedulers to construct $\Theta_e$ essentially behaves like $\Theta_1$ while we are in states which belong to $M^u_1$ and like $\Theta_w$ if we are in states which belong to $M^u_w$.
			Furthermore, any path in $M^u_e$ from the initial state to a state $v$ has to pass through $M^u_1$ within $M^u_e$ first, until $M^u_1$ "terminates" in some value $v'$ with $d_1(v') > 0$ and then continue with $M^u_{v'}$, making this splitting of $\lozenge v$ within $M^u_e$ sound.
			Notably, there could be some states which are not covered by these schedulers, as we only consider sub-MDPs $\mathcal{M}^u_{w}$ with $d_1(w) > 0$.
			However, all of those states are unreachable and $\Theta_e$ can soundly pick arbitrary actions for them.
			\item The second point is analogous to the first point: 
			We consider a memoryless scheduler for $M^u_e$ and decompose it into memoryless schedulers for $M^u_1$ and all the $M^u_{v'}$.
			By the inductive hypothesis these schedulers correspond to some distributions generated by $e_1$ and all the $e_2[x \mapsto v']$, which allows us to construct a matching distribution $d_e \in  \llbracket e\rrbracket$.\\
		\end{itemize}
		
		One thing remains to show, namely that the constructed ADD $M^u_e$ is actually an ADD for the compiled formula $(\varphi_2[x\overset{\tau_1}{\mapsto} \varphi_1], 
		\gamma_1 \land \gamma_2[x\overset{\tau_1}{\mapsto} \varphi_1],\, 
		t_1 \,@\, t_2)$.
		It should be clear that $M^u_e$ adheres to the variable ordering defined by $t_1 \,@\, t_2$.
		Now let $k = len(t_1)$ and let $l = len(t_2)$.
		We will now consider some boolean vector $b = (b_1,\dots, b_k, b_{k + 1}, \dots, b_{k + l})$ as an assignment from flips in $t_1 \,@\, t_2$ to $\TRUE$ or $\FALSE$. 
		We will show that the constructed ADD $M^u_e$ evaluated using $b$ returns the same value $v$ as the function for the compiled formulas evaluated with $b$. 
		To this end, we have to consider the definitions from \Cref{sec:ndice_add}, which define the function for the compiled formulas as follows:
		$$\left(\varphi_2[x\overset{\tau_1}{\mapsto} \varphi_1] \middle|  
		\gamma_1 \land \gamma_2[x\overset{\tau_1}{\mapsto} \varphi_1]\right)_{t_1 @ t_2}$$
		
		We will begin by dividing $b$ into two vectors $b' = (b_1,\dots, b_k)$ and $b'' = (b_{k+1},\dots, b_{k+l})$.
		We can now consider $(\dot{\varphi}_1 | \gamma_1)_{t_1} (b')$, which will be either a value $v'$ or the special value $R$.
		\\
		
		Assume now that $(\dot{\varphi}_1 | \gamma_1)_{t_1} (b')$ evaluated to some value $v'$, then we can conclude that $\gamma_1  (b') = \TRUE$ and $\varphi_1  (b') = v'$.
		Next, we simplify the function for the compiled formulas by substituting the variables from $t_1$ with the values from $b'$, which results in the following equation:
		$$
			\left(\varphi_2[x\overset{\tau_1}{\mapsto} \varphi_1] \middle|  
			\gamma_1 \land \gamma_2[x\overset{\tau_1}{\mapsto} \varphi_1]\right)_{t_1 @ t_2} (b)
			=
			\left(\varphi_2[x\overset{\tau_1}{\mapsto} v'] \middle|  
			\gamma_2[x\overset{\tau_1}{\mapsto} v']\right)_{t_2} (b'')
		$$
		Next, we will also partially evaluate the ADD $M^u_e$ with $b'$.
		If we now consider the structure of $M^u_e$, then we note that we first have to step through the ADD $M^u_1$, which by the inductive hypothesis is an ADD for $(\dot{\varphi}_1 | \gamma_1)_{t_1}$.
		Thus, evaluating $M^u_1$ with $b'$ will necessarily lead to the output node $v_1$ as well, but in the case of $M^u_e$ this output node has been replaced by the ADD $M^u_{v'}$.
		Finally, by the inductive hypothesis $M^u_{v'}$ is an ADD for the function $\left(\varphi_2[x\overset{\tau_1}{\mapsto} v'] \middle|  
		\gamma_2[x\overset{\tau_1}{\mapsto} v']\right)_{t_2}$, which means both $M^u_{v'}$ and this function will return the same value for $b''$.
		Thus, $M^u_e$ and the compiled formulas for $e$ return the same value for $b$.
		\\
		
		We now consider the case $(\dot{\varphi}_1 | \gamma_1)_{t_1} (b') = R$.
		We can conclude that $\gamma_1 (b') = \FALSE$, which in particular means $\left(\gamma_1 \land \gamma_2[x\overset{\tau_1}{\mapsto} \varphi_1]\right)(b) = \FALSE$.
		Therefore, the function for the compiled formulas as a whole will evaluate to $R$.
		Similarly to the previous case, we can reduce the evaluation of  $M^u_e$ to the ADD $M^u_1$.
		Again by the inductive hypothesis $M^u_1$ will also reach an output node containing $R$, which remains an output node in $M^u_e$.

	\end{itemize}
\end{proof}

\subsection{Correctness of Normalization and Reduction}

The previous section showed that there exists some ADD for the compiled formula which can be lifted to an unnormalized MDP, so that the unnormalized semantics align with the reachability probabilities by the schedulers in the MDP.
This section will build towards the actual soundness theorem in two steps:
We first show that the compiled formulas can be lifted to a normalized MDP which then corresponds to the normalized semantics of the compiled expression.
Afterwards we show that all reduction operations on ADDs can be applied without changing the maximum reachability probabilities in the MDP, which allows us to use the reduced ADD instead.

	We will not immediately continue with the soundness proof of the \NDICE inference process, but instead we will first introduce the notion of \textit{normalized} MDPs.
	Normalized MDPs will be a useful tool to allow us to reason about ordinary reachability probabilities instead of conditional reachability probabilities when we consider the soundness of reduction operation on Decision Diagrams relative to the lifted MDP.
	To this end, normalized MDPs are constructed from \textit{unnormalized} MDPs (see \ref{sec:MDP_inf}) as follows:
	
	We first remind ourselves that the ultimate goal of our inference process is to determine probabilities of the form $Pr^{max}_{\mathcal{M}^u_e} (\lozenge v \;|\; \lozenge A)$.
	Notably, \citet{baier2014computing} showed how these conditional reachability properties can be reduced to ordinary reachability properties by transforming the underlying MDP.
	However, we will not use their work directly, as our MDPs offer a lot of structure to exploit for an easier transformation.
	Instead, we use a simplified transformation which was inspired by \citet{junges2021runtime}.
	
	The key idea is to re-distribute the probability mass of paths which do not reach an accepting state, to paths which do reach an accepting state.
	To this end, we will mimic rejection sampling: 
	Whenever we reach a state $s_r$ with $R\in L(s_r)$ , i.e.\ a state which corresponds to the output $R$ in the ADD, then we "restart the execution" and return to the initial state of the MDP.
	We thus lift an ADD for $(\dot{\varphi}, \gamma, t)$ to a \textit{normalized} MDP $\mathcal{M}_e$ by lifting it to the \textit{unnormalized} MDP $\mathcal{M}^u_e$, and replacing the transition $s_r \xrightarrow{d, 1} s_r$ with $s_r \xrightarrow{d, 1} s_{init}$.
	Notably, this minor transformation is sufficient to show that for any memoryless scheduler $\Theta$ we have $Pr^{\Theta}_{\mathcal{M}^u_e} (\lozenge v \;|\; \lozenge A) = Pr^{\Theta}_{\mathcal{M}_e} (\lozenge v)$. 
	We formally prove this later as Equation~\ref{Equation2}.

\begin{lemma}[Normalized Correspondence]\label{NormalizedCorrespondenceLemma}
	Let $e$ be a \NDICE expression without function calls 
	and let $\{x_i\colon \tau_i\} \vdash e \colon \tau \leadsto (\dot{\varphi}, \gamma, t)$. Then given any values $\{v_i\colon \tau_i\}$ there exists an ADD for $(\dot{\varphi}[x_i\mapsto v_i], \gamma[x_i\mapsto v_i], t)$ that can be lifted to a \textit{normalized} MDP $\mathcal{M}_e$ so that for all $v \colon \tau$ we have:
	$$ Pr^{max}_{M_e} (\lozenge v) = \max \{d(v) \,|\, d \in \llbracket e[x_i\mapsto v_i]\rrbracket_D\}$$
\end{lemma}
\begin{proof}
	This proof will be divided into multiple steps:
	We first connect the conditional reachability probabilities of memoryless schedulers in the unnormalized MDP to the normalized semantics of the compiled expression via Equation \ref{Equation1}.
	Afterwards, we show that these conditional reachability probabilities of memoryless schedulers in the unnormalized MDP are equal to ordinary reachability probabilities for memoryless schedulers in the normalized MDP with Equation \ref{Equation2}.
	Finally, we use the two equations and some standard facts from MDP theory to prove the theorem at the very end.\\
	
	By \Cref{UnnormalizedCorrespondence}, there exists an ADD $\mathcal{A}$ for $(\dot{\varphi}[x_i\mapsto v_i], \gamma[x_i\mapsto v_i], t)$ which can be lifted to an unnormalized MDP $\mathcal{M}^u_e$ that corresponds to the unnormalized semantics of $\llbracket e[x_i\mapsto v_i] \rrbracket$.
	We will begin by showing that the normalized semantics of $e$ are connected to conditional reachability probabilities in $\mathcal{M}^u_e$:
	
	Now let $d' \in \llbracket e[x_i\mapsto v_i] \rrbracket_D$ be a normalized distribution, then by the definition of the normalized semantics there is some $d \in \llbracket e[x_i\mapsto v_i] \rrbracket$ with $d'(v) = \frac{1}{d_D} d(v)$.
	By \Cref{UnnormalizedCorrespondence}, there exists a memoryless scheduler $\Theta$ with $d(v) = Pr^\Theta_{\mathcal{M}^u_e}(\lozenge v)$ for all values $v$.
	By definition of the normalizing constant $d_D$ for $d$, we have $d_D = \sum_{v} d(v)$, and by construction of $\mathcal{M}^u_e$ we thus have:
	$$d_D = \sum_{v} d(v) = \sum_{v \neq R} Pr^\Theta_{\mathcal{M}^u_e}(\lozenge v) = Pr^\Theta_{\mathcal{M}^u_e}(\lozenge A)$$
	We can now make use of $Pr^\Theta_{\mathcal{M}^u_e}(\lozenge v \land \lozenge A) = Pr^\Theta_{\mathcal{M}^u_e}(\lozenge v)$, which follows from the labeling function of $\mathcal{M}^u_e$, to show the following:
	\begin{align*}
		d' (v) &= \frac{1}{d_D} d(v)\\ 
		&= \frac{1}{d_D} Pr^\Theta_{\mathcal{M}^u_e}(\lozenge v)\\ 
		&= \frac{1}{Pr^\Theta_{\mathcal{M}^u_e}(\lozenge A)} Pr^\Theta_{\mathcal{M}^u_e}(\lozenge v) \\
		&= \frac{Pr^\Theta_{\mathcal{M}^u_e}(\lozenge v \land \lozenge A) }{Pr^\Theta_{\mathcal{M}^u_e}(\lozenge A)} \\
		&= Pr^\Theta_{\mathcal{M}^u_e}(\lozenge v \,|\, \lozenge A)
	\end{align*}
	Note that in the case of $d_D = Pr^\Theta_{\mathcal{M}^u_e}(\lozenge A) = 0$ both the conditional reachability and the normalized distribution are defined to be $0$.
	Therefore, for every distribution $d' \in \llbracket e[x_i\mapsto v_i] \rrbracket_D$ there exists a corresponding memoryless scheduler $\Theta$ so that $d'(v)$ lines up with the conditional reachability probabilities for $\Theta$ in $\mathcal{M}^u_e$.
	Note that we can also prove the opposite direction, i.e. for every memoryless scheduler $\Theta$ there exists a corresponding distribution $d' \in \llbracket e[x_i\mapsto v_i] \rrbracket_D$.
	This direction follows as \Cref{UnnormalizedCorrespondence} provides a matching distribution in the unnormalized semantics for every memoryless scheduler, and afterwards we can use the same reasoning steps in the opposite direction.
	Therefore, we arrive at the following equation for all $v\in \llbracket \tau \rrbracket$:
	\begin{align*}
		 \max\limits_{\Theta \text{ memoryless}} Pr^{\Theta }_{\mathcal{M}^u_e} (\lozenge v \,|\, \lozenge A) = \max \{d(v) \,|\, d \in \llbracket e[x_i\mapsto v_i]\rrbracket_D\} \tag{I}\label{Equation1}
	\end{align*}

	Next up, we will link the conditional reachability probabilities in the unnormalized MDP to the ordinary reachability probabilities in the normalized MDP.
	To this end, we will once again consider the ADD $\mathcal{A}$, which we previously lifted to the unnormalized MDP $\mathcal{M}^u_e$, and we will now lift it to the normalized MDP $\mathcal{M}_e$.
	Now consider both $\mathcal{M}^u_e$ and $\mathcal{M}_e$:
	They share the same states, initial state, labeling function and even the transitions are almost identical.
	The only difference can be found in states $s_r$ which were lifted from terminal nodes $u$ in the ADD with $\mathit{val}(u) = R$, where $\mathcal{M}^u_e$ has the transition $s_r \xrightarrow{d, 1} s_r$ and $\mathcal{M}_e$ has the transition $s_r \xrightarrow{d, 1}s_{init}$.
	Furthermore, schedulers simply map states to distributions over actions, and $\mathcal{M}^u_e$ and $\mathcal{M}_e$ share the same set of states and enabled actions for each state.
	Therefore, we are able to use any memoryless scheduler $\Theta$ for $\mathcal{M}^u_e$ as a memoryless scheduler for $\mathcal{M}_e$.
	Thus, we will prove that for any memoryless scheduler $\Theta$ and value $v \neq R$ we have:
	\begin{align*}
		Pr^{\Theta }_{\mathcal{M}^u_e} (\lozenge v \,|\, \lozenge A) = Pr^{\Theta }_{\mathcal{M}_e} (\lozenge v) \tag{II}\label{Equation2} 
	\end{align*}
	In the following we assume $Pr^{\Theta }_{\mathcal{M}^u_e} (\lozenge A) > 0$ to avoid division by $0$.
	Should this assumption not hold then both sides of the equation are $0$ and the equation holds trivially.
	
	Now let $\Theta$ be some memoryless scheduler for $\mathcal{M}^u_e$ and $\mathcal{M}_e$, and let $v \neq R$ be some value.
	In order to prove \ref{Equation2}, we will consider a path through the MDP $\mathcal{M}^u_e$ for $\Theta$.
	All paths start in the initial state, and eventually reach a state which was lifted from a terminal node of the original ADD. 
	We now distinguish three outcomes for these paths, depending on the kind of trap state of $\mathcal{M}^u_e$ where the path gets trapped:
	\begin{itemize}
		\item We can reach a trap state $s_v$ with $v \in L(s_v)$, i.e. we reach a state representing the desired value
		\item We can reach a trap state $s_a$ with $v \notin L(s_a)$ and $A \in L(s_a)$, i.e. we reach a state representing a different value, but still an accepting run
		\item We can reach a trap state $s_r$ with $R \in L(s_r)$, i.e. we reach a state representing a rejected run
	\end{itemize}
	By construction of $\mathcal{M}^u_e$, these three cases are all mutually exclusive.
	Notably, $\Theta$ fully resolves the non-determinism in $\mathcal{M}^u_e$, and we can thus associate a probability to each of these outcomes.
	We write $p_{v}$ for the probability to reach a trap state $s_v$ with $v \in L(s_v)$, and similarly we will define $p_{a}$ and $p_{r}$.
	With these probabilities defined, we now get:
	$$ Pr^{\Theta }_{\mathcal{M}^u_e} (\lozenge v \,|\, \lozenge A) 
	= \frac{Pr^{\Theta }_{\mathcal{M}^u_e} (\lozenge v \land \lozenge A)}{Pr^{\Theta }_{\mathcal{M}^u_e} ( \lozenge A)} 
	= \frac{Pr^{\Theta }_{\mathcal{M}^u_e} (\lozenge v)}{Pr^{\Theta }_{\mathcal{M}^u_e} ( \lozenge A)} 
	= \frac{p_{v}}{p_{v} + p_a}$$
	Now consider a path through the normalized MDP $\mathcal{M}_e$ for $\Theta$.
	In particular, note that every path through the normalized MDP $\mathcal{M}_e$ can be divided into multiple runs through the underlying ADD, where every time we end up in a state $s_r$ with $R \in L(s_r)$ we return to the initial state.
	This process repeats until we reach either a state $s_v$ with $v \in L(s_v)$ representing the desired value or an accepting state $s_a$ with $v \notin L(s_a)$ and $A \in L(s_a)$.
	Notably, the probabilities to reach such states $s_v$, $s_a$ or $s_r$ on a single root-to-terminal-node path through the underlying ADD are the same as they were for $\mathcal{M}^u_e$, which is reasonable as the inner structure of the unnormalized and normalized MDP are the same.
	To this end, we once again consider the probabilities $p_{v}$, $p_{a}$ and $p_{r}$.
	
	Now consider the paths in $\lozenge v$. 
	All of these paths are made up by a series of rejected runs through the ADD, i.e. runs that result in the value $R$, followed by a single accepting run.
	With regard to the MDP, this means these paths may reach a state $s_r$ with $R \in L(s_r)$ arbitrarily often, and then eventually they reach a state $s_v$ with $v \in L(s_v)$. 
	Therefore, the probability to reach a state labeled with $v$ after reaching a state labeled $R$ exactly $i$ times, can be determined as $p_v \cdot p^i_r$.
	We can now determine the probability of $\lozenge v$ by considering the sum over all values for $i$.
	We thus get:
	\begin{align*}
		Pr^\Theta_{\mathcal{M}_e}(\lozenge v) &= \sum\limits_{i = 0}^{\infty} p_v \cdot p^i_r \\
		&=  p_v \cdot \sum\limits_{i = 0}^{\infty}  p^i_r \\
		&= p_v \cdot \frac{1}{1- p_r} \\
		&= \frac{p_v}{p_v + p_a}
	\end{align*}
	Above we made use of the fact $p_v + p_a + p_r = 1$ which follows from their definition and the structure of the MDP.
	We now have proven that both sides of Equation~\ref{Equation2} are equal to the same quotient, therefore the equation holds.\\
	
	At this point, there is one concern we have not addressed explicitly yet:
	Are memoryless schedulers truly sufficient?
	It is well-known that deterministic, memoryless schedulers are sufficient for ordinary reachability probabilities, see for example \citet{forejt2011automated}.
	However, the same is not the case for conditional reachability probabilities, where sometimes finite memory schedulers are required to maximize conditional reachability probabilities \cite{andres2008conditional}.
	Nonetheless, for us memoryless schedulers are sufficient in this context as well:
	 
	Equation~\ref{Equation1} links the distributions in the normalized semantics to memoryless schedulers, no finite memory schedulers are needed here.
	Similarly, Equation~\ref{Equation2} links every memoryless scheduler in the normalized MDP to a memoryless scheduler in the unnormalized MDP and vice versa.
	Thus, all contexts which address conditional reachability probabilities in our work only require memoryless schedulers, and when we maximize ordinary reachability probabilities then memoryless schedulers are already known to be sufficient.
	Finally, we can connect the things we have shown so far to prove the desired property.
	\begin{align*}
		Pr^{max}_{M_e} (\lozenge v) 
		&=  \max\limits_{\Theta}\; Pr^{\Theta }_{\mathcal{M}_e} (\lozenge v) & \text{Def. of $Pr^{max}$}\\
		&= \max\limits_{\Theta \text{ memoryless}} Pr^{\Theta }_{\mathcal{M}_e} (\lozenge v)  & \text{memoryless optimality}\\
		&= \max\limits_{\Theta \text{ memoryless}} Pr^{\Theta }_{\mathcal{M}^u_e} (\lozenge v \,|\, \lozenge A) & \text{Equation \ref{Equation2}} \\
		&= \max \{d(v) \,|\, d \in \llbracket e[x_i\mapsto v_i]\rrbracket_D\} & \text{Equation \ref{Equation1}}
	\end{align*}
\end{proof}

\begin{lemma}[Reduction Soundness]\label{ReductionSoundnessLemma}
	All reduction operations on ADDs preserve optimal maximum reachability when lifted to a normalized MDP
\end{lemma}
\begin{proof}
	Let $M$ be a normalized MDP lifted from some ADD which contains redundancies as defined in \Cref{sec:ndice_add}. 
	Within this proof, we will use an ADD node $u$ to refer to the corresponding state in the lifted MDP.
	Now consider $Pr^{max}_M (\lozenge v)$ for some value $v \neq R$.
	
	For this proof, we consider reachability properties on MDPs as characterized through \textit{Bellman equations} \cite{forejt2011automated}.
	To this end, we first need to define $S^0 \coloneqq \{s \in S \,|\, Pr^{max}_M (\lozenge v) = 0\}$ to characterize the set of states from which no state labeled $v$ can be reached.
	Note that $S^0$ can be determined using graph based algorithms such as e.g. breadth first search, as it suffices to know if transitions are possible, while their probabilities can be ignored \cite{forejt2011automated}.
	Afterwards, the Bellman equations themselves are defined as follows, where $x(s)$ describes the probability of reaching a state labeled $v$ from the state $s$.
	$$ x(s) = \begin{cases}
		1 &, v \in L(s)\\
		0 &, s \in S^0\\
		\max\limits_{a \in \mathit{Act(s)}} \sum\limits_{s' \in S} P(s, a)(s') \cdot x(s') &,\text{otherwise}
	\end{cases}$$
	$Pr^{max}_M (\lozenge v)$ is then the value for the initial state by the \textit{least} solution to the system of equations described by the Bellman Operator.
	All reduction operations on ADDs take the same form:
	They identify a redundant node $u$ in the ADD as defined by the restrictions from \Cref{sec:ndice_add}, and then remove that redundancy by redirecting all edges going into $u$ to a different node $w$. 
	Here, $w$ is either a duplicate of $u$ in the case of the "no duplicates" rules, or on of the successors of $u$ in the case of the "no redundant tests" rule.
	
	The key idea of this proof is to show that the two nodes $u$ and $w$ have the same Bellman equations, which means we can redirect all edges going into $u$ to $w$ without affecting the reachability probability. 
	Simply put, if the equation for $x(u)$  is the same as the equation for $x(w)$, then it does not matter if other nodes refer to $x(u)$ or $x(w)$ in their Bellman equation, and the redundancies can be removed without affecting the solution to the overall system of equations.
	We now consider the different reduction operations on ADDs:
	For the sake of simplicity, we will refer to states in the MDP by the corresponding ADD node they were lifted from.
	
	\begin{itemize}
		\item No-redundant tests: Assume the ADD contains an inner node $u$ with $\mathit{else}(u) = \mathit{then}(u)$, then we remove $u$ and redirect all edges going into $u$ to $\mathit{else}(u)$.
		
		We make a case distinction on $u \in S^0$.
		If we have $u \in S^0$, then $\mathit{else}(u) \in S^0$ also has to hold. 
		Therefore, we get $x(u) = 0 = x(\mathit{else}(u))$.
		Now assume $u\notin S^0$:
		
		We consider the Bellman equation for $u$, depending on whether $u$ contains a probabilistic or nondeterministic variable:
		\begin{itemize}
			\item $u$ is probabilistic with variable $f_\theta$, thus $\mathit{Act(u)} = \{d\}$ with the following Bellman equation:
			$$x(u) = \theta\cdot x(\mathit{then}(u)) + (1 - \theta)\cdot x(\mathit{else}(u)) = x(\mathit{else}(u))$$
			\item $u$ is non-deterministic: thus $\mathit{Act(u)} = \{l,r\}$ with the following Bellman equation:
			$$x(u) = \max \{x(\mathit{then}(u)), \; x(\mathit{else}(u))\} = \max \{x(\mathit{else}(u)),\; x(\mathit{else}(u))\} = x(\mathit{else}(u))$$
		\end{itemize} 
		Thus we get $x(u) = x(\mathit{else}(u))$, and we can thus replace $u$ with $\mathit{else}(u)$ in the Bellman Operator equation for all other states, which is equivalent to redirecting the transitions from $u$ to $\mathit{else}(u)$ in the MDP.\\

		\item No duplicate tests: If for $u \neq w$ with $\mathit{var}(u) = \mathit{var}(w)$, $\mathit{else}(u) = \mathit{else}(w)$ and $\mathit{then}(u) = \mathit{then}(w)$ then we can redirect all edges going into $u$ to $w$.
		
		We make a case distinction on $u \in S^0$.
		If we have $u \in S^0$, then $\mathit{then}(u) \in S^0$ and $\mathit{else}(u) \in S^0$ also have to hold, and we can conclude $w \in S^0$.
		Therefore, we get $x(u) = 0 = x(w)$.
		Now assume $u\notin S^0$:
		
		We once again make a case distinction on whether $u$ contains a probabilistic or nondeterministic variable, and consider the equations generated by the Bellman Operator:
		In the probabilistic case, we get $x(u) = \theta \cdot x(\mathit{then}(u)) + (1- \theta)\cdot x(\mathit{else}(u)) = x(w)$, meaning $u$ always behaves exactly like $w$.
		
		In the nondeterministic case, we get the equations $x(u) = \max\{x(\mathit{then}(u)),\; x(\mathit{else}(u))\}$ and $x(w) = \max\{x(\mathit{then}(u)),\; x(\mathit{else}(u))\}$.  
		Once again both states adhere to the same equation. \\

		\item No duplicate outputs: Given two terminal nodes in the ADD  $u$ and $w$ with $val(u) = val(w)$ and $u \neq w$, we can redirect all edges going into $u$ to $w$.
		
		We make a case distinction on the value in the leaf node:
		\begin{itemize}
			\item $val(u) = v$: In this case we have $v \in L(u)$ and therefor also $v \in L(w)$.
			We can therefore conclude $x(u) = 1 = x(w)$.
			\item $val(u) = v'$ with $v' \neq v$ and $v' \neq R$: In this case $u$ is a trap state and unable to reach a state labeled with $v$.
			We thus have $u \in S^0$ and $w \in S^0$, meaning we can conclude $x(u) = 0 = x(w)$.
			Redirecting the edges is thus sound.
			\item $val(u) = R$: We make a case distinction on $u \in S^0$.
			If we have $u \in S^0$, then $s_{init} \in S^0$ also has to hold, and we can conclude $w \in S^0$.
			Therefore, we get $x(u) = 0 = x(w)$.
			If we have $u\notin S^0$, then we get $x(u) = x(s_{init}) = x(w)$.
		\end{itemize}
		
	\end{itemize}
\end{proof}

\begin{theorem}[Inference Soundness]\label{SoundnessTheoremProof}
	Let $e$ be a \NDICE expression without function calls 
	and let $\{x_i\colon \tau_i\} \vdash e \colon \tau \leadsto (\dot{\varphi}, \gamma, t)$. Then for any values $\{v_i\colon \tau_i\}$ the  reduced ADD for the formulas $(\dot{\varphi}[x_i\mapsto v_i],\, \gamma[x_i\mapsto v_i],\, t)$ can be lifted to a normalized MDP $\mathcal{M}_e$ so that:
	$$ Pr^{max}_{\mathcal{M}_e} (\lozenge v) = \max \{d(v) \,|\, d \in \llbracket e[x_i\mapsto v_i]\rrbracket_D\}$$
\end{theorem}
\begin{proof}
	The existence of some ADD with the desired property is given by \Cref{NormalizedCorrespondenceLemma}. 
	Using \Cref{ReductionSoundnessLemma}, we can remove redundancies from that ADD until we end up with the reduced ADD, and still preserve the desired property.
	
	At this point, we would like to remark that exhaustively applying reduction operations to some ADD will in fact always yield the reduced or canonical ADD.
	This holds, as every reduction operation removes one node from the ADD, meaning we only need to apply finitely many reductions until we reach an ADD which cannot be reduced.
	Furthermore, the reduction operations do not change the function represented by the ADD, and also there is only a single reduced ADD for a given function.
\end{proof}

\subsection{Correctness of Function Calls and Program Compilation}\label{ProgramAppendix}
We will now consider the soundness of our inference method in the presence of functions and function calls.
To this end we need to consider some preliminary definitions and lemmas:

\begin{definition}[Table compatibility]
	Let $T$ be a function table, let $\Phi$ be a compiled function table, and let $\Gamma$ be a type environment.
	We call $T$ and $\Phi$ \textit{compatible}, if for every function identifier $\FUNCNAME$ with $\Gamma(\FUNCNAME) = \tau_1 \rightarrow \tau_2$ and $\Phi(\FUNCNAME) = (x_{arg}, \dot{\varphi}, \gamma, t)$, and for every input value $v_x \colon \tau_1$, there exists an ADD for $(\dot{\varphi}[x_{arg}\mapsto v_x],\, \gamma[x_{arg}\mapsto v_x],\, t)$ that can be lifted to an unnormalized MDP $\mathcal{M}^u_e$ so that:
	\begin{itemize}
		\item for all distributions $d \in T(\FUNCNAME)(v_x)$ there exists a memoryless scheduler $\Theta$ for $\mathcal{M}^u_e$ so that for all values $v$ we have: $d(v) = Pr^\Theta_{\mathcal{M}^u_e}(\lozenge v)$
		\item for all memoryless schedulers $\Theta$ for $\mathcal{M}^u_e$ there exists a distribution $d \in T(\FUNCNAME)(v_x)$ so that for all values $v$ we have: $d(v) = Pr^\Theta_{\mathcal{M}^u_e}(\lozenge v)$
	\end{itemize} 
\end{definition}

\begin{lemma}[Unnormalized Correspondence with Procedure Calls]\label{UnnormalizedCorrespondenceFunctions} 
	Let $e$ be a \NDICE expression with function calls, $T$ and $\Phi$ be compatible function tables, and let $\{x_i\colon \tau_i\}, \Phi \vdash e \colon \tau \leadsto (\dot{\varphi}, \gamma, t)$. 
	Then given any values $\{v_i\colon \tau_i\}$ there exists an ADD for $(\dot{\varphi}[x_i\mapsto v_i], \gamma[x_i\mapsto v_i], t)$ that can be lifted to an unnormalized MDP $\mathcal{M}^u_e$ so that:
	\begin{itemize}
		\item for all distributions $d \in \llbracket e[x_i\mapsto v_i] \rrbracket^T$ there exists a memoryless scheduler $\Theta$ for $\mathcal{M}^u_e$ so that for all values $v$ we have: $d(v) = Pr^\Theta_{\mathcal{M}^u_e}(\lozenge v)$
		\item for all memoryless schedulers $\Theta$ for $\mathcal{M}^u_e$ there exists a distribution $d \in \llbracket e[x_i\mapsto v_i] \rrbracket^T$ so that for all values $v$ we have: $d(v) = Pr^\Theta_{\mathcal{M}^u_e}(\lozenge v)$
	\end{itemize} 
\end{lemma}
\begin{proof}
	Proof by induction on $e$, almost all cases follow exactly as in the proof for \Cref{UnnormalizedCorrespondence}. 
	This follows as both the semantics as well as the compilation for those statements are fully independent from the function tables, and as such their behavior is completely unchanged.
	The only additional case is the function call, $\FUNCNAME(x)$, which will be proven as follows:\\
	
	Let $e = \FUNCNAME(a)$, where $a$ is an atomic expression from \AEXP, and thus either a value $v$ or a variable identifier $x$. 
	We assume that $a$ is a variable $x$, as the case where $a$ is a value $v$ is subsumed by this case. 
	Now assume $\Phi(\FUNCNAME) = (x_{arg}, \dot{\varphi}, \gamma, t)$ and $(\dot{\varphi}', \gamma', t')= \texttt{RefreshFlips}(x_{arg}, \dot{\varphi}, \gamma, t)$, then we get $\Gamma, \Phi \vdash \FUNCNAME(x) \colon \tau \leadsto (\dot{\varphi}'[x_{arg} \mapsto x], \gamma'[x_{arg} \mapsto x], t')$.
	
	Now let $v_x$ be the value substituted for $x$, then we get $\llbracket \FUNCNAME(x)[x \mapsto v_x] \rrbracket^T = \llbracket \FUNCNAME(v_x)\rrbracket^T = T(\FUNCNAME)(v_x)$. 
	Note that by the table compatibility of $T$ and $\Phi$, the desired correspondence property holds for $T(\FUNCNAME)(v_x)$ and $(\dot{\varphi}[x_{arg} \mapsto v_x], \gamma[x_{arg} \mapsto v_x], t)$.
	
	Next, we consider the compilation result $(\dot{\varphi}'[x_{arg} \mapsto x], \gamma'[x_{arg} \mapsto x], t')$ in more detail, in particular we observe the following:
	$$(\dot{\varphi}'[x_{arg} \mapsto x][x\mapsto v_x], \gamma'[x_{arg} \mapsto x][x\mapsto v_x], t') \\
	= (\dot{\varphi}'[x_{arg} \mapsto v_x], \gamma'[x_{arg} \mapsto v_x], t')$$
	Finally, note that the only difference between $\dot{\varphi}'$ and $\dot{\varphi}$ as well as $\gamma'$ and $\gamma$ is that the variables were renamed by \texttt{RefreshFlips}.
	Thus, any ADD for $(\dot{\varphi}[x_{arg} \mapsto v_x], \gamma[x_{arg} \mapsto v_x], t)$ can be transformed into an ADD for $(\dot{\varphi}'[x_{arg} \mapsto v_x], \gamma'[x_{arg} \mapsto v_x], t')$ by simply renaming the variables in the same manner.
	This allows us to retrieve the correspondence property between $T(\FUNCNAME)(v_x)$ and $(\dot{\varphi}'[x_{arg} \mapsto v_x], \gamma'[x_{arg} \mapsto v_x], t')$ from the correspondence property for $T(\FUNCNAME)(v_x)$ and $(\dot{\varphi}[x_{arg} \mapsto v_x], \gamma[x_{arg} \mapsto v_x], t)$. 
	Thus the proof is complete.
\end{proof}

\begin{lemma}[Unnormalized Correspondence for Programs]\label{UnnormalizedCorrespondencePrograms}
	Let $p$ be a \NDICE program, and let $T$ and $\Phi$ be compatible function tables. 
	Now let $\Gamma, \Phi \vdash p \colon \tau \leadsto (\dot{\varphi}, \gamma, t)$. Then there exists an ADD for $(\dot{\varphi}, \gamma, t)$ that can be lifted to an unnormalized MDP $\mathcal{M}^u_e$ so that:
	\begin{itemize}
		\item for all distributions $d \in \llbracket p \rrbracket^T$ there exists a memoryless scheduler $\Theta$ for $\mathcal{M}^u_e$ so that for all values $v$ we have: $d(v) = Pr^\Theta_{\mathcal{M}^u_e}(\lozenge v)$
		\item for all memoryless schedulers $\Theta$ for $\mathcal{M}_e$ there exists a distribution $d \in \llbracket p \rrbracket^T$ so that for all values $v$ we have: $d(v) = Pr^\Theta_{\mathcal{M}^u_e}(\lozenge v)$
	\end{itemize} 
\end{lemma}
\begin{proof}
	Proof by induction on $p$:
	\begin{itemize}
		\item $p = \bullet e$: 
		We have $\llbracket \bullet e\rrbracket^T = \llbracket e\rrbracket^T$, and by the rules in \Cref{nDice Function Compilation} we have the compilation judgment $\Gamma, \Phi \vdash \bullet e  \colon \tau \leadsto (\dot{\varphi}, \gamma, t)$ if and only if we have $\Gamma, \Phi \vdash e \colon \tau \leadsto (\dot{\varphi}, \gamma, t)$. 
		This case thus reduces to the behavior of expressions with function calls, and the rest follows from \Cref{UnnormalizedCorrespondenceFunctions}.\\
		
		\item $p = \FUNC{\FUNCNAME (x_1 \colon \tau_1) \colon \tau_2}{e}\; p'$: 
		Assume that $\Gamma, \Phi \vdash \FUNC{\FUNCNAME (x_1 \colon \tau_1) \colon \tau_2}{e} \leadsto (\dot{\varphi}_\FUNCNAME, \gamma_\FUNCNAME, t_\FUNCNAME)$, and further let $\Phi' = \Phi \cup \{\FUNCNAME \rightarrow (x_1, \dot{\varphi}_\FUNCNAME, \gamma_\FUNCNAME, t_\FUNCNAME)\}$. 
		Now assume that for the sub-program $p'$ we have $\Gamma\cup \{\FUNCNAME \mapsto \tau_1 \rightarrow \tau_2\}, \Phi' \vdash  p' \colon \tau \leadsto (\dot{\varphi}, \gamma, t)$. 
		We can then conclude that the program as a whole yields the following compilation judgment:
		$$\Gamma, \Phi \vdash  \FUNC{\FUNCNAME (x_1 \colon \tau_1) \colon \tau_2}{e}\; p' \colon \tau \leadsto (\dot{\varphi}, \gamma, t)$$
		Thus the compilation of $p$ essentially reduces to the compilation of $p'$ with the updated compiled function table $\Phi'$. 
		We can show that the same thing is true for the semantics, i.e. the semantics of $p$ reduce to the semantics of $p'$ with an updated function table.
		Hence, let $T' = T \cup \{\FUNCNAME \rightarrow \llbracket \FUNC{\FUNCNAME (x_1 \colon \tau_1) \colon \tau_2}{e} \rrbracket^T\}$, and we get:
		$$\llbracket \FUNC{\FUNCNAME (x_1 \colon \tau_1) \colon \tau_2}{e}\; p' \rrbracket^T = \llbracket p' \rrbracket^{T'}$$
		This case of the lemma now follows from the inductive hypothesis for $p'$ with the updated function tables $T'$ and $\Phi'$, assuming that we can prove that the updated function tables are compatible. 
		By assumption, compatibility is given for any function identifier in $T$ and $\Phi$, we thus only need to consider $\FUNCNAME$:
		
		Let $v\colon \tau_1$ be an arbitrary input value to $\FUNCNAME$, we then get:
		$$T'(\FUNCNAME)(v) = \llbracket \FUNC{\FUNCNAME (x_1 \colon \tau_1) \colon \tau_2}{e} \rrbracket^T(v) = \llbracket e[x_1\mapsto v]\rrbracket^T$$
		We further have $\Phi'(\FUNCNAME) = (x_1, \dot{\varphi}_\FUNCNAME, \gamma_\FUNCNAME, t_\FUNCNAME)$, with $\Gamma, \Phi \vdash \FUNC{\FUNCNAME (x_1 \colon \tau_1) \colon \tau_2}{e} \leadsto (\dot{\varphi}_\FUNCNAME, \gamma_\FUNCNAME, t_\FUNCNAME)$ by definition of $\Phi'$. 
		Furthermore, with the rules from \Cref{nDice Function Compilation} we can conclude from the compilation of the function $f$ that we must have $\Gamma\cup\{x_1 \colon \tau_1\}, \Phi \vdash e \colon \tau_2 \leadsto (\dot{\varphi}_\FUNCNAME, \gamma_\FUNCNAME, t_\FUNCNAME)$.
		
		It remains to show that there is some ADD for $(\dot{\varphi}_f[x_1\mapsto v], \gamma_f[x_1\mapsto v], t_f)$ which can be lifted to an unnormalized MDP so that the schedulers for that MDP correspond to the distributions in $\llbracket e[x_1\mapsto v]\rrbracket^T$.
		This follows from \Cref{UnnormalizedCorrespondenceFunctions} for $e$, with the compatible function tables $T$ and $\Phi$.
	\end{itemize}
\end{proof}

\begin{lemma}[Correctness of Program Inference]\label{ProgramSoundnessProof}
	Let $p$ be a \NDICE program  
	and  let $\{\}, \{\} \vdash p \colon \tau \leadsto (\dot{\varphi}, \gamma, t)$. 
	Then the reduced ADD for $(\dot{\varphi}, \gamma, t)$ can be lifted to an normalized MDP $\mathcal{M}_e$ so that:
	$$ Pr^{max}_{\mathcal{M}_e} (\lozenge v) = \max \left\{d(v) \,|\, d \in \llbracket e\rrbracket^{\{\}}_D \right\}$$
\end{lemma}
\begin{proof}
	The proof is analogous to the proof of \Cref{SoundnessTheoremProof}, and we can simply retrace the steps with the same justifications:
	\Cref{UnnormalizedCorrespondencePrograms} yields an ADD which can be lifted to an unnormalized MDP that corresponds to the unnormalized semantics of $p$. 
	We can then argue as in \Cref{NormalizedCorrespondenceLemma} that lifting this ADD to a normalized MDP yields an MDP which satisfies the desired equation.
	Finally, \Cref{ReductionSoundnessLemma} allows us to soundly reduce the ADD, and we have shown that the reduced ADD can be lifted to a normalized MDP which satisfies the desired equation.
\end{proof}

\begin{theorem}[Inference Soundness]\label{app:ProgramSoundness}
	Let $p$ be a \NDICE program  
	and let $\{\}, \{\} \vdash p \colon \tau \leadsto (\dot{\varphi}, \gamma, t)$. 
	Then the reduced ADD for $(\dot{\varphi}, \gamma, t)$ can be lifted to an unnormalized MDP $\mathcal{M}_p$ so that:
	$$ Pr^{max}_{M_p} (\lozenge v \,|\, \lozenge A) = \max \left\{d(v) \,|\, d \in \llbracket p\rrbracket^{\{\}}_D\right\}$$
\end{theorem}
\begin{proof}
	Follows directly from \Cref{ProgramSoundnessProof}, as well as the construction of the normalized MDP as argued in the proof of \Cref{NormalizedCorrespondenceLemma}.
\end{proof}

	\section{Implementation Details and Detailed Experimental Results}\label{app:implementation}

\subsection{Probabilistic State Compression}\label{app:compression}

Our implementation reduces the state space of the lifted MDP by repeatedly applying \Cref{theo:compression} to states which where lifted from (inner) nodes with probabilistic variables. 
This theorem formally specifies the Compression procedure, and all target states always fulfill the requirements listed in \Cref{theo:compression}.
The soundness of each compression step follows from the theorem below, and we could apply the theorem exhaustively, but we enforce a maximum number of outgoing transitions for each state in the MDP, which guarantees that the Compression of the MDP as a whole can be completed in linear time relative to the size of the MDP.
In our implementation, that limit is set at $40$ transitions. 

\begin{theorem}[Probabilistic State Compression]\label{theo:compression}
	Let $M = (S, s_{init}, \mathit{Act}, P, \mathit{AP}, L)$ be an MDP and let $s_p \in S$ be a state with only one enabled action $\alpha_p$, as well as $s_p \neq s_{init}$, $P(s_p, \alpha_p, s_p) = 0$ and $L(s_p) = \emptyset$. 
	Now let $M' = (S/ \{s_p\}, s_{init}, \mathit{Act}, P', \mathit{AP}, L)$ be the compressed MDP where we define
	$$ P'(s, \alpha, s' ) \coloneqq P(s, \alpha, s') + \left(P(s, \alpha, s_p) \cdot P(s_p, \alpha_p, s')\right) $$
	We then have for all labels $v \in \mathit{AP}$:
	$$Pr^{max}_{M}(\lozenge v) = Pr^{max}_{M'}(\lozenge v)$$
\end{theorem}
\begin{proof}
	We once again consider the maximum reachability as characterized through Bellman equations \cite{forejt2011automated}. 
	To this end, we first need to define $S^0_M \coloneqq \{s \in S \,|\, Pr^{max}_M (\lozenge v) = 0\}$ to characterize the set of states from which no state labeled $v$ can be reached.
	Note that $S^0$ can be determined using graph based algorithms such as e.g. breadth first search, as it suffices to know if transitions are possible, while their probabilities can be ignored\cite{forejt2011automated}.
	Afterwards, the Bellman equations themselves are defined as follows, where $x(s)$ describes the probability of reaching a state labeled $v$ from the state $s$.
	$$ x(s) = \begin{cases}
		1 &, v \in L(s)\\
		0 &, s \in S^0_M\\
		\max\limits_{\alpha \in \mathit{Act(s)}} \sum\limits_{s' \in S} P(s, \alpha, s') \cdot x(s') &,\text{otherwise}
	\end{cases}$$
	$Pr^{max}_M (\lozenge v)$ is then the value for the initial state by the \textit{least} solution to the system of equations described by the Bellman Operator.
	We will write $x(s)$ for the Bellman equation for $M$ and $x'(s)$ for the Bellman equation of the compressed MDP $M'$.
	
	We will now show that for all states $s \in S/ \{s_p\}$ we have $x(s) = x'(s)$.
	This is trivial for all states $s$ with $v \in L(s)$ and follows for all states $s\in S^0$ from $S^0_M = S^0_{M'}$.
	We now consider the remaining case. We first show:
	\begin{align*}
		&\sum\limits_{s' \in S} P(s, \alpha, s') \cdot x(s') \\
		&= P(s, \alpha, s_p) \cdot x(s_p) + \sum\limits_{s' \in S/\{s_p\}} P(s, \alpha, s') \cdot x(s') \\
		&=_1 P(s, \alpha, s_p) \cdot \sum\limits_{s' \in S} P(s_p, \alpha_p, s') \cdot x(s') + \sum\limits_{s' \in S/\{s_p\}} P(s, \alpha, s') \cdot x(s') \\
		&=_2 P(s, \alpha, s_p) \cdot \sum\limits_{s' \in S/\{s_p\}} P(s_p, \alpha_p, s') \cdot x(s') + \sum\limits_{s' \in S/\{s_p\}} P(s, \alpha, s') \cdot x(s') \\
		&= \sum\limits_{s' \in S/\{s_p\}} P(s, \alpha, s_p) \cdot P(s_p, \alpha_p, s') \cdot x(s') + \sum\limits_{s' \in S/\{s_p\}} P(s, \alpha, s') \cdot x(s') \\
		&= \sum\limits_{s' \in S/\{s_p\}} \left(P(s, \alpha, s') + P(s, \alpha, s_p) \cdot P(s_p, \alpha_p, s')\right) \cdot x(s') \\
		&= \sum\limits_{s' \in S/\{s_p\}} P'(s, \alpha, s') \cdot x(s') \\
	\end{align*}
	In $=_1$ we use that $L(s_p) = \emptyset$, meaning that the Bellman equation for $s_p$ is described by the final rule.
	We also use that $s_p$ only has a single enabled action, which simplifies the Bellman equation to just a sum as the maximum over all actions is trivially found with $\alpha_p$.
	Note that $x(s_p) = \sum_{s' \in S} P(s_p, \alpha_p, s') \cdot x(s')$ also holds in the case of $s_p \in S^0_M$, as we then either have $s' \in S^0_M$ or $P(s_p, \alpha_p, s') = 0$ for all states $s' \in S$, meaning the sum as whole will be $0$.
	Afterwards, $=_2$ follows from $P(s_p, \alpha_p, s_p) = 0$.
	It now follows:
	\begin{align*}
		x(s) & = \max\limits_{\alpha \in \mathit{Act(s)}} \sum\limits_{s' \in S} P(s, \alpha, s') \cdot x(s') \\
		&= \max\limits_{\alpha \in \mathit{Act(s)}} \sum\limits_{s' \in S/\{s_p\}} P'(s, \alpha, s') \cdot x(s') \\
	\end{align*}
	Therefore, the system of equations as defined by $x(s)$ is the same as the system of equations as defined by $x'(s)$ up to renaming of $x(s)$. 
	Therefore, the least solution to one system of equations is also the least solution for the other system of equations, and we get $x(s) = x'(s)$ for all states $s \in S/\{s_p\}$.
	Finally, we have $s_p \neq s_{init}$ by assumption, and we can conclude $Pr^{max}_{M}(\lozenge v) = x(s_{init}) = x'(s_{init}) = Pr^{max}_{M'}(\lozenge v)$
\end{proof}

\subsection{Reduction to Boolean Inference}\label{app:boolean_reduction}
In this section we show that we can soundly reduce the semantics of arbitrarily typed programs to boolean programs.
The key idea is that we wrap the expression $e$ we are interested in into a let statement $\LET{x = e}{ x \leftrightarrow v}$, called the boolean reduction of $e$ for $v$, which essentially evaluates $e$ and then tests if the execution yielded the desired value $v$.
The idea is that it does not matter if we ask for the maximum probability for $e$ to return a particular value, or for the maximum probability for the boolean reduction of $e$ for $v$ to return $\TRUE$.
To this end we define a simple equivalence or equality test on values with the following semantics:
$$ \llbracket v_1 \leftrightarrow v_2 \rrbracket^T \Coloneqq \begin{cases}
	\TRUE| 1\rangle &, v_1 = v_2 \\
	\FALSE| 1\rangle &, v_1 \neq v_2 \\
\end{cases}$$
With this definition in hand we present \Cref{theo:boolean_reduction}, which relates the (unnormalized) distributions in the semantics of expressions and their boolean reductions to each other.
Afterwards, \Cref{theo:boolean_reduction_norm} then considers normalized distributions and proves the soundness of this inference approach.
\begin{lemma}[Unnormalized Boolean Reduction]\label{theo:boolean_reduction}
	Let $e$ be an \NDICE expression, and let $T$ be a function table.
	Then for all values $v \in V$ we have:
	\begin{enumerate}
		\item $\forall d \in \llbracket e \rrbracket^T \; \exists d' \in \llbracket \LET{x = e}{ x \leftrightarrow v}\rrbracket^T: d(v) = d'(\TRUE) \land \sum\limits_{w \neq v} d(w) = d'(F)$
		\item $\forall d' \in \llbracket \LET{x = e}{ x \leftrightarrow v} \rrbracket^T \; \exists d \in \llbracket e \rrbracket^T: d(v) = d'(\TRUE) \land \sum\limits_{w \neq v} d(w) = d'(F)$
	\end{enumerate}
\end{lemma}
\begin{proof}
	Let $v \in V$ be some value.
	\begin{enumerate}
		\item Let $d \in \llbracket e \rrbracket^T$ be some distribution.
		Then, by the semantics of the let-statement in \Cref{fig:semantics} we have
		$$d' \in \llbracket \LET{x = e}{ x \leftrightarrow v} \rrbracket^T \Leftrightarrow \exists d_1 \in \llbracket e \rrbracket^T: \forall w, d_1(w) > 0: \,\exists d_{w} \in \llbracket x\leftrightarrow v [x \mapsto w] \rrbracket^T: $$
		$$ \forall u \in V : d'(u) = \sum_{w,\, d_1(w) > 0} d_1(w) \cdot d_{w}(u)$$
		Note that $\llbracket x\leftrightarrow v[x \mapsto w] \rrbracket^T = \llbracket w\leftrightarrow v \rrbracket^T$ is always a singleton set. 
		We can therefore pick $d_1$ to be $d$, and pick $d_w$ as the only element of $\llbracket w\leftrightarrow v \rrbracket^T$ to define the distribution $d' \in \llbracket \LET{x = e}{ x \leftrightarrow v} \rrbracket^T$.
		We first consider both $d(v)$ and $d'(\TRUE)$:
		\begin{align*}
			d'(\TRUE) &= \sum_{w,\, d_1(w) > 0} d_1(w) \cdot d_{w}(\TRUE) \\
			&= \sum_{w,\, d(w) > 0} d(w) \cdot d_{w}(\TRUE) \\
			&=_1 d(v) \cdot d_{v}(\TRUE) \\
			&= d(v)
		\end{align*}
		Note that the final two equations hold as $d_w(\TRUE) = 0$ for all $w \neq v$, and $d_w(\TRUE) = 1$ for $w = v$ by definition.
		Finally, notice that $=_1$ holds both if we have have $d(v) > 0$ or not. 
		If we have $d(v) > 0$, then $=_1$ holds as argued above because $d(v) \cdot d_v(\TRUE)$ is a proper term in the sum and is also the only non-zero summand.
		If we have $d(v) = 0$, then $=_1$ holds as both sides of the equation are $0$, in particular the sum is zero because all the $d_w(\TRUE)$ are $0$.
		
		The second part of the conjunction can be proven similarly.
		However, this time we have  $d_w(\FALSE) = 0$ for all $w = v$, and $d_w(\FALSE) = 1$ for $w \neq v$, allowing us to show the following equations:
		\begin{align*}
			d'(\FALSE) &= \sum_{w,\, d_1(w) > 0} d_1(w) \cdot d_{w}(\FALSE) \\
			&= \sum_{w,\, d(w) > 0 \land w \neq v} d(w) \cdot d_{w}(\FALSE) \\
			&= \sum_{ w \neq v} d(w) \cdot d_{w}(\FALSE) \\
			&= \sum_{ w \neq v} d(w)
		\end{align*}

		\item Let $d' \in \llbracket \LET{x = e}{ x \leftrightarrow v} \rrbracket^T$.
		Therefore, by the semantics of the let-statement as defined in \Cref{fig:semantics} there exists a distribution 
		$d \in \llbracket e \rrbracket^T$, and for all values $w$ with $d(w) > 0$ there exists a distribution $d_{w} \in \llbracket x\leftrightarrow v [x \mapsto w] \rrbracket^T$ so that 
		$$ \forall u \in V : d'(u) = \sum_{w,\, d(w) > 0} d(w) \cdot d_{w}(u)$$
		We now consider $d'(\TRUE)$ and $d(v)$.
		Once again, we note that $\llbracket x\leftrightarrow v[x \mapsto w] \rrbracket^T = \llbracket w\leftrightarrow v \rrbracket^T$ is always a singleton set, meaning we can simplify this sum in accordance with the semantics of $w\leftrightarrow v$ as follows:
		$$ d'(\TRUE) = \sum_{w,\, d(w) > 0} d_1(w) \cdot d_{w}(\TRUE) =_2 d(v) \cdot d_v(\TRUE) = d(v) $$
		These equations are justified very similarly to the case above, where we note that $d_w(\TRUE) = 0$ for all $w \neq v$, and $d_w(\TRUE) = 1$ for $w = v$ by definition.
		Furthermore, the second part of the conjunction can be proven analogously to $(1)$ as well.
	\end{enumerate} 
\end{proof}

\begin{theorem}[Normalized Boolean Reduction]\label{theo:boolean_reduction_norm}
	Let $e$ be an \NDICE expression, and let $T$ be a function table.
	Then for all values $v \in V$ we have:
	\begin{enumerate}
		\item $\forall d \in \llbracket e \rrbracket_D^T \; \exists d' \in \llbracket \LET{x = e}{ x \leftrightarrow v}\rrbracket_D^T: d(v) = d'(\TRUE) \land \sum\limits_{w \neq v} d(w) = d'(F)$
		\item $\forall d' \in \llbracket \LET{x = e}{ x \leftrightarrow v} \rrbracket_D^T \; \exists d \in \llbracket e \rrbracket_D^T: d(v) = d'(\TRUE) \land \sum\limits_{w \neq v} d(w) = d'(F)$
	\end{enumerate}
\end{theorem}
\begin{proof}
	The theorem follows directly from \Cref{theo:boolean_reduction}, which correlates the unnormalized distributions in the semantics of both expressions.
	Afterwards, the normalizing constants for corresponding distributions can also be shown to be equal, which directly proves the required relation between the normalized semantics of both expressions
\end{proof}

\subsection{Extended Experimental Results}\label{app:extended_results}

\setboolean{useComma}{true}

This section extends the empirical results from \Cref{sec:implementation} with the standard deviation over the reported runtimes, all standard deviations are measured over at least 5 evaluations of the respective benchmark.
We do not report standard deviations for the sizes of the BDDs, ADDs and MDPs again as their numbers are constant between all runs of a benchmark.
Below, \Cref{tab:Ndet_extended} extends \Cref{tab:Ndet}.

\begin{table}[H]
	\caption{Performance of \NDICE on the benchmarks. 
		The \DICE, \NDICE and Storm columns report the total runtime of the respective algorithm, \NDICE is run with state compression and split into \textit{compilation} and \textit{model checking}.
		$\NDICE_u$ reports the total inference time without state compression.
		$\NDICE_u$ reports the total inference time without state compression.
		The reported time is the mean plus or minus a single standard deviation over 5 runs.
		All times are in ms, with a timeout (TO) after 20 minutes.
		}
	\label{tab:Ndet_extended}
	\begin{adjustbox}{max width=\textwidth}
	\begin{tabu}{c | c | ccc | c | c}
		\toprule
		Benchmark & \DICE &  $\mathit{compilation}$ & $\textit{checking}$ & \NDICE & $\NDICE_u$ & Storm \\
		\midrule
		Runway-7 & - & $8\pm 0.4$ & {$\approx 0 \pm 0$} & {$8 \pm 0.4$} & {$8 \pm 0.4$} & {$18 \pm 2.8$} \\
		Runway-15 & - & $93 \pm 2.3$ & {$1 \pm $} & {$94 \pm 2$} & {$94 \pm 2 $} & {$131 \pm 2.7$} \\
		Runway-30 & - & $531 \pm 7$ & {$7 \pm 0.1 $} & {$538 \pm 7.1 $} & {$538 \pm 7.5$ } & {$1\, 620 \pm 14$} \\
		Runway-40 & - & $999 \pm 8$ & {$11 \pm 0.5 $} & {$1\,010 \pm 8.2$} & {$1\,011 \pm  8.8$} & {$5\,212 \pm 61$} \\
		Runway-45 & - & $1\,514 \pm 19$ & {$15 \pm 0.4$} & {$1\,529 \pm 18$} & {$1\,530 \pm 20 $} & {$8\,673 \pm 66$} \\
		
		\midrule
		CouponProb-10 & 
		$721 \pm 11$ &
		$1\,379 \pm 46$ & {$\approx 0 \pm 0$} & {$1\,379 \pm 46$} & 
		{$1\,379 \pm 46$} & $5\,894 \pm 106$ \\
		CouponProb-11 & 
		$1\,958 \pm 33$ & 
		$3\,904 \pm 62$ & {$\approx 0 \pm 0$} & {$3\,904 \pm 62$} & 
		{$3\,904 \pm 62$} & $18\,839 \pm 206$ \\
		CouponProb-12 & 
		$5\,798 \pm 161$ & 
		$10\,484 \pm 216$ & {$\approx 0 \pm 0$} & {$10\,484 \pm 216$} & 
		{$10\,484 \pm 217$} & $56\,434 \pm 1\,115$ \\
		CouponProb-13 & 
		$20\,691 \pm 106$ & 
		$32\,153 \pm 546$ & {$\approx 0 \pm 0$} & {$32\,153 \pm 546$} & 
		{$32\,153 \pm 546$} & $171\,358 \pm 442$ \\
		CouponProb-14 & 
		$63\,418 \pm 280$ & 
		$96\,068 \pm 647$ & {$\approx 0 \pm 0$} & {$96\,068 \pm 647$} & 
		{$96\,068 \pm 647$} & $641\,614 \pm 231\,908$ \\
		
		\midrule
		CouponNdet-6 & 
				- & $105 \pm 1$ &  {$14 \pm 0.7$} & {$119 \pm 1.3$} & {$120 \pm 1.8$} & $162 \pm 7.6$ \\
		CouponNdet-7 & 
				- & $276 \pm 3$ & {$41 \pm 1.8$} & {$317 \pm 2.5$} & {$317 \pm 3.4$} & $607 \pm 25$ \\
		CouponNdet-8 & 
				- & $731 \pm 19$ & {$96 \pm 2$} & {$827 \pm 20$} & {$824 \pm 19$} & $2\,208 \pm 99$ \\
		CouponNdet-9 & 
				- & $2\,264 \pm 49$ & {$389 \pm 33$} & {$2\,653 \pm 66$} & {$2\,662 \pm 64$} & $8\,110 \pm 23$ \\
		CouponNdet-10 & 
				- & $6\,105 \pm 107$ & {$1\,257 \pm 81$} & {$7\,362 \pm 114$} & {$7\,293 \pm 117$} & $25\,862 \pm 388$ \\
		 
		\midrule
		Network-2000 & - & $2\,218 \pm 16$ & {$24 \pm 0.9$} & {$2\,242 \pm 17$} & {$2\,242 \pm 17$} & $25 \pm 0.4$ \\
		Network-4000 & - & $9\,530 \pm 46$ & {$55 \pm 3$} & {$9\,585 \pm 45$} & {$9\,585 \pm 48$} & $39 \pm 0.4$ \\
		Network-6000 & - & $22\,270 \pm 98$ & {$84 \pm 1.3$} & {$22\,354 \pm 98$} &{$22\,358 \pm 102$} & $54 \pm 1.1$ \\
		Network-8000 & - & $40\,763 \pm 98$ & {$122 \pm 6$} & {$40\,885 \pm 97$} & {$40\,885 \pm 98$} & $68 \pm 1.6$ \\
		Network-10000 & - & $67\,081 \pm 207$ & {$153 \pm 1.7$} & {$67\,234 \pm 255$} & {$67\,243 \pm 214$} & $84 \pm 1.4$ \\
		
		\midrule
		Rabin-50 & - & $324\pm 1.1$ & {$\approx 0 \pm 0$} & {$ 324 \pm 1.1$} & {$ 324 \pm  1.1$} & $17 \pm 4.5$ \\
		Rabin-100 & - & $1\,962\pm 49$ & {$1 \pm 0$} & {$ 1\,963 \pm  49$} & {$ 1\,963 \pm  49$} & $23 \pm 0.7$ \\
		Rabin-250 & - & $23\,355\pm 141$ & {$2 \pm 0$} & {$ 23\,357 \pm  141$} & {$ 23\,357 \pm  142$} & $82 \pm 3.5$ \\
		Rabin-500 & - & $233\,494\pm 1\,759$ & {$6 \pm 0.5$} & {$233\,500 \pm 1\,759$} & {$233\,501 \pm 1\,760$} & $327 \pm 6.9$ \\
		
		\midrule 
		Survey & $1 \pm 0.2$ & $1 \pm 0$ & $\approx 0 \pm 0$ & $1 \pm 0$ & $1 \pm 0$ & $14 \pm 2$  \\
		Hepar2 & $99 \pm 1$ & $309 \pm 3$ & $\approx 0 \pm 0$ & $309 \pm 3$ & $310 \pm 3$ & $23\,094 \pm 146$ \\
		Insurance & $580 \pm 4$ & $2\,227 \pm 13$ & $\approx 0 \pm 0$ & $2\,227 \pm 13$ & $2\,232 \pm 14$ & $9\,460 \pm 146$ \\
		Water & $9\,241 \pm 47$ & $9\,288 \pm 37$ & $\approx 0 \pm 0$ & $9\,288 \pm 37$  & $9\,347 \pm 37$   & $12\,545 \pm 46$ \\
		Alarm & $160\,104 \pm 367$ & $161\,385 \pm 533$ & $\approx 0 \pm 0$ & $161\,385 \pm 533$ & $161\,385 \pm 532$ & $5\,335 \pm 9$ \\
		
		\midrule 
		Survey-Ndet & - & $1 \pm 0.1$ & $\approx 0 \pm 0$ & $1 \pm 0.1$ & $1 \pm 0.1$ & $14 \pm 0.5$ \\
		Hepar2-Ndet & - & $280 \pm 2.6$ & $14 \pm 0.4$ & $294 \pm 2.9$ & $295 \pm 3.1$ & $23\,909 \pm 131$ \\
		Insurance-Ndet & - & $2\,217 \pm 12$ & $\approx 0 \pm 0$ & $2\,217 \pm 12$ & $2\,220 \pm 12$ & $10\,345 \pm 137$ \\
		Water-Ndet & - & $9\,296 \pm 97$ & $\approx 0 \pm 0$ & $9\,296 \pm 97$ & $9\,298 \pm 98$ & $13\,153 \pm 60$ \\
		Alarm-Ndet & - & $161\,786 \pm 1\,166$ & $47 \pm 2$ & $161\,633 \pm 1\,167$ & $161\,636 \pm 1\,166$ & $5\,219 \pm 17$ \\
		
		\midrule
		3sat-0 & $19\,511 \pm 72$ & $19\,535 \pm 73$ & {$\approx 0 \pm 0$} & {$19\,535 \pm 73$} & {$19\,535 \pm 73$} & TO \\
		3sat-25 & - & $19\,528 \pm 29$ & {$43 \pm 1$} & {$19\,571 \pm 29$} & {$19\,571 \pm 31$} & TO \\
		3sat-50 & - & $19\,544 \pm 58$ & {$69 \pm 0.4$} & {$19\,613 \pm 29$} & {$19\,612 \pm 30$} & TO \\
		3sat-75 & - & $19\,514 \pm 63$ & {$72 \pm 1.5$} & {$19\,586 \pm 64$} & {$19\,584 \pm 130$} & TO \\
		3sat-100 & - & $19\,567 \pm 128$ & {$107 \pm 3$} & {$19\,674 \pm 128$} & {$19\,676 \pm 132$} & TO \\
		\bottomrule 
	\end{tabu}

\end{adjustbox}

\end{table}

	}{}

\end{document}